\documentclass[11pt]{article}
\usepackage[utf8]{inputenc}
\usepackage[margin=1in,letterpaper]{geometry}

\usepackage{graphicx}
\usepackage{amsmath}
\usepackage{amsthm}
\usepackage{amssymb}
\usepackage{amsfonts}
\usepackage{mathabx}
\usepackage{algorithm}
\usepackage{algorithmic}
\usepackage{paralist}
\usepackage{color}
\usepackage{mathtools}
\usepackage{xspace}
\usepackage{enumitem}
\usepackage[title]{appendix}
\usepackage[compact]{titlesec}

\usepackage{booktabs,multicol,multirow}
\usepackage{xcolor}
\usepackage{nameref}
\definecolor{ForestGreen}{rgb}{0.1333,0.5451,0.1333}
\definecolor{DarkRed}{rgb}{0.5,0.1,0.1}
\definecolor{DarkBlue}{rgb}{0.1,0.1,0.5}
\usepackage[linktocpage=true,
pagebackref=true,colorlinks,
linkcolor=DarkRed,citecolor=ForestGreen,
bookmarks,bookmarksopen,bookmarksnumbered,ocgcolorlinks]{hyperref}

\usepackage{thmtools}
\usepackage{thm-restate} 

\definecolor{mygray}{gray}{0.95}
\usepackage{mdframed}
  {\begin{mdframed}[backgroundcolor=mygray,roundcorner=10pt,leftmargin=20,rightmargin=20,outerlinewidth=2,innertopmargin=0]\begin{result}}
  {\end{result}\end{mdframed}}

\newenvironment{shadetheorem}
  {\begin{mdframed}[backgroundcolor=mygray,roundcorner=10pt,leftmargin=20,rightmargin=20,outerlinewidth=2,innertopmargin=0]\begin{theorem}}
  {\end{theorem}\end{mdframed}}

\definecolor{ultragray}{gray}{0.99}
  {\begin{mdframed}[backgroundcolor=ultragray,roundcorner=10pt,leftmargin=20,rightmargin=20,outerlinewidth=2,innertopmargin=0]\begin{definition}}
  {\end{definition}\end{mdframed}}

\makeatletter
\newcommand*{\bx}{bx}
\newcommand*{\IfBold}{
  \ifx\f@series\bx
    \expandafter\@firstoftwo
  \else
    \expandafter\@secondoftwo
  \fi
}
\makeatother

\DeclareMathOperator{\poly}{poly}
\DeclareMathOperator{\polylog}{polylog}

\declaretheorem{theorem}
\declaretheorem[numberwithin=section]{lemma}

\declaretheorem[numberlike=lemma]{corollary}
\declaretheorem[numberlike=lemma]{claim}

\declaretheorem[numberlike=lemma]{remark}
\declaretheorem{result}

\usepackage[capitalise,nameinlink]{cleveref}
\crefname{theorem}{Theorem}{Theorems}
\crefname{section}{Section}{Sections}
\crefname{lemma}{Lemma}{Lemmas}
\crefname{observation}{Observation}{Observations}
\crefname{algorithm}{Algorithm}{Algorithms}
\crefname{mdalg}{Algorithm}{Algorithms}
\crefname{step}{Step}{Steps}
\crefname{fact}{Fact}{Facts}
\crefname{claim}{Claim}{Claims}
\crefname{part}{Property}{Properties}
\crefname{issue}{Challenge}{Challenges}
\crefname{tech-issue}{Technical Challenge}{Technical Challenges}

\theoremstyle{definition}
\declaretheorem{definition}
\theoremstyle{definition}

\DeclareMathOperator*{\Exp}{\mathbb{E}}

\DeclarePairedDelimiter{\abs}{\lvert}{\rvert}%
\DeclarePairedDelimiter{\card}{\lvert}{\rvert}%
\DeclarePairedDelimiter{\set}{\lbrace}{\rbrace}%
\DeclarePairedDelimiter{\range}{\lbrack}{\rbrack}%
\DeclarePairedDelimiter{\parens}{\lparen}{\rparen}%
\DeclarePairedDelimiter{\floor}{\lfloor}{\rfloor}%
\usepackage{color}
\usepackage{graphicx}
\usepackage{tcolorbox}
\tcbuselibrary{skins,breakable}
\tcbset{enhanced jigsaw}

\newenvironment{tbox}{\begin{tcolorbox}[
		enlarge top by=5pt,
		enlarge bottom by=5pt,
		 boxsep=0pt,
                  left=4pt,
                  right=4pt,
                  top=10pt,
                  arc=0pt,
                  boxrule=1pt,toprule=1pt,
                  colback=white
                  ]
	}
{\end{tcolorbox}}

\theoremstyle{definition}
\newtheorem{mdalg}{Algorithm}
\newenvironment{Algorithm}{\begin{tbox}\begin{mdalg}}{\end{mdalg}\end{tbox}}

\renewcommand{\paragraph}[1]{\medskip\noindent{\bf #1}\xspace}

\newcommand{\Qed}[1]{\ensuremath{\qedsymbol_{\,\,\textnormal{\cref{#1}}}}}

\renewcommand{\qedsymbol}{\nobreak \ifvmode \relax \else
      \ifdim\lastskip<1.5em \hskip-\lastskip
      \hskip1.5em plus0em minus0.5em \fi \nobreak
      \vrule height0.75em width0.5em depth0.25em\fi}

\title{A Distributed Palette Sparsification Theorem}

\author{
Maxime Flin\\
\small Reykjavik University\\
\small \texttt{maximef@ru.is}\and
Mohsen Ghaffari\\
\small MIT\\
\small \texttt{ghaffari@mit.edu}\and
Magn\'us M. Halld\'orsson\\
\small Reykjavik University\\
\small \texttt{mmh@ru.is}\and
Fabian Kuhn\\
\small University of Freiburg\\
\small \texttt{kuhn@cs.uni-freiburg.de}\and
Alexandre Nolin\\
\small CISPA\\
\small \texttt{alexandre.nolin@cispa.de}
}

\date{}

\newcommand{\alg}[2][]{{\IfBold{\MakeUppercase{#2}}{\textup{\textsc{#2}}}}{#1}\xspace}

\newcommand{\multitrial}[1][]{\alg[#1]{SlackColor}}
\newcommand{\slackcolor}[1][]{\alg[#1]{SlackColor}}

\newcommand{\matching}{\alg{Matching}}
\newcommand{\augpath}{\alg{AugPath}}
\newcommand{\growtree}{\alg{GrowTree}}

\newcommand{\harvest}{\alg{Harvest}}
\newcommand{\randompush}{\alg{RandomPush}}

\newcommand{\model}[1]{\ensuremath{\mathsf{#1}}\xspace}
\newcommand{\distream}{\model{LocalStream}}

\newcommand{\NCC}{\model{NCC}}
\newcommand{\CONGEST}{\model{CONGEST}}
\newcommand{\LOCAL}{\model{LOCAL}}
\newcommand{\congest}{\model{CONGEST}}
\newcommand{\local}{\model{LOCAL}}

\newcommand{\ov}[1]{\overline{#1}}

\newcommand{\cA}{\ensuremath{\mathcal{A}}\xspace}

\newcommand{\cG}{\ensuremath{\mathcal{G}}\xspace}
\newcommand{\cK}{\ensuremath{\mathcal{K}}\xspace}

\newcommand{\cS}{\ensuremath{\mathcal{S}}\xspace}
\newcommand{\cX}{\ensuremath{\mathcal{X}}\xspace}
\newcommand{\cW}{\ensuremath{\mathcal{W}}\xspace}

\newcommand{\evt}{\ensuremath{\mathcal{E}}\xspace}

\newcommand{\avail}{\ensuremath{\mathsf{avail}}}

\newcommand{\hC}{\widehat{C}}

\newcommand{\cC}{\widecheck{C}}

\newcommand{\sparse}[1]{\ensuremath{\widetilde{#1}}}
\newcommand{\Gsparse}{\ensuremath{\widetilde{G}}}
\newcommand{\Vsparse}{\ensuremath{V_{\mathsf{sparse}}}}

\newcommand{\avganti}{\ensuremath{\bar{d}}}

\newcommand{\smin}{s_{\min}}

\newcommand{\emax}{e_{\max}}
\newcommand{\Next}{N_{\mathrm{ext}}}

\newcommand{\K}{\ensuremath{2\alpha}}
\renewcommand{\epsilon}{\varepsilon}

\newcommand{\hatd}{\ensuremath{\widehat{d}}}
\newcommand{\hS}{S'}
\newcommand{\eqdef}{\stackrel{\text{\tiny\rm def}}{=}}

\newcommand{\Lg}{\ensuremath{L^{\mathsf{G}}}}
\newcommand{\Lh}{\ensuremath{L^{\mathsf{H}}}}

\begin{document}
\maketitle

\begin{abstract}
The celebrated palette sparsification result of [Assadi, Chen, and Khanna SODA'19] shows that to compute a $\Delta+1$ coloring of the graph, where $\Delta$ denotes the maximum degree, it suffices if each node limits its color choice to $O(\log n)$ independently sampled colors in $\{1, 2, \dots, \Delta+1\}$. They showed that it is possible to color the resulting sparsified graph---the spanning subgraph with edges between neighbors that sampled a common color, which are only $\tilde{O}(n)$ edges---and obtain a $\Delta+1$ coloring for the original graph. However, to compute the actual coloring, that information must be gathered at a single location for centralized processing. We seek instead a local algorithm to compute such a coloring in the sparsified graph. The question is if this can be achieved in $\poly(\log n)$ distributed rounds with small messages.

Our main result is an algorithm that computes a $\Delta+1$-coloring after palette sparsification  with $O(\log^2 n)$ random colors per node and runs in $O(\log^2 \Delta + \log^3 \log n)$ rounds on the sparsified graph, using $O(\log n)$-bit messages.
We show that this is close to the best possible: any distributed  $\Delta+1$-coloring algorithm that runs in the \LOCAL model on the sparsified graph, given by palette sparsification, for any $\poly(\log n)$ colors per node, requires $\Omega(\log \Delta / \log\log n)$ rounds. This distributed palette sparsification result leads to the first $\poly(\log n)$-round algorithms for $\Delta+1$-coloring in two previously studied distributed models: the Node Capacitated Clique, and the cluster graph model.
\end{abstract}

\pagenumbering{roman}
\thispagestyle{empty}
\newpage

\thispagestyle{empty}
\tableofcontents
\newpage

\pagenumbering{arabic}

\section{Introduction}
The  \emph{Palette Sparsification Theorem} of Assadi, Chen, and Khanna (ACK, henceforth) \cite{ACK19} is a beautiful and powerful sparsification result for the $\Delta+1$-coloring problem: the problem of assigning a color $c(v)\in\set{1, \ldots, \Delta+1}$ to each node $v\in V$ of an $n$-node graph $G=(V,E)$ such that adjacent nodes $u,v\in V$, for which $uv\in E$, receive different colors. Here, $\Delta$ is the maximum degree of the graph. ACK show that we can $\Delta+1$-color any graph $G$, by list-coloring a \emph{sparse} sub-graph $\sparse{G}$, which has only $\tilde{O}(n)$ edges. Their theorem led to several breakthroughs for sublinear algorithms, including graph streaming algorithms, sublinear query algorithms, and massively parallel computation algorithms.

More precisely, the theorem states that for any graph $G$, if we independently sample random a list $L(v)$ of $O(\log n)$ colors for each vertex $v\in V$, with high probability, the graph $G$ is $L$-list-colorable. That is, there exists a coloring of $G$ where each $v$ is assigned a color $c(v)\in L(v)$. To compute a $\Delta+1$-coloring of $G$, one then computes an $L$-list-coloring of the sub-graph $\sparse{G}$ retaining only edges $uv\in E$ where $L(u)\cap L(v)\neq\emptyset$. A simple argument shows that $\sparse{G}$ is sparse and has maximum degree $O(\log^2 n)$, thereby giving the aforementioned sub-linear algorithms. 

The ACK result gives rise to the hope that there might be an ultimately scalable (distributed) solution for the $\Delta+1$ coloring, where each graph node needs to interact and coordinate with only $\poly(\log n)$ of its neighbors. However, all known applications of the palette sparsification theorem require gathering the sparsified subgraph $\sparse{G}$ in one location, and solving the resulting list-coloring problem in a centralized fashion. This is prohibitively expensive in distributed models with restrictive communication, e.g., if each node can send/receive only $\poly\log n$ bits per round.

In this paper, we remedy this problem by giving a nearly-optimal \emph{distributed} version of the palette sparsification theorem. Informally, we show that there is a fast distributed algorithm for coloring the sparsified subgraph, and using communications only on the sparsified graph (modulo a small relaxation in the graph's degree, compared to ACK). This leads to the first poly-logarithmic randomized algorithms in constrained settings studied in the distributed literature \cite{rozhovn2022undirected, ghaffari2015flow,ghaffari2016distributed,ghaffari2013cut,ghaffari2022universalCut,AGGHSKL19}.

\subsection{Background and State of the Art}
\paragraph{Distributed Coloring.}
The $\Delta+1$-coloring problem has been one of the central problems in the study of distributed graph algorithms \cite{PanconesiS97,johansson99,SW10,fraigniaud16,BEPSv3,HSS18,CLP20,GGR20,HKMT21,GK21,HKNT22}. In fact, this was the main problem studied by Linial in his celebrated paper introducing the \local model \cite{linial92}.
In this model, we have a communication network between processors, abstracted as an undirected graph, and this is also the graph for which we want to compute a vertex coloring. Each vertex is equipped with a $O(\log n)$-bit unique identifier (where $n=|V|$) and communicates in synchronous rounds with its neighbors. The variant of this model with $O(\log n)$-bit messages is known as the \congest model \cite{peleg00}.

In recent years, there has been exciting progress on sublogarithmic time randomized algorithms \cite{BEPSv3,HSS18,CLP20,GK21,HKNT22,HNT22,ghaffari2023fasterMIS} culminating in state-of-the-art complexities of $O(\log^3 \log n)$ in \congest and $\widetilde{O}(\log^2 \log n)$ in \local. 
In fact, when $\Delta \ge \Omega(\log^4 n)$ --- which is the interesting range for \cite{ACK19} --- the best round complexity known is $O(\log^* n)$ \cite{HKNT22,HNT22}. 

In constrained distributed models such as \textit{cluster graphs} and \textit{the node congested clique} (see \cref{sec:intro-corollaries}), where nodes can effectively send/receive only $\poly\log n$ bits per rounds (or more generally, $\poly\log n$ bit aggregate summaries of the messages), no $\poly\log n$ algorithm is known. A major impediment is this: all known algorithms work by computing the coloring gradually, and in the intermediate steps, nodes need to learn which colors are already used by their neighbors. This forces communications that need $\Omega(\Delta)$ bits. 

The palette sparsification theorem of \cite{ACK19} reduces the problem of $\Delta+1$-coloring $G$ to a list-coloring problem on a graph $\sparse{G}$ with $O(\log^2 n)$ maximum degree. Hence, it seemingly opens the road for ultimately scalable distributed algorithms, where each node sends/receives only $\poly(\log n)$ bits. However, that hinges on whether $\sparse{G}$ can be colored fast \emph{distributively}. 
Unfortunately, the proof of \cite{ACK19} is intrinsically centralized (for reasons explained in \cref{sec:ack}). All applications of \cite{ACK19} use centralization to compute the $\Delta+1$-coloring. The research question at the core of our paper is to investigate the discrepancy between the locality of the $\Delta+1$-coloring problem and the locality of the induced list-coloring problem on the sparsified graph.

\begin{quote}
\begin{center}
\emph{Can the sparsified graph be colored \textbf{locally}?}
\end{center}
\end{quote}

\subsection{Our Results}
Our answer is two-fold. We design an algorithm for $\Delta+1$-coloring such that the color of a vertex $v$ depends only on its $O(\log^2\Delta)$-hop neighborhood in \sparse{G} (when $\Delta \ge \Omega(\log^4 n)$), and we concretely give efficient distributed algorithms with small messages to compute such a coloring. Conversely, we show that no algorithm can achieve a locality smaller than $\widetilde{\Omega}(\log\Delta)$. We next state the results in a more formal manner. 

We present a \congest algorithm to list-color the sparsified graph in $O(\log^2\Delta)$ rounds when $\Delta \ge \Omega(\log^4 n)$. When $\Delta \le O(\log^4 n)$, the input graph $G$ is sparse already and can be colored by the $O(\log^3\log n)$-round state-of-the-art \congest algorithm \cite{HKNT22,HNT22,GK21}.

\vspace{.5em}
\begin{mdframed}[backgroundcolor=mygray,roundcorner=10pt,leftmargin=20,rightmargin=20,outerlinewidth=2,innertopmargin=0]
  \begin{restatable}{theorem}{DistPaletteSparsificationThm}[Distributed Palette Sparsification Theorem]
  \label{thm:DPS}
   Suppose that each node in a graph $G$ samples $\Theta(\log^2 n)$ colors u.a.r.\ from $[\Delta+1]$. 
   There is a distributed message-passing algorithm operating on the sparsified graph,
   that computes a valid list-coloring in $O(\log^2 \Delta + \log^3\log n)$ rounds, using $O(\log n)$-bit messages.
 In particular, each node needs to communicate with only $O(\log^4 n)$ different neighbors.
 \end{restatable}
  \end{mdframed}
\vspace{.5em}

\paragraph{Our Techniques in a Nutshell.} We shall give an overview of our algorithm in \cref{sec:tech-intro}. For now, we merely mention three aspects in which our algorithm differs significantly from both streaming and distributed algorithms.
\begin{enumerate}
\item Contrary to \cite{ACK19}, we cannot afford to color dense clusters sequentially. In particular, when we color a cluster, we cannot assume that colors on the outside are adversarial. We give an algorithm that functions as long as conflicting colors with the outside are only a small fraction of the color space.
To ensure that this property holds, we \emph{precondition} clusters. That is, we reduce the number of connections between clusters beforehand, so that, later in the algorithm, random decisions on the outside only harm a small enough fraction of nodes on the inside. This preconditioning might be useful in other applications of palette sparsification.
\item We introduce a new technique of \emph{augmenting trees} to distributively color dense clusters of the graph. It consists of $O(\log\Delta)$ steps for growing trees rooted at uncolored nodes such that if a leaf can recolor itself, we can color the root. We show that a constant fraction of the uncolored nodes are colored by this process, resulting in the $O(\log^2\Delta)$ runtime.
\item To reach the nearly-optimal $O(\log^2\Delta)$ runtime, we need this process to succeed with high probability \emph{even when $o(\log n)$ nodes remain uncolored}. We overcome this issue by locally amplifying probabilities and using that only few nodes remain to resolve contentious efficiently.
\end{enumerate}

\paragraph{Lower Bound.}
We give evidence that this round complexity is in the right ballpark, by showing that $\Omega(\log \Delta / \log\log n)$ rounds are needed to compute a valid coloring after sparsification of neighborhoods by uniformly random palette sparsification.

\begin{shadetheorem}
\label{thm:intro-lowerbound}
Any \local algorithm that operates on the sparsified graph and computes a $(\Delta+1)$-coloring with at least a constant probability of success needs $\Omega\parens*{\frac{\log\Delta}{\log\log n}}$ rounds.
This holds even if the original graph is a $(\Delta+1)$-clique, even if the distributed algorithm running on the sparsified graph uses unbounded messages, and even if each node samples a large $\poly\log n$ number of colors in the sparsification.
\end{shadetheorem}

\subsection{Corollaries for Other Models} 
\label{sec:intro-corollaries}

\paragraph{Distributed Streaming.}
The semi-streaming model -- where we skim through the (very large) set of edges of the graph and store only $\poly\log n$ bit of memory per node before solving the problem -- has been studied extensively \cite{AMS96, BJKST02, CCF02, CM04, AGM12, ACK19}. A frequent technique in this setting is \emph{(distributed) sketching} \cite{AGM12,KLMMS14,GMT15,ACK19}: nodes locally compress their neighborhoods to $\poly\log n$-bit sketches before combining all of them centrally. 
We find it helpful to think of our algorithm in the following similar setting: first, nodes look at their edges in one streaming pass, using only $\poly\log n$ memory; nodes then communicate in a distributed fashion using only edges they stored locally. We emphasize, however, that it is crucial for our applications that \emph{nodes communicate only with $\polylog n$ nodes per round}. Contrary to the semi-streaming model, it does not suffice to reduce the problem to $\widetilde{O}(n)$ edges: each neighborhood must contain at most $\poly\log n$ edges. See \cref{sec:local-stream} for a more precise definition.

\paragraph{Coloring Cluster Graphs.} A natural situation that arises frequently in distributed graph algorithms is that of \textit{cluster graphs}, as we explain next (this appears under various names, see e.g., \cite{rozhovn2022undirected, ghaffari2015flow,ghaffari2016distributed,ghaffari2013cut,ghaffari2022universalCut}). Suppose that in the course of some algorithm, the nodes have been partitioned into vertex-disjoint (low-diameter) clusters. The corresponding \textit{cluster graph} is an abstract graph with one node for each cluster, where two clusters are adjacent if they contain neighboring nodes.
Note that this corresponds to (graph-theoretically) \textit{contracting} each cluster, an operation that is easy centrally but has no meaningful distributed counterpart. In distributed settings with such cluster graphs, we often need to solve certain graph problems on this cluster graph to facilitate other computations. Distributed computation on the cluster graph assumes we have a low-depth cluster tree that spans each cluster and can be used for broadcast and convergecast in the cluster. One round of communication on the cluster graph involves: (1) broadcasting a $\poly(\log n)$-bit message from the cluster center to all its nodes; (2) passing information on the edges between neighboring clusters; (3) convergecasting any $\poly(\log n)$-bit aggregate function from the cluster nodes to the center. 

Prior to our work, it remained open whether one can compute a $\Delta+1$-coloring in $\poly\log n$ rounds of communication on the cluster graph. Here, $\Delta$ denotes the maximum number of clusters that are adjacent to a cluster. The more traditional approaches to $\Delta+1$-coloring (e.g., \cite{linial92, johansson99, BEPSv3}) fall short of this $\poly\log n$ round complexity goal as they usually need to learn the colors remaining available to one cluster, after some partial coloring of other clusters, and that may require gathering $\Delta$ bits at the cluster center. Our distributed palette sparsification theorem resolves this and gives the first efficient distributed $\Delta+1$-coloring on cluster graphs, as we state informally next. See \Cref{sec:clusterGraphs} for definitions and the actual result.

\paragraph{Coloring in the Node Capacitated Clique.}  Another immediate consequence of our work is the first $\poly\log n$-round Node Capacitated Clique algorithm for $\Delta+1$-coloring. This model was introduced by \cite{AGGHSKL19} to capture peer-to-peer systems in which nodes have access to global communication in the network while being restrained to $O(\log n)$-bit messages to $O(\log n)$ nodes within one communication round.
To identify the edges to use in the sparsified graph, we assume the nodes have access to shared randomness; alternatively, as we show in \cref{sec:NCC}, this can be replaced by an existential construction\footnote{We believe that our algorithm can be implemented using $\poly\log n$-wise independent random bits to sample each of the $\poly\log n$ colors in the lists. It would be at the cost of a higher $\poly\log n$ round complexity and possibly larger $\poly\log n$ number of colors in lists. As this is not a major contribution and requires a significant amount of extra technicalities, we reserve this for future work.}.

\subsection{Related Work and Problems}

The groundbreaking palette sparsification theorem of Assadi, Chen, and Khanna \cite{ACK19} showed that $\Delta+1$-coloring was possible in the semi-streaming model, even in dynamic streams.
The sparsification property was fundamental, as it allowed also for optimal algorithms in two, seemingly unrelated models: a sublinear-time algorithm in the query model, and a two-round algorithm in the Massively Parallel Computation model with $\tilde{O}(n)$ memory per machine.
The theorem was extended to several more constrained coloring problems in \cite{AA20}, such as $O(\Delta/\log \Delta)$-coloring triangle-free graphs, and to $\deg+1$-list coloring in \cite{HKNT22}.
It was also a crucial ingredient in the recent semi-streaming algorithm for $\Delta$-coloring \cite{AKM22}.

Palette sparsification is a form of a sampling technique that holds in the restrictive \emph{distributed sketching model}.
The latter corresponds to multi-party communication with shared blackboard model and vertex partitioned inputs, as well as to the \emph{broadcast congested clique} (though the congested clique term usually refers to the case that the message size is $O(\log n)$). 
Starting with the seminal work of \cite{AGM12}, many graph problems have been solved with distributed sketching.
Though, notably, the problems of maximal independent set and maximal matching---which are closely related to $\Delta+1$-coloring---have been shown to require much larger space \cite{AKZ22}, even if allowed multiple rounds of writing to the shared blackboard.

Coloring plays a central role in distributed algorithms as a natural approach to breaking symmetry and scheduling access to exclusive resources.
In particular, the original work of Linial \cite{linial92} introducing local algorithms and the \local model was specifically about the $\Delta+1$-coloring problem.
Since then, there has been a lot of work on local coloring algorithms,
both randomized \cite{johansson99,BEPSv3,SW10,HSS18,CLP20,HKMT21,HKNT22} and deterministic (e.g., \cite{barenboim15,BEG18,MT20,GK21,ghaffari2023fasterMIS}).

The Node-Capacitated Clique model was introduced in \cite{AGGHSKL19} to model distributed systems built on top of virtual overlay networks. They gave algorithms for the maximal independent set problem and $O(a)$-coloring, with time linear in $a$, where $a$ is the arboricity of the graph. The question of efficient $\Delta+1$-coloring has remained open.

Many models of distributed computing, both in theory and in nature, are much more restrictive than \local or \congest, both in terms of communication abilities and processing power: e.g., beeping model \cite{CK10}, wireless (ad-hoc, unstructured...), programmable matter \cite{Amoebot14}, networked finite-state machines \cite{EW13}. These often capture distributed features of the natural and physical world. A logical direction is therefore to identify models that strongly limit the power of the nodes, yet allow for fast distributed computation. Few features are as fundamental as limiting the amount of space available.

A recent work by \cite{FGHKN23} uses some (of the earlier) subroutines from our work to $\Delta+1$-color graphs in the broadcast congest model of distributed computing. In that model, per round, each node must send the same $O(\log n)$-bit message to all of its neighbors. They adapt \cref{alg:ACD,alg:matching-compact} to compute an almost-clique decomposition and a colorful matching in $O(1)$ rounds. However, we emphasize that the core technical challenges and contributions in the two works are different.

\section{Technical Introduction}
\label{sec:tech-intro}

In this section, we outline the techniques we use to prove \cref{thm:DPS,thm:intro-lowerbound}.
Our algorithm builds on existing literature, both streaming \cite{ACK19,AA20,AKM22} and distributed \cite{SW10,EPS15}. It differs nonetheless from both in key aspects. On the one hand, streaming algorithms \cite{ACK19,AA20} require a global view of the sparsified graph -- which we avoid by computing the coloring in a distributed fashion. On the other hand, existing distributed algorithms \cite{BEPSv3,HSS18,CLP20,HKNT22} require that nodes communicate with all of their neighbors -- which we avoid since nodes communicate only on the sparsified graph.

In \Cref{sec:ack}, we review the palette sparsification theorem of \cite{ACK19} and explain why it does not extend to our setting. Then, in \Cref{subsec:ourApproach}, we outline the technical novelties in our algorithm. Finally, in \Cref{sec:loweroverview}, we describe an overview of our lower bound result.

\subsection{Comparison with Palette Sparsification}
\label{sec:ack}

The proof of \cite{ACK19} relies on a variant of the \emph{sparse-dense decomposition} introduced by \cite{Reed98}. Variants of this decomposition were used in earlier distributed coloring algorithms~\cite{hsinhao_coloring, CLP18}. This decomposition partitions the graph into a set of ``\emph{locally sparse}'' vertices and a collection of ``\emph{almost-cliques}''. For intuition, consider a (somewhat degenerate) example of a sparse node: a vertex of degree $\Delta/2$. If we try to color this vertex with a color from $[\Delta+1]$ uniformly at random, it has probability $1/2$ to succeed (no matter what colors its neighbors choose). Repeating this process iteratively $O(\log n)$ times is enough to color all such vertices with high probability. It was observed by \cite{EPS15} that this reasoning extends to sparse vertices (with a more involved analysis). It is worth noting that such randomized color trials are easily implemented in \LOCAL (and \congest) and form a core component of fast randomized distributed coloring algorithms \cite{johansson99,BEPSv3,CLP20}. Indeed, the part of the algorithm of \cite{ACK19} for coloring sparse vertices is also already distributed. \cref{thm:DPS} improves the $O(\log n)$-round algorithm of \cite{ACK19} to $O(\log\Delta)$ by using the faster algorithm of \cite{SW10} to color sparse nodes, but we do not have any significant technical novelty in that part.

In \cite{ACK19}, most of the effort (and novelty) goes into the handling of almost-cliques. They iterate over almost-cliques \emph{sequentially} and color each one assuming the coloring on the outside is \emph{adversarial}. Clearly, in our setting, we cannot afford to process almost-cliques one by one. Furthermore, to achieve the $O(\log^2\Delta)$ runtime claimed by \cref{thm:DPS}, for reasons expounded in \cref{sec:intro-augtree}, we cannot assume the outside colors to be adversarial. When we color each almost-clique, we must carefully resolve contentions with the outside, including other almost-cliques that are getting colored in parallel.

To color a fixed almost-clique $C$, \cite{ACK19} looks for a perfect matching in the bipartite graph with vertices of $C$ on the one side, colors in $[\Delta+1]$ on the other and an edge between $v\in C$ and $c\in[\Delta+1]$ if $c$ is in $L(v)$ and not used by an outside neighbor. If $|C| \le \Delta+1$, classic results from random graph theory (e.g., \cite[Section VII.3]{bollobas1998random}) show that, with high probability, a perfect matching exists, and therefore a list-coloring is possible. While the mere existence of the matching is enough for \cite{ACK19}, \cref{thm:DPS} provides a \emph{distributed algorithm} to compute it. Note that learning the full topology of $C$ in $O(\log\Delta)$ rounds of \local to then decide on a matching does not work because it does not account for conflicts with concurrent almost-cliques. Furthermore, our algorithm uses $O(\log n)$-bit messages, which prohibit centralization approaches. Our main technical contribution is the design of an $O(\log^2\Delta)$-round algorithm using $O(\log n)$-bit messages to compute this matching in almost-cliques in parallel (see \cref{sec:augmenting-path}), while also managing outside conflicts. 

In large almost-cliques $|C|\ge \Delta+1$, such a perfect matching cannot exist. To deal with those, \cite{ACK19} shows the existence of a \emph{colorful matching}. Namely, they color pairs of anti-neighbors in $C$ (pair of nodes that are not connected by an edge) using the same color so that the number of uncolored nodes decreases twice as fast as the number of free colors. We give a fast distributed version of the sequential algorithm of \cite{ACK19} (see \cref{sec:intro-colorful-matching}). While this is not the main technical contribution of our work, we think this procedure itself might find applications in future distributed coloring algorithms.

Another noteworthy challenge for us is regarding \emph{probability amplification}. Contrary to the centralized setting, we cannot afford algorithms with a constant probability of success. Indeed, amplifying success probability with $O(\log n)$ independent repetitions would exceed our $O(\log^2\Delta)$ runtime. We deal with this issue by increasing the number of colors sampled, compared to \cite{ACK19}, such that we always have concentration on large enough quantities, i.e. at least $\Omega(\log n)$. Naturally, trying more colors creates more conflicts, which must be resolved.

\subsection{Our Approach and New Ideas}
\label{subsec:ourApproach}

\subsubsection{Step 1: Preconditioning of Almost-Cliques}
\label{sec:intro-preconditionning}

Like \cite{ACK19}, our algorithm heavily relies on an $\epsilon$-almost-clique decomposition. More formally, it decomposes the graph into a set $\Vsparse$ of locally sparse nodes and a collection $C_1, \ldots, C_t$ of almost-cliques. An $\epsilon$-almost-clique is a cluster $|C|\le (1+\epsilon)\Delta$ such that each node $v\in C$ has $|N(v)\cap C|\ge (1-\epsilon)\Delta$ neighbors in $C$ (see \cref{def:ACD}). The preconditioning step strengthens the properties of our almost-cliques. More precisely, it computes a partial coloring such that all uncolored nodes are clustered in almost-cliques and have $O(\Delta/\log n)$ connections to uncolored nodes in other almost-cliques, compared to the usual $\epsilon\Delta$ (see \cref{thm:strong-ACD} for a formal definition). This property is key to ensure, later in our algorithm, that random decisions outside of a cluster cannot seriously impede its progress on the inside (see \cref{lem:conditionning-outside,lem:avail}). We now give further details on that aspect. Readers that are not familiar with palette sparsification results may skip the remainder of this subsubsection on the first reading.

The key property of cliques that our algorithm uses is that when $k$ nodes are uncolored, there are $k$ ``available colors'' that are used by no one in the clique. The colorful matching (\cref{sec:intro-colorful-matching}) allows us to extend this to almost-cliques. However, we still need to ensure that if a node (re)colors itself with one of these $k$ available colors, it will not create a conflict with an external neighbor. 
For the sake of concreteness, assume only one node is left to color in almost-clique $C$, i.e., $k=1$. In our algorithm, a constant fraction of $C$ samples $K\log n$ colors for some large enough $K > 0$. This way, the probability that at least one node in this almost-clique finds that one available color is at least $1-(1-1/\Delta)^{K\Delta\log n}\ge 1-1/\poly(n)$.
Having nodes sample more colors has a drawback: it increases the competition for colors.
If a node $v\in C$ has $\epsilon\Delta$ external neighbors in active almost-cliques, the probability that at least one of them blocks the one color that $v$ is looking for is $1-(1-1/\Delta)^{\epsilon K\Delta\log n}\ge 1-1/\poly(n)$ (as $\epsilon K \in\Theta(1)$).
During preconditioning, we ensure that nodes in active almost-cliques have at most $\Delta/(K\log n)$ external neighbors in other active almost-cliques. The probability that an external neighbor blocks the one color that $v$ is looking for becomes $1-(1-1/\Delta)^{\frac{\Delta K\log n}{K\log n}} \approx 1-1/e$. Therefore, only a small fraction of $C$ is affected by the randomness outside of $C$. This argument is made formal in \cref{lem:conditionning-outside}.

To precondition almost-cliques, we use an idea from distributed coloring algorithms \cite{SW10,EPS15,HKMT21}. Namely, nodes that have $\Omega(\Delta/\log n)$ connections to nodes in other almost-cliques are $\Omega(\Delta/\log n)$ sparse. By coloring nodes in a carefully chosen order, we get \cref{thm:strong-ACD}.
As it only uses well-known techniques from distributed coloring, we defer the analysis to \cref{sec:strong-ACD}. While the preconditioning algorithm in itself is not a major contribution of our work, we believe it could find further use in the (distributed) coloring literature.

\subsubsection{Step 2: Distributed Colorful Matching}
\label{sec:intro-colorful-matching}

A colorful matching is defined as a matching in the complement graph of the almost-clique such that endpoints of a matched edge are colored the same (\cref{def:colorful-matching}). 
Intuitively, this reduces the size of the clique: if one merges matched nodes, one reduces the number of nodes in the almost-clique while maintaining a proper coloring of the original graph. In an almost-clique of $(1+\epsilon)\Delta$ nodes, finding $\epsilon\Delta$ such pairs essentially reduces the coloring problem to that of coloring a clique.
This technique was introduced by \cite{ACK19} in the first palette sparsification theorem to deal with that exact issue and we claim no novelty in its use. In \cite{ACK19}, however, only the existence of a large enough matching is proven, whereas we also need to compute it efficiently in a distributed setting.

We define $\avganti_C$ as the average anti-degree in $C$: such that $\avganti_C|C|/2$ is the number of anti-edges. When nodes try a random color in $[\Delta+1]$, an anti-edge is monochromatic with probability $1/(\Delta+1)$. Therefore, the expected number of monochromatic edges is $\delta\avganti_C$ for some small constant $\delta > 0$, because a node can retain its color with constant probability. Using tools similar to \cite{EPS15}, one can show this random variable is concentrated near its mean with probability $1-e^{-\Omega(\avganti_C)}$. It implies that in cliques with $\avganti_C \ge \Omega(\log n)$, for any constant $K=O(1/\epsilon)$, we can accumulate $K\avganti_C$ anti-edges in the colorful matching in $O(K/\delta)$ rounds (\cref{lem:colorful-matching-porous}).

We use a different approach when $\avganti_C \le O(\log n)$. Note that when $\avganti_C$ is smaller than some constant, nodes have hardly any anti-edges. 
Hence we can also assume $\avganti_C \ge \Omega(1)$; meaning the previous algorithm succeeds with constant probability.
Instead of trying a single color, nodes try $\Theta(\log n)$ colors at the same time.
Clearly, a large enough colorful matching exists: using the sampled colors, the previous process can be implemented in $O(\log n)$ rounds.
To find that matching efficiently (in $O(\log \Delta)$ rounds), we capitalize on the fact that there are few non-edges in the clique (\cref{lem:colorful-matching-porous}). 
Since the probability of a non-edge having both endpoints sample a common color is $\Theta(\log^2 n/\Delta)$, only $O(\avganti_C\log^2 n)=O(\log^3 n)$ potential monochromatic non-edges are sampled. By taking advantage of the high expansion property of the sparsified clique, we can disseminate the list of $O(\log^3 n)$ sampled monochromatic edges in $O(\log\Delta)$ rounds to all nodes in the clique, which can then compute the colorful matching locally. Because $\Omega(\Delta)$ colors are available in the clique, the concurrent coloring of external neighbors can block at most a small fraction of the colors (\cref{lem:avail}).

\subsubsection{Step 3: Augmenting Trees}
\label{sec:intro-augtree}

To color almost-cliques, we take advantage of the fast expansion of the sparsified almost-clique to find many \emph{augmenting paths}. Our definition of augmenting path corresponds precisely to the one for computing maximal matching in the random bipartite graphs induced by the random lists of colors~\cite{hopcroft1973n, motwani1994average}.
We emphasize, however, that general-purpose algorithms for maximal matching do not directly apply in our setting because of conflicts between concurrent almost-cliques.
Furthermore, computing an \emph{exact} maximum matching is a global problem, and in fact requires at least $\Omega(\sqrt{n})$ rounds of \congest in general, even in low-diameter graphs~\cite{ahmadi2018distributed}. 

We now explain how we color all almost-cliques in $O(\log^2\Delta)$ rounds. 
The algorithm runs $O(\log\Delta)$ iterations of the following process. Suppose $k$ is number of uncolored nodes in $C$ at the current iteration.
We say that color $c$ is \emph{available} to a node $v$ if $c$ is neither used in $C$ nor by external neighbors of $v$ that were colored during the preconditioning step (these nodes will never change colors).  
In each iteration, we grow a forest of \emph{augmenting paths} (\cref{def:augmenting-path}). An augmenting path is a path $u_0,u_1,\ldots u_i$ in the \emph{sparsified} almost-clique such that 1) $u_0$ is the only uncolored node, 2) each $u_{j-1}$ for $j\in[i]$ can (re)color itself with the color of $u_{j}$ and 3) the last node $u_i$ of the path knows an available color $c$.
Provided with such a path, we can recolor $u_i$ with $c$ and each $u_{j-1}$ with the color of $u_j$, thereby coloring the uncolored endpoint $u_0$. 
Our algorithm builds on the following idea: if we have a path $u_0, \ldots, u_i$ verifying 1) and 2) but not 3), then $u_i$ samples a uniform color $c\in[\Delta+1]$ and finds an available one with probability $\Omega(k/\Delta)$. 

The two prior steps of our algorithm are key to ensure $u_i$ has probability $\Omega(k/\Delta)$ to find an available color \emph{that it can adopt}.
The colorful matching ensures (almost) all nodes of $C$ will have $k$ colors available (\cref{lem:clique-palette}). On the other hand, the argument sketched in \cref{sec:intro-preconditionning} shows that because of the preconditioning step, with high probability over the randomness outside of $C$, at least $\Omega(\Delta)$ nodes in $C$ can adopt $k/2$ of the colors available to them. We say of nodes that cannot adopt $k/2$ available colors that they are spoiled (\cref{def:spoiled-nodes}). Since they represent a small fraction of $C$ and the path explores the almost-clique randomly, we are unlikely to fail due to spoiled nodes (\cref{lem:unspoiled-leaves}).

This simple algorithm colors $u_0$ with probability $\Omega(k/\Delta)$. 
To color each uncolored node with constant probability, even when $k \ll \Delta$, we grow $\Delta/(\alpha k)$ paths verifying 1) and 2) from each uncolored nodes for some large enough constant $\alpha > 1$. The expected number of paths to find an available color is $\Delta/(\alpha k)\cdot \Omega(k/\Delta) = \Omega(1/\alpha)$. Therefore, we color $\Omega(k/\alpha)$ nodes in expectation. Since we must avoid collisions between paths from different uncolored nodes, we find it helpful to further restrict paths to grow trees. See \cref{fig:high-level-aug-path} for a high-level description of one iteration.

\begin{figure}[ht!]
    \centering
\begin{mdframed}
\begin{center}
    \underline{\alg{Augmenting Path Algorithm}}
\end{center}
\vspace{1em}

Let $k$ be the number of uncolored nodes.

\paragraph{Growing the forest.} The uncolored nodes are the roots of the forest. 

Repeat $O(\log\frac{\Delta}{\alpha k})$ times:
\begin{enumerate}
    \item Each leaf samples a set $S_u$ of $O(\log n)$ colors.
    \item Remove from $S_u$ the colors used by external neighbors, nodes in the forest, or sampled by other leaves. 
    \item For each $c\in S_u$ find $v_{c,u}\in C$ colored $c$ and connect them to $v$ in the forest.
\end{enumerate}

\paragraph{Harvesting the trees.} 
\begin{enumerate}
    \item If $k \ge \Omega(\log n)$, each leaf tries one random color in $[\Delta+1]$. If a leaf can retain its color, we recolor the path connecting it to the root, root included.
    \item If $k \le O(\log n)$, each leaf tries $O(\log n)$ random colors in $[\Delta+1]$. The roots learn $\Theta(k)$ colors with which leaves in their trees can recolor themselves. They disseminate this list in $O(\log\Delta)$ rounds to all nodes of $C$. Nodes then compute a color-leaf matching and recolor the corresponding paths.
\end{enumerate}
\end{mdframed}
\caption{High level description of one iteration.}
    \label{fig:high-level-aug-path}
\end{figure}

An iteration has two phases: the \emph{growing phase} (\cref{alg:growtree}), where we grow the trees, and the \emph{harvesting phase} (\cref{alg:harvest-high,alg:harvest-low}), where we try to recolor augmenting paths. The growing phase needs $O(\log\frac{\Delta}{\alpha k})$ rounds because each time we increase the number of paths by a constant factor. 
The harvesting phase differs depending on $k$. That is because when $k\le O(\log n)$, we cannot show progress \emph{with high probability} with a simple concentration on $k$.
See \cref{sec:augmenting-path-overview} for a more detailed description.

\subsection{Lower Bound}
\label{sec:loweroverview}

We complement our upper bound with a lower bound on the distributed complexity of coloring a graph after random palette sparsification. The lower bound applies even if the graph prior to the palette sparsification was simply a complete graph. Concretely, the lower bound states the following. Assume that we have a complete graph $K_n$ on $n$ nodes $V=\set{1,\dots,n}$ that we want to color with $n=\Delta+1$ colors. Every node $v\in V$ samples a random subset $S_v$ of colors $C=\set{1,\dots,n}$ as follows. For each node $v\in V$ and each color $x\in C$, $x$ is included in $S_v$ independently with probability $p=f(n)/n$, where $f(n)\geq c\ln n$ for a sufficiently large constant $c$ and $f(n)\leq \polylog n$. Recall that the sparsified graph is the graph induced by all edges of $K_n$ between nodes $u,v\in V$ with $S_u\cap S_v\neq \emptyset$. We prove that any distributed message passing algorithm on the sparsified graph requires $\Omega(\log n / \log\log n)$ rounds to properly color the $K_n$ in such a way that each node $v$ is colored with a color from its sample $S_v$. This holds even if the message sizes are not restricted.

\paragraph{Relation to perfect matching on random bipartite graphs.} Note that this is equivalent to the following bipartite matching problem. Define a bipartite graph $B=(V\cup C, E_B)$, where $C=\set{1,\dots,n}$ represents the set of colors. There is an edge between nodes $v\in V$ and $x\in C$ whenever $x\in S_v$. A \emph{valid coloring} of the nodes in $V$ then corresponds to a \emph{perfect matching} in the bipartite graph $B$. Note that if $p\geq c\ln n / n$ for a sufficiently large constant $c>0$, then the bipartite graph $B$ has a perfect matching with high probability. This (in even sharper versions) is well known in the random graph literature~(e.g., \cite[Section VII.3]{bollobas1998random}) and can be proven by checking Hall's condition for any non-empty subset of $V$. Our lower bound essentially shows that distributedly computing a perfect matching in the random graph $B$ requires $\Omega(\log n / \log\log n)$ rounds with at least constant probability, even in the \LOCAL model (i.e., even if the nodes in $B$ can exchange arbitrarily large messages). Note that the sparsified subgraph of $K_n$ and the bipartite graph $B$ can simulate each other with only constant overhead in the distributed setting. Any $T$-round algorithm on the sparsified graph can be run in $O(T)$ rounds on $B$ and any $T$-round algorithm on $B$ can be run in $O(T)$ rounds in the sparsified subgraph of $K_n$ (in the second case, each color node $x\in C$ can be simulated by one of the nodes $v\in V$ for which $x\in S_v$).

\paragraph{Lower bound on computing a perfect matching in random bipartite graphs.} In general, it is not too surprising that computing a perfect matching of a graph is a global problem where nodes at different ends of the graph need to coordinate. Consider for example the problem of computing a perfect matching of a $2n$-node cycle. There are exactly two such perfect matchings and deciding which of the two matchings to choose cannot be decided without global coordination within the cycle. However, the case of a random bipartite graph needs much more care.

Our lower bound is based on the following observation regarding perfect matchings in bipartite graphs. Let $v_0$ be some node of a bipartite graph $H$ and for each $d\geq 0$, let $V_d$ be the set of nodes of $H$ that are at hop distance exactly $d$ from $v_0$. Since $H$ is bipartite, a node in a set $V_d$ can only be connected to nodes in sets $V_{d-1}$ and $V_{d+1}$. Clearly, $|V_0|=1$ and, because $v_0$ must be matched, there must be exactly one matching edge between nodes in $V_0$ and nodes in $V_1$. Further, since every other node of $V_1$ must be matched to nodes in $V_2$, there must be exactly $|V_1|-1=|V_1|-|V_0|$ matching edges between nodes in $V_1$ and nodes in $V_2$. With a similar argument, the number of matching edges between nodes in $V_2$ and nodes in $V_3$ is exactly $|V_2|-|V_1|+|V_0|$. By extending this argument, one can see that for every $d$, the number of matching edges between $V_d$ and $V_{d+1}$ depends on the sizes of all the sets $V_0,\dots,V_d$. Changing the size of a single one of those sets also changes the number of matching edges between $V_d$ and $V_{d+1}$.

For the lower bound proof, we now proceed as follows. Assume that there is a $T$-round distributed algorithm that computes a perfect matching of the random bipartite graph $B$. We consider some node $v_0$ in the random bipartite graph $B$ and two integers $\ell$ and $h$ such that $\ell>0$ and $h-\ell>T$. We consider the decisions of the assumed distributed perfect matching algorithm for nodes in $V_h$. Note that in $T$ rounds, nodes in $V_h$ do not see nodes at distance more than $T$, and in particular, they do not see nodes in $V_\ell$. However, by the above observation, the number of matching edges between nodes in $V_h$ and nodes in $V_{h+1}$ depends on the knowledge of $|V_\ell|$. If $T$ is sufficiently small and a large fraction of the graph is outside the $T$-hop neighborhoods of nodes in $V_h$, then even collectively, the nodes in $V_h$ have significant uncertainty about the value of $|V_\ell|$. Therefore, they cannot determine the number of matching edges between $V_h$ and $V_{h+1}$ (and thus their matching edges) with reasonable probability. The actual proof that formalizes this intuition is somewhat tedious. The details appear in \Cref{sec:lower}.

\section{Preliminaries}

\paragraph{Notation.} For any integer $k \ge 1$, we write $[k]$ for the set $\set{1, 2, \ldots, k}$. For a tuple $X=(X_1, X_2, \ldots, X_k)$ and some $i\in [k]$, we define $X_{\le i}=(X_1, \ldots, X_i)$. For a graph $G=(V,E)$ and any set $S\subset V$, we denote by $N_G(S)=\set{v\in V\setminus S: \exists u\in S\text{ and } uv\in E}$ the neighborhood of $S$. Moreover, for any sets $S_1, S_2\subseteq V$, let $E_G(S_1, S_2)=E\cap (S_1\times S_2)$ be the set of edges between nodes of $S_1$ and $S_2$. 

A \emph{partial} $(\Delta+1)$-coloring is a function from the nodes $V$ to $[\Delta+1]\cup\set{\bot}$ that assigns colors in $[\Delta+1]$ or no colors (in the form of the null color $\bot$) to vertices, such that adjacent vertices have different colors in $[\Delta+1]$ (but they may both have color $\bot$). For some partial $(\Delta+1)$-coloring of the graph and any set $S\subseteq V$, we denote by $\widehat{S}$ the uncolored nodes of $S$ and by $\widecheck{S}$ the colored ones.

When we say that an event happens ``with high probability", we mean it occurs with probability $1-1/\poly(n)$ for a suitably large polynomial in $n$ to union bound over polynomially many such events.

\subsection{Distributed Coloring}

A standard technique in distributed coloring used by randomized algorithms is to have each node repeatedly try a color picked uniformly at random in its palette: the set of colors not already used by its neighbors. It was introduced by \cite{johansson99} and is used in all efficient distributed algorithms \cite{BEPSv3, CLP20, HKNT22}. We introduce this notion here for further use.

\begin{definition}[Palette]
\label{def:palette}
The \emph{palette} $\Psi_v$ of a node $v$ with respect to some coloring of the nodes is the set of colors that are not used by its neighbors.
\end{definition}

\begin{definition}[Slack]
The slack of $v$ is the difference between the size of its palette and its uncolored degree.
\end{definition}

If nodes have slack proportional to their degree, they can be colored in $O(\log^* n)$ rounds of \local by trying random colors. 
The following result has origins in \cite{SW10} and was generalized by \cite{CLP20}. 
It is straightforward to see that the proof of \cite{HKNT22} extends to our setting.

\begin{lemma}[{Lemma 1 in \cite{HKNT22}}]
\label{lem:multi-trial}
Consider the $(\deg+1)$-list coloring problem where each node $v$ has slack $s(v)=\Omega(d(v))$.
Let $1<\smin \leq \min_v s(v)$ be globally known.
For every $\kappa\in (1/\smin,1]$, there is a randomized \LOCAL algorithm {\slackcolor[$(\smin)$]} that in $O(\log^* \smin+1/\kappa)$ rounds properly colors each node $v$ w.p.\ $1 - \exp(-\Omega(\smin^{1/(1+\kappa)})) - \Delta e^{-\Omega(\smin)}$, even conditioned on arbitrary random choices of nodes at distance $\geq 2$ from~$v$.
Using $O(\log n)$-bit messages, it requires $O(\log\Delta)$ rounds of communication, 
and requires nodes to sample up to $\Theta\parens{\frac{s_v\log n}{\Delta}}$ colors in $[\Delta+1]$.
\end{lemma}

\subsection{Sparse-Dense Decomposition}
We use a decomposition of the graph into \emph{locally sparse} nodes, which have many non-edges from in their neighborhood, and dense clusters called \emph{almost-cliques} (also informally called \emph{cliques}). Almost-clique decomposition was first introduced in graph theory by \cite{Reed98} and has been used extensively in streaming \cite{ACK19,AW22} and distributed \cite{HSS18,CLP20,HKNT22} algorithms.

\begin{definition}[Sparsity]
\label{def:sparsity}
The \emph{sparsity} of a node $v$ is the value $\zeta_v=\frac{1}{\Delta}\parens*{\binom{\Delta}{2} - |E(N(v))|}$. We say a node is \emph{$\zeta$-sparse} if $\zeta_v \ge \zeta$, otherwise it is \emph{$\zeta$-dense}.
\end{definition}

\begin{definition}[Almost-Clique Decomposition]
\label{def:ACD}
For $\epsilon\in (0,1/3)$, a $\epsilon$-almost-clique decomposition is a partitioning of the vertices into sets $\Vsparse, C_1, \ldots, C_k$ for some $k$ such that:
\begin{enumerate}
    \item All $v\in\Vsparse$ are $\Omega(\epsilon^2\Delta)$-sparse.
    \item For any $i\in[k]$, almost-clique $C_i$ has the following properties:
    \begin{enumerate}
    \item $|C_i|\le (1+\epsilon)\Delta$;
    \item $|N(v)\cap C|\ge (1-\epsilon)\Delta$ for all nodes $v\in C_i$.
    \end{enumerate}
\end{enumerate}

For a dense node $v$ in some almost-clique $C$, we call its \emph{external degree} $e_v$ the number of neighbors $v$ has  outside of its almost-clique, i.e.\ $e_v=|N(v)\setminus C|$, and its \emph{anti-degree} $a_v$ the number of non-neighbors in $C$, i.e. $a_v=|C\setminus N(v)|$. We denote by $\avganti_C=\sum_{v\in C} a_v/|C|$ the average anti-degree of $C$.
\end{definition}

\section{Palette Sampling and The Sparsified Graph}
\label{sec:streaming}

\paragraph{Parameters.} We assume that $\Delta \ge \Omega(\log^4 n)$ as otherwise nodes can store all their adjacent edges and simply run the \congest algorithm of \cite{HKMT21}. We define the following parameters\footnote{which we have not attempted to optimize.} for our algorithm:
\begin{equation}
\label{eq:parameters}
\begin{array}{llll}
  \alpha &\eqdef 500:\quad&\text{quantify the number of leaves that an augmenting tree must have,}\\
  \beta &\eqdef \floor{C\log n}:\quad &\text{used to bound the number of sampled colors,}\\
  \epsilon &\eqdef 10^{-8}:\quad &\text{a small enough constant.}\\
\end{array}
\end{equation}

The constant $C$ in $\beta$ is sufficiently large for high probability events to hold, even when we union bound over polynomially many events. It is independent of the constants $\alpha$ and $\epsilon$. We use the following relation between our parameters:
\begin{equation}
\label{eq:eps-alpha}
\epsilon \le 1/\alpha^2\qquad\text{ and }\qquad 2\alpha < 1/(18\epsilon) \ .
\end{equation}

\subsection{Palette Sampling}

Similar to \cite{ACK19}, we see the lists of random colors $L(v)$ in our algorithm as a source of \emph{fresh random colors}. Whenever a node samples a color in $[\Delta+1]$, it reveals a new color from its list. To simplify the analysis, we partition $L(v)$ into sub-lists, each used for a different purpose in the algorithm. The main difference with \cite{ACK19} is that lists are larger: $O(\log^2 n)$ colors instead of $O(\log n)$.

\begin{Algorithm}
\label{alg:palette-sampling}
Palette Sampling for a node $v$.

\begin{itemize}
\item $L_1(v)$: sample $O(\log^2 n)$ colors independently and uniformly at random in $[\Delta+1]$. 
\item $L_2(v)= \set{L_{2,i}(v), \text{for $i\in [O(1/\epsilon)]$}} \cup \set{L_2^*(v)}$: for each $i\in[O(1/\epsilon)]$, sample each color in $L_{2,i}(v)$ independently with probability $\frac{1}{4\Delta}$. Sample each color independently in $L_2^*(v)$ with probability $\gamma\cdot\frac{\beta}{\Delta}$ for some constant $\gamma=\Theta(1/\epsilon^2)$ defined in \cref{lem:ack-colorful-matching}.
\item $L_3(v)=\set{L_{3,i}(v)=\Lg_{3,i}(v)\cup \Lh_{3,i}(v), \text{ for $i\in[\beta]$}}$ where we sample each color $c\in[\Delta+1]$ in $\Lg_{3,i}$ and $\Lh_{3,i}$ independently with probability $\frac{20\beta^2}{\Delta}$.
\end{itemize}
\end{Algorithm}

By union bound, a fixed color $c\in[\Delta+1]$ is included in $L_2(v)$ with probability $O\parens{\frac{\beta}{\epsilon^2\Delta}}$ and in $L_3(v)$ with probability $O(\frac{\beta^2}{\Delta})$. Since each color is included in $L_2$ independently, a simple Chernoff bound proves the following claim:

\begin{claim}
\label{lem:lists-size}
With high probability, $|L_2(v)|\le O(\log n/\epsilon^2)$ and $|L_3(v)|\le O(\log^2n)$ for all $v\in V$. Furthermore, each $\Lg_{3,i}(v)$ and $\Lh_{3,i}(v)$ contains at least $6\beta$ different colors.
\end{claim}

Our algorithm will color each $v\in V$ with a color from its list $L(v)$. Following the observation of \cite{ACK19}, edges between nodes with non-intersecting lists can be dropped. A standard argument shows the induced subgraph is sparse with high probability.

\begin{definition}[The Sparsified Graph]
\label{def:conflict-graph}
For a graph $G=(V,E)$ and lists $L(v)$ such as described in \cref{alg:palette-sampling}, let $\Gsparse=(V,\widetilde{E})$ be the subgraph of $G$ with edges $uv\in E$ such that $L(u)\cap L(v)\neq\emptyset$. We call $\Gsparse$ the \emph{sparsified graph}. For any set $S\subseteq V$, we denote by $\sparse{S}$ the induced subgraph $\Gsparse[S]$.
\end{definition}

\begin{claim}[{\cite[Lemma 4.1]{ACK19}}]
\label{lem:sparse-graph}
For any graph $G$, w.h.p., the sparsified graph $\sparse{G}$ has maximum degree $O(\log^4 n)$.
\end{claim}

\subsection{Decomposition and Properties}
\label{sec:decomposition}

Similarly to other distributed algorithms \cite{HSS18,CLP20,HKNT22} our algorithm needs to know the almost-clique decomposition in order to compute the coloring. However, existing \congest algorithms require communication with $\gg\poly\log n$ nodes \cite{HKMT21,HNT22}. Building on \cite{HNT22}, we give a \congest algorithm where nodes need only to communicate on a sparse subgraph of $G$. \cref{alg:high-level-ACD} gives an overview of the communication needed.

We emphasize that \cref{alg:high-level-ACD} uses a sparse subgraph of $G$ that is independent of the sparsified subgraph induced by the random lists (\cref{def:conflict-graph}). We found it simpler to state our algorithm this way. Observe, however, that \cref{alg:high-level-ACD} could be implemented using edges sampled from random lists. We briefly explain why. Two nodes $u$ and $v$ adjacent in the sparsified graph share at least one color $c$. To know if they share a large fraction of their neighborhood --- i.e.,  if they are friends (\cref{def:friend-edges}) --- notice that the number of nodes in $N(u)\cap N(v)$ that sample $c$ is concentrated, and therefore provides an unbiased estimator for the size of this set. Using a bandwidth compression technique introduced by \cite{HNT22}, $u$ and $v$ can compare their neighborhoods in $O(1)$ rounds using $O(\log n)$ bandwidth. To know if $v$ has many friends --- i.e., if it is popular (\cref{def:popular}) --- notice that its neighboring edges are sampled independently in the sparsified graph. Therefore, a node will detect a lot of friendly edges in the sparsified graph if and only if it is sufficiently popular. We prefer the following less technical and more general algorithm that does not depend on the sparsified graph and could be of independent interest.

\begin{Algorithm}
\label{alg:high-level-ACD}
High-level algorithm computing $\epsilon$-almost-clique decomposition.

\textbf{Input:} the sparsified graph $\sparse{G}=(V,\sparse{E})$.

\textbf{Output:} an $\epsilon$-almost-clique decomposition of $G$.

\begin{enumerate}
\item Each node $v$ samples a value $r(v)\in [\Theta(\Delta/\epsilon)]$ using public randomness.

\item Each node $v$ with degree at least $\Delta/2$ computes
\begin{itemize}[leftmargin=*]
\item the set $F(v)=\set{r(u): u\in N(v), r(u) \le \sigma}$ where $\sigma=\Theta(\log n/\epsilon^4)$.
\item a set $E_s(v)$ of $O(\log n/\epsilon^2)$ random edges.
\end{itemize}

\item (Communication Phase) After \emph{$O(1/\epsilon^4)$ rounds of communication using edges $E_s(v)$} and $O(\log\Delta)$ rounds of communication on $\sparse{G}$, each $v$ knows if it is sparse or dense as well as the unique identifier of its cluster if it is dense.
\end{enumerate}
\end{Algorithm}

\begin{remark}
Some important remarks about \cref{alg:high-level-ACD}.
\begin{enumerate}
\item  
Note that $F(v)$ depends on the randomness of the entire neighborhood of $v$. This is not an issue for any of our applications as it can easily be computed on a stream with $O(\log n/\epsilon^4)$ local memory, with public randomness in the Node Congested Clique and with aggregation of a single $O(\log n/\epsilon^4)$-bitmap in cluster graphs.
\item Nodes sample edges that might not belong to $\sparse{G}$. Nonetheless, we assume they can communicate along these edges. To remove this assumption, one could encode the sampling of $E_s(v)$ within the palette sampling process. Observe that adding $E_s(v)$ does not affect the sparsity of $\sparse{G}$\footnote{This is the reason we sample $E_s(v)$ only for high-degree nodes; nodes of degree less than $\Delta/2$ will be sparse anyway.}. We phrase it this way for simplicity.
\end{enumerate}
\end{remark}

\begin{restatable}{lemma}{ACD}
\label{lem:acd}
There is a $O(\log\Delta)$-round algorithm computing an $\epsilon$-almost-clique decomposition. It only broadcasts $O(\log n/\epsilon^4)$-bit messages and samples $O(\log n/\epsilon^2)$ edges per node.
\end{restatable}

Since \cref{lem:acd} is a rather straightforward extension of \cite{HNT22}, we defer its proof to the appendix (see \cref{sec:ACD}).

\paragraph{Useful properties of the sparsified clique.}
Let $C$ be an $\epsilon$-almost-clique. The sparsified clique $\sparse{C}$ is a random graph on (almost) $\Delta$ nodes where edges are sampled with probability $O(\log^4 n/\Delta)$. As such, the graph $\sparse{C}$ has typical properties of random graphs.

\begin{lemma}[Expansion]
\label{lem:expansion}
Let $C$ be an almost-clique, and assume $\Delta \ge \beta^4$.
With high probability, for all subsets $S\subseteq C$ of size at most $|S|\le 3\Delta/4$,
\[ |E_{\sparse{C}}(S,C\setminus S)|\ge |S|\beta^4/40 \qquad\text{and}\qquad |N_{\sparse{C}}(S)\cap (C\setminus S)| \ge \Omega(|S|). \]
\end{lemma}

\begin{proof}
  We show that for any fixed set $S \subseteq C$, the edge expansion property ($|E_{\sparse{C}}(S,C\setminus S)|\ge |S|\beta^4/40$) holds w.p.\ $1-\exp(-\Omega(\card{S} \beta))$. Since there are at most $\card{C}^x = \exp(O(x \log \Delta))$ subsets $S \subseteq C$ of size $x$, this allows us to claim by union bound that, w.h.p., the edge expansion property holds for all subsets $S \subseteq C$. The vertex expansion property then follows easily.
 
    Let us consider a fixed subset $S \subseteq C$, and set $\ov{S} \eqdef C\setminus S$ and $L \eqdef \max\parens*{1,\frac{\Delta}{\card{S}\beta}}$.
    We partition the colors into $(\Delta+1) / L$ groups of size $L$. Let $B_i \subseteq [\Delta+1]$ be the colors of bucket number $i$. We introduce several random variables.
    \begin{itemize}
        \item For each color $c \in [\Delta +1]$ and $e=uv \in E_G[S\times \ov{S}]$, let $X_{c,e}$ be the indicator random variable for whether $e$'s $S$-endpoint node $u$ sampled color $c$. Let $X_{c,u} = \sum_{v \in \ov{S}:uv \in E[S \times \ov{S}]} X_{c,uv}$ and $X_{c} = \sum_{u \in S} X_{c,u}$.
        \item For each color $c \in [\Delta +1]$ and $e=uv \in E_G[S\times \ov{S}]$, let $Y_{c,e}$ be the indicator random variable for whether both of $e$'s endpoint nodes $u$ and $v$ sampled color $c$. Note that $X_{c,e} = 0 \Rightarrow Y_{c,e} = 0$. Let $Y_{c,v} = \sum_{u \in S :uv \in E[S \times \ov{S}]} Y_{c,uv}$ and $Y_{c} = \sum_{v \in \ov{S}} Y_{c,v} = \sum_{e \in E[S \times \ov{S}]} Y_{c,e}$.
        \item For each $i \in [\Delta / L]$, let $Z_i = \sum_{c\in B_i} \sum_{e \in E[S \times \ov{S}]} (X_{c,e} \cdot Y_{c,e})$ be the contribution of bucket number $i$ to the number of edges between $S$ and $\ov{S}$.
    \end{itemize}
    Note that the random variables $Z_i$ are defined with multiplicity, i.e., possibly counting each edge multiple times. We will fix that soon. 
    
    A node $v\in S$ has at least $|N_G(v)\cap\ov{S}|\ge|\ov{S}|-\epsilon\Delta \ge (1/4-\epsilon)\Delta \ge \Delta/5$ neighbors in $\ov{S}$ before sparsification -- and at most $\Delta$. In total, $E_G[S \times \ov{S}]$ contains between $\card{S}\Delta/5$ and $\card{S}\Delta$ edges.
    We have
    \[ \Exp[X_{c,e}] = \beta^2 / (\Delta+1), \quad \Exp[X_c] \in [\beta^2 \card{S}/5,\beta^2 \card{S}],\quad \text{and} \quad \Exp[Z_i] \in [L\beta^2 \card{S}/5,L\beta^2 \card{S}].\]

    Let us now define auxiliary random variables $X'_{c,e}$, $Y'_{c,e}$, $Z'_i$, and related quantities in a manner that avoids the issue of overcounting. We reveal the colors in increasing order, and for each color $c$, let $\cG_c$ be the event that smaller colors already sampled $\card{S}\beta^4/20$ distinct edges into the sparsified graph. Note that $\cG_c$ is fully determined by the randomness of earlier colors. Let $X'_{c,e} = X_{c,e}$ and $Y'_{c,e} = Y_{c,e}$ when $\cG_c$ holds. When $\cG_c$ does not hold, let $X'_{c,e}$ and $Y'_{c,e}$ always equal $0$ if both endpoints of $e$ sampled an earlier color, and otherwise let $X'_{c,e} = X_{c,e}$ and $Y'_{c,e} = Y_{c,e}$ as before. $Z'_i = \sum_{c\in B_i} \sum_{e \in E[S \times \ov{S}]} (X'_{c,e} \cdot Y'_{c,e})$.
    Since $\card{S}\beta^4/10 \le \card{S}\Delta/10$, there are always at least $\card{S}\Delta/10$ edges for which $\Exp[X'_{c,e}] > 0$ and $\Exp[Y'_{c,e}] > 0$.
    We have
    \[\Exp[X'_c] \in [\beta^2 \card{S}/10,\beta^2 \card{S}],\quad \text{and} \quad\Exp[Z_i] \in [L\beta^2 \card{S}/10,L\beta^2 \card{S}].\]

    Consider a bucket $B_i$. We now make all random decisions regarding whether each node in $S$ samples each color $c \in B_i$. Consider the summation $\sum_{c \in B_i} \frac {1}{\Delta} X'_{c}$.
    It has an expected value of at least  $L\cdot \beta^2 \card{S}/(10\Delta)$, and is a sum of random variables $\frac 1 \Delta X'_{c}$ whose sampling spaces are contained in the range $[0,1]$. Hence, $\sum_{c \in B_i}X'_{c} \geq L\cdot \beta^2 \card{S}/20$ w.p.\ at least $1-\exp(-\Omega(L\cdot \beta^2 \card{S}/\Delta)) \geq 1-\exp(-\Omega(\beta))$ by \cref{lem:chernoff} (Chernoff bound), i.e., w.h.p. Furthermore, for each $v \in \ov{S}$ and color $c \in [\Delta+1]$, at most $O(\beta^2)$ of its neighbors in $S$ sample $c$, w.h.p. 

    Let us now make all random decisions regarding whether nodes in $\ov{S}$ sample each color in $B_i$. Let us analyze the summation
    \[Z'_i = \sum_{c\in B_i} \sum_{e \in E[S \times \ov{S}]} (X'_{c,e} \cdot Y'_{c,e}) = \sum_{c\in B_i} \sum_{v \in \ov{S}} \parens*{\sum_{u \in S: uv \in E[S \times \ov{S}]} (X'_{c,uv} \cdot Y'_{c,uv})}.\]
    $Z'_i$ has an expected value of at least $L\cdot \beta^4 \card{S}/(10\Delta)$.
    As random choices where fixed in $S$ s.t.\ each vertex $v \in \ov{S}$ has at most $O(\beta^2)$ of its neighbors in $S$ sample each color $c$, the inner term $\parens*{\sum_{u \in S: uv \in E[S \times \ov{S}]} (X'_{c,uv} \cdot Y'_{c,uv})}$ is distributed in $[0,O(\beta^2)]$, i.e., 
    the decision of any given $v \in \ov{S}$ for each color impacts the sum by at most $O(\beta^2)$. We can thus divide the sum by $O(\beta^2)$ in order to get random variables distributed in $[0,1]$ and apply \cref{lem:chernoff}. This gives that $Z'_i \geq L\cdot \beta^4 \card{S}/(20\Delta)$ w.p.\ at least $1-\exp(-\Omega((L\cdot \beta^4 \card{S}/\Delta)/\beta^2)) \geq 1-\exp(-\Omega(\beta))$.

    Hence, each bucket contributes at least $L\cdot \beta^4 \card{S}/(20\Delta)$ to the overall sum $Z'=\sum{i \in [\Delta / L]} Z'_i$, w.h.p., regardless of the choices of previous buckets. Let us analyze the probability that $\Delta/(2L)$ or more buckets fail to have this good contribution. Consider a specific set of $\Delta/(2L)$ buckets. The probability that they all fail to have a good contribution is $\exp(-\Omega(\beta\Delta/L))$. There are less than $2^{\Delta/L}$ ways to choose $\Delta/(2L)$ buckets out of $\Delta/L$, so by union bound the probability that less than $\Delta/(2L)$ buckets have a good contribution is at most $2^{\Delta/L} \cdot \exp(-\Omega(\beta\Delta/L)) = \exp(-\Omega(\beta\Delta/(2L)))$. Finally, $\exp(-\Omega(\beta\Delta/L)) = \exp(-\Omega(\min(\beta\Delta,\card{S}\beta^2))) < \exp(-\Omega(\card{S}\beta))$, so the sparsified graph contains at least $\beta^4 \card{S}/40$ edges w.p.\ at least $1-\exp(-\Omega(\card{S}\beta))$.
    
    This small failure probability $\exp(-\Omega(\card{S}\beta))$ allows us to union bound over all choices of $S$, using $\beta$ a large enough $\Omega(\log n)$. Hence, w.h.p., the edge expansion property holds. The vertex expansion follows from \cref{lem:sparse-graph}. Since the maximum degree in $\sparse{G}$ is $O(\beta^4)$, the number of nodes in $\ov{S}$ is at least $\frac{|E_{\sparse{C}}(S,\ov{S})|}{O(\beta^4)} = \Omega(|S|)$. 
\end{proof}

\Cref{lem:expansion} implies the following result as any two nodes can reach more than half of the clique in $O(\log\Delta)$ hops.

\begin{corollary}
\label{cor:spar-clique-low-diam}
The sparsified almost-clique $\sparse{C}$ has diameter $O(\log\Delta)$. 
\end{corollary}

Also, observe that two nodes from the same almost-clique that sampled the same color must be within distance 2 in the sparsified graph. 

\begin{claim}
\label{claim:color-group}
For a clique $C$, let $u$ and $v$ be two nodes of $C$ and $c\in[\Delta+1]$ be an arbitrary color. Then, with high probability, there exist at least $2\beta^2/5$ nodes $w\in N_G(u)\cap N_G(v)\cap C$ that sample $c\in L(w)$. In particular, if $c\in L(u)\cap L(v)$ then $u$ and $v$ are at two hops from each other in the sparsified graph $\sparse{G}$.
\end{claim}

\begin{proof}
Let $q\eqdef\frac{\beta^2}{\Delta}$. Nodes $u$ and $v$ share at least $(1-2\epsilon)\Delta$ neighbors in $C$ and each $w\in N_G(u)\cap N_G(v)\cap C$ sampled the color $c$ with probability (at least) $q$. In expectation, at least $(1-2\epsilon)\Delta \cdot q \ge (4/5)\beta^2$ such $w$ sampled $c$. Since each $w$ samples its color independently, the classic Chernoff bound applies. At least $(2/5)\beta^2$ shared neighbors sampled $c$ with probability $1-\exp(-\Omega(\beta^2)) \gg 1-1/\poly(n)$. 
\end{proof}

\section{The Distributed Palette Sparsification Theorem}
\label{sec:color-logdelta}

In this section, we give the \congest algorithm for our main theorem. Again, we assume $\Delta \ge \Omega(\log^4 n)$ and show a runtime of $O(\log^2\Delta)$.

\DistPaletteSparsificationThm*

\subsection*{Step 1: Preconditioning Almost-Cliques}
When we compute the colorful matching or build augmenting trees, nodes might sample $\Theta(\log n)$ random colors within a round. If a node has $\Omega(\Delta)$ neighbors, this might result in all colors being blocked by its external neighbors. To circumvent this issue, we use standard techniques from distributed coloring to strengthen guarantees given by the almost-clique decomposition (\cref{def:ACD}).

\begin{restatable}{theorem}{strongACD}
\label{thm:strong-ACD}
Let $\epsilon \in (0,1/3)$ be \underline{a constant independent of $n$ and $\Delta$}, and $\eta$ be any number (possibly depending on $n$ and $\Delta$) such that $\Delta/\eta \ge K\log n$ for a large enough constant $K > 0$. There exists an algorithm computing a partial coloring of $G$ where all uncolored nodes are partitioned in almost-cliques $C_1, \ldots, C_t$ for some $t$ such that for any $i\in[t]$, almost-clique $C_i$ is such that: 
\begin{enumerate}
    \item\label[part]{part:SACD-upp} $|C_i| \le (1+\epsilon)\Delta$;
    \item\label[part]{part:SACD-inside} $|N(v)\cap C_i| \ge (1-\epsilon)\Delta$ for all nodes $v\in C_i$;
    \item\label[part]{part:SACD-ext} $|N(v)\cap \bigcup_{j\neq i} C_j| \le \emax=\Delta/\eta$
\end{enumerate}

Furthermore, the algorithm runs in $O(\log\Delta + \log \eta)$ rounds, uses $O(\eta\log n)$ colors from lists $L_1$, and samples $O(\eta\log n)$ edges per node.
\end{restatable}

Note that the bound on the external degree given by \cref{part:SACD-ext} is much stronger than the one from the classical almost-clique decomposition.
Henceforth, we assume we are given the coloring and decomposition of \cref{thm:strong-ACD} with maximum external degree
\begin{equation}
\label{eq:emax}
\emax\eqdef \Delta/\eta \quad\text{where}\quad \eta\eqdef\max\parens{160\alpha\beta, L_{\max}/\epsilon} \ ,
\end{equation}
where $L_{\max}=O(\log n)$ is the upper bound on the size of lists $L_2$ of \cref{lem:lists-size}. As this uses only standard techniques from distributed coloring, we sketch the algorithm here and defer the complete proof to \cref{sec:strong-ACD}. 

\begin{proof}[Proof Sketch]
Compute an $(\epsilon/3)$-almost-clique decomposition in $O(\log\Delta)$ rounds (by \cref{lem:acd}). We use the fact that nodes with external degree $\Delta/\eta$ are $\Omega(\Delta/\eta)$-sparse; hence receive permanent slack from randomized color trials (\cref{lem:slackgen}).

To ensure \cref{part:SACD-ext}, we divide cliques of the almost-clique decomposition into two categories: \emph{introvert} cliques, with at most $(2\epsilon/3)\Delta$ nodes of high external degree; and \emph{extrovert} cliques, where more than $(2\epsilon/3)\Delta$ nodes have high external degree. We begin by generating slack in $\Vsparse$ and extroverted cliques. We next color nodes of low-external-degree in extroverted cliques using the $(\epsilon/3)\Delta$ temporary slack provided by their inactive neighbors of high external degrees. We color sparse nodes and high-external-degree nodes in introverted cliques using their permanent slack. By \cref{lem:multi-trial}, this takes $O(\log\Delta)$ rounds. We finish by coloring the high-external-degree nodes in extroverted cliques in two steps: first $O(\log\log n)$ rounds of randomized color trials to reduce their degree to $O(\Delta/\log n)$, then $O(\log\Delta)$ rounds using their permanent slack. What remains uncolored are then only the low-external-degree nodes in introverted cliques.

To detect nodes of high external degree, we use random edge samples. As we sample $O(\eta\log n)=O(\log^2 n)$ edges per node, it is not an issue for our applications: nodes communicate to only $O(\log^2 n)$ nodes during this step.
\end{proof}

\subsection*{Step 2: Colorful Matching}
The remaining uncolored nodes are very dense (more than $\Delta/\eta$-dense). We find a large matching of anti-neighbors in each clique and color (the endpoints of) each such node-pair with a different color. Such matchings are called \emph{colorful} and were introduced by \cite{ACK19} in the original palette sparsification theorem.

\begin{definition}[Colorful Matching]
\label{def:colorful-matching}
For any partial coloring of the nodes,
a matching $M$ in $\ov{C}$ (anti-edges in $C$) is \emph{colorful} if and only if the endpoints of each edge in $M$ are colored the same.
\end{definition}

If a clique has a colorful matching, the set of colors that are not used in the clique approximates well the palette of nodes \emph{with small anti-degree}. We call this set of colors \emph{the clique palette}.

\begin{definition}[Clique Palette]
\label{def:clique-palette}
For an almost-clique $C$, 
the \emph{clique palette} $\Psi_C$ (w.r.t.\ a valid coloring of the vertices) is the set of colors not used by nodes of $C$.
\end{definition}

The following lemma formalizes the idea that $\Psi_C$ and $\Psi_v$ are similar (for most nodes $v$).
Recall that $\hC$ and $\cC$ denote respectively the set of uncolored and colored nodes of $C$.

\begin{lemma}
\label{lem:clique-palette}
Let $C$ be an almost-clique with a colorful matching $M$ and fix any partial coloring of the nodes. We say $v\in C$ is \emph{promising} (with respect to $M$) if it satisfies $a_v\le |M|$, and otherwise it is \emph{unpromising}.
For each promising node $v$, we have \[ |\Psi_C\cap \Psi_v| \ge |\hC|. \]
\end{lemma}

\begin{proof}
Let $C$, $M$ and $v$ be as described above. We have $|\Psi_C|\ge \Delta-|\cC|+|M|$ because $\Psi_C$ loses at most one color per colored node but saves one for each color used by the colorful matching.
Since nodes are either colored or uncolored, i.e. $|C|=|\hC|+|\cC|$, we can lower bound the number of colors in the clique palette by $|\Psi_C|\ge |\hC|+\Delta-|C|+|M|$. On the other hand, observe that $\Delta \ge |N(v)\cap C| + e_v$ and $|C| = |N(v)\cap C|+a_v$. Hence, we have $|\Psi_C|\ge |\hC| + |M| + e_v - a_v \ge |\hC|+e_v$, using that $v$ is promising. The lemma follows as $|\Psi_C\cap\Psi_v| \ge |\Psi_C| - e_v \ge |\hC|$.
\end{proof}

In our setting, the existence of such a matching is not enough; we must compute it in few rounds. In \cref{sec:colorful-matching}, we describe how to compute a colorful matching of $\Theta(K\cdot\avganti_C)$ anti-edges in $O(K)$ rounds for any $K = O(1/\epsilon)$. 

\begin{restatable}{theorem}{colorfulmatching}
\label{thm:colorful-matching}
Let $K > 0$ be a constant such that $K < 1/(18\epsilon)$.
There is a $O(\log\Delta)$-round algorithm computing a colorful matching of size at least $K\cdot\avganti_C$ with high probability in all cliques $C$ of average anti-degree $\avganti_C\ge 1/(2\alpha)$.
\end{restatable}

\begin{corollary}
    \label{cor:promising}
    After Step 2,
    there are at most $\Delta/\alpha$ \emph{unpromising} nodes in $C$.
\end{corollary}

\begin{proof}
In a clique $C$ with $\avganti_C \le 1/(2\alpha)$, at most $|C|/(2\alpha)$ nodes have anti-degree at least 1, by Markov inequality.
In a clique $C$ with $\avganti_C \ge 1/(2\alpha)$, we compute a colorful matching $M$ with $2\alpha\cdot\avganti_C$ edges. By Markov inequality, at most $|C|/(2\alpha)$ nodes have anti-degree more than $|M|$. In both cases, at most $(1+\epsilon)\Delta/\K \le 2\Delta/(2\alpha) = \Delta/\alpha$ nodes are unpromising.
\end{proof}

\subsection*{Step 3: Coloring Dense Nodes}

\paragraph{Reducing the number of uncolored nodes.}
When nodes try colors from their palettes, they get colored with constant probability. In our setting, nodes cannot directly sample colors from their palette as they must use colors from the lists they sampled in \cref{alg:palette-sampling}. If they have large enough palette though, (uninformed) sampling $O(\log n)$ colors in $[\Delta+1]$ is enough to find one in their palette with constant probability.

\begin{lemma}
\label{lem:reduce}
There exists a $O(\log\log n)$-round algorithm such that, with high probability, the number of uncolored nodes in each almost-clique is afterwards at most $\Delta/(\alpha\beta)$. Furthermore, nodes only use $O(\log n\cdot\log\log n)$ fresh random colors.
\end{lemma}

\begin{proof}
Consider a clique $C$ with $k \ge \Delta/(\alpha\beta)$ uncolored nodes. By \cref{lem:clique-palette}, every node has $|\Psi_C\cap\Psi_v| \ge k$ colors in its palette.
A node is set as active (independently) with probability $1/4$. For a node $v$, denote its uncolored degree by $\hatd_v$. If $\hatd_v< |\Psi_v|/2$, a random color in $\Psi_v$ colors $v$ with probability $1/2$. Otherwise, by the classic Chernoff bound, with probability $1-e^{-\Omega(\hatd_v)} \ge 1-e^{-\Omega(k)} \ge 1-1/\poly(n)$, node $v$ has at most $\hatd_v/2$ active neighbors. Therefore, for any conditioning on the colors tried by active neighbors, $v$ retains a uniform random color $c\in\Psi_v$ with probability $1/2$. 
If $v$ samples $t=\alpha\beta$ colors in $[\Delta+1]$ independently, it fails to find at least one color from its palette with probability $(1-|\Psi_v\cap\Psi_C|/\Delta)^t \le (1-1/(\alpha\beta))^t \le \exp(-t/(\alpha\beta)) \le 1/2$. Overall, a fixed node $v$ retains a color with probability $1/4\times1/2\times 1/2=1/16$. So the expected number of uncolored nodes $\Exp[|\hC|]$ in $C$ decreases by constant factor each round. By Chernoff with domination \cref{lem:chernoff}, it holds with probability $1-e^{-\Omega(|\hC|)}=1-1/\poly(n)$ at each round (because $\Delta \ge \Omega(\log^4 n)$). After $O(\log(\alpha\beta))=O(\log\log n)$ rounds, with high probability, $|\hC| < \Delta/(\alpha\beta)$.
\end{proof}

\paragraph{Finishing the coloring with augmenting paths.} 
Now that the number of uncolored nodes is small, we resort to the new technique of coloring with \emph{augmenting paths} outlined in \cref{sec:intro-augtree}.

\begin{restatable}{theorem}{thmaugpath}
\label{thm:augpath}
Assume all cliques have at most $\Delta/(\alpha\beta)$ uncolored nodes and at most $\Delta/\alpha$ unpromising nodes. There is a $O(\log\Delta)$-round algorithm \augpath that colors a constant fraction of the nodes in each clique with high probability.
\end{restatable}

Before giving a detailed description and proof of the \augpath algorithm in \cref{sec:augmenting-path}, we conclude this section with the proof of our main theorem.

\subsection*{Proof of \cref{thm:DPS}}
If $\Delta = O(\log^4 n)$, nodes can store all their adjacent edges and color the graph in $O(\log^3\log n)$ rounds of \congest \cite{BEPSv3,GK21}.

Assume now $\Delta \ge \Omega(\log^4 n)$. We precondition almost-cliques (using \cref{thm:strong-ACD}) with $\eta=O(\log n)$ (\cref{eq:emax}) in $O(\log\Delta)$ rounds and using $O(\log^2 n)$ random colors. 
The colorful matching requires $O(\log\Delta)$ rounds (\cref{thm:colorful-matching}) and almost-cliques have at most $\Delta/\alpha$ unpromising nodes (\cref{cor:promising}). 
For $O(\log\log n)=O(\log\Delta)$ rounds, nodes try random colors from their palettes. All cliques are left with $\Delta/(\alpha\beta)$ uncolored nodes (by \cref{lem:reduce}). 
We run \augpath for $O(\log\Delta)$ times. Each time, the number of uncolored nodes decreases by a constant factor with high probability (\cref{thm:augpath}). Overall, we use $O(\log^2\Delta)$ rounds to complete the coloring. \Qed{thm:DPS}

\section{Augmenting Paths}
\label{sec:augmenting-path}

This section is dedicated to the central argument of \cref{thm:DPS}.

\thmaugpath*

We first give a high-level description of the complete algorithm in \cref{sec:augmenting-path-overview}. \cref{sec:aug-grow} contains the proofs related to the first part of the algorithm: growing the augmenting trees. We call the phase of recoloring augmenting paths \emph{harvesting the trees} and address it in \cref{sec:harvest-tree}.

\subsection{Detailed Description of the Algorithm}
\label{sec:augmenting-path-overview}

\begin{Algorithm}
\label{alg:aug}
\augpath.

\textbf{Input:} A partial coloring such that each almost-clique has $|\hC|<\Delta/(\alpha\beta)$ uncolored nodes, at most $\Delta/\alpha$ unpromising nodes, and each $v\in C$ is connected to at most $\emax$ nodes in other almost-cliques.

\vspace{1em}
For each almost-clique $C$, do \emph{in parallel}:
\begin{enumerate}
\item Count the number of uncolored nodes $k = \card{\hC}$.
\item $F=\growtree$ (\cref{alg:growtree}). 
\item if $k\ge \beta$, run $\harvest(k,F)$ for high $k$ (\cref{alg:harvest-high}).

\item if $k < \beta$, run $\harvest(k,F)$ for small $k$ (\cref{alg:harvest-low}).
\end{enumerate}
\end{Algorithm}

\begin{definition}[Augmenting Path]
\label{def:augmenting-path}
Let $P=u_1u_2, \ldots, u_t$ be a path in $C$ where $u_1$ is uncolored and $u_i$ has color $c_i$ for each $2 \le i \le t$. 
We say it is an \emph{augmenting path} if $u_t$, the colored endpoint of $P$, knows a color $c\neq c_t$ such that if we recolor each node $u_i$ using color $c_{i+1}$ for $i\in[t-1]$ and $u_t$ using $c$, the coloring of the graph remains proper.
\end{definition}

From an uncolored node $u=u_1$ in a clique with one uncolored nodes. Start with the path $P=u_1$ and as long as $P=u_1,\ldots, u_i$ is not augmenting, do the following: $u_i$ samples a color $c\in[\Delta+1]$, if $c$ is not used in the clique $P$ is an augmenting path; if $c$ is used by a node $u_{i+1}\in C$, add $u_{i+1}$ to the end of $P$ and repeat this process.
Unfortunately, this algorithm is not fast enough as each time we extend $P$, we find an augmenting path with probability $1/\Delta$. Hence, we need to spend $\Omega(\Delta)$ rounds exploring the clique before finding the one available color. To speed-up this process, we grow a tree of many augmenting paths.

\begin{definition}[Augmenting Tree/Forest]
\label{def:augmenting-tree}
An \emph{augmenting tree} is a tree such that each root-leaf path is augmenting, provided the leaf finds an available color.
An \emph{augmenting forest} is a set of disjoint augmenting trees.
\end{definition}

Say we computed an augmenting forest such that all trees have $\Omega(\Delta/k)$ leaves. Since a leaf finds an available color in $\Psi_C$ with probability $\Omega(k/\Delta)$, each tree contains an augmenting path with constant probability. 

\paragraph{Technical challenges.}
This process can fail in several ways.
\begin{enumerate}
\item\label[tech-issue]{part:issue-balanced} 
We need to show a constant probability of progress \emph{for each uncolored node}. It is not enough to have that all leaves in $C$ recolor their path with probability $\Omega(k/\Delta)$. We need to show that (with constant probability) \emph{each tree} finds a leaf with which it can recolor its root. Moreover, trees connecting to different roots must be disjoint.
\item\label[tech-issue]{part:spoiled-nodes} Consider a tree $T$ and one of its leaves $v\in T$. If $v$ has a high anti-degree, it has few available colors among the ones available in the clique (\cref{lem:clique-palette}). Similarly, it is possible that all external neighbors of $v$ block the $k$ colors it has available. We call such nodes \emph{spoiled} and must ensure that they only represent a small fraction of every tree.
\item\label[tech-issue]{part:issue-inside} Assuming all trees have $\Theta(\Delta/k)$ \emph{unspoiled} leaves, the leaves used to recolor augmenting paths in each tree have to use different colors. We say that we \emph{harvest} the trees. When $k$ is $\Omega(\log n)$, if each leaf try one color, w.h.p., the number of conflicts between trees is small, so the issue is merely to detect them. When $k$ is $O(\log n)$, as leaves try $\Theta(\log n)$ colors (to ensure to be successful w.h.p.), many conflicts may arise.
\end{enumerate}

\paragraph{Growing balanced augmenting trees.}
To overcome \cref{part:issue-balanced}, when growing the trees, we ensure \emph{they all grow at the same speed}. More precisely, the \growtree algorithm (\cref{alg:growtree}) takes as input a forest $F$ and finds \emph{exactly} $\beta$ children for each leaf in $F$ for $\floor{\log_\beta(\Delta/(\alpha k))}$ rounds so that each tree has $\Omega(\Delta/\beta k)$ leaves. Nodes then sample a precise number of colors to ensure that w.h.p.\ the majority of them finds enough leaves to grow trees to $\Theta(\Delta/k)$ leaves (\cref{lem:unspoiled-leaves}).

\paragraph{Bounding the number of spoiled nodes.} When a leaf $v$ attempts to recolor its path to some uncolored node, it must sample colors in $\Psi_C\cap\Psi_v$. This might be a problem for two reasons. First, if $v$ is unpromising (\cref{lem:clique-palette}). 
The second possibility is it that its $k$ colors in $\Psi_C\cap\Psi_v$ are blocked by external neighbors. The latter eventuality demands more caution. In particular, if we allow adversarial behavior on the outside, it might be that external neighbors block the one remaining color in $\Psi_C$ for all nodes.

\begin{definition}[Spoiled Node]
\label{def:spoiled-nodes}
We say a node $v\in C$ is \emph{spoiled} if $|\Psi_C\cap \Psi_v| \le k/2$ \emph{after conditioning on the outside}.
\end{definition}

We deal with \cref{part:spoiled-nodes} in two ways. 
We previously computed a colorful matching (\cref{def:colorful-matching}) of size $\Theta(\avganti_C)$ for a large enough constant to reduce the number of unpromising nodes to a sufficiently small fraction of $C$ (\cref{cor:promising}).
Second, we show that it is very unlikely that nodes outside of $C$ block more than $k/2$ colors for a large fraction of $C$. It stems from the two following observations
\begin{itemize}
    \item external neighbors in other cliques try $\Theta(\log n)$ \emph{uniform colors} in $[\Delta+1]$, and 
    \item the preconditioning of almost-cliques  reduced the external degree to $\emax=O(\Delta/\log n)$. (\cref{thm:strong-ACD})
\end{itemize}
So, with high probability, $\Omega(\Delta)$ nodes in $C$ have $\Omega(k)$ available colors in $\Psi_C$.
More precisely, in \cref{lem:conditionning-outside}, we show that with high probability \emph{over the randomness outside of $C$}, at most $3\Delta/\alpha$ nodes are spoiled in $C$. Then, \cref{lem:unspoiled-leaves} show that with high probability \emph{over the randomness inside} $C$, all trees have $\Theta(\Delta/k)$ unspoiled leaves.

\paragraph{Harvesting trees.}
While \cref{part:spoiled-nodes} was about the conflict with external neighbors, \cref{part:issue-inside} is about the conflicts inside the clique. Say leaves sample one color. For a fixed tree $T_u$ and one of its unspoiled leaves $v$, the expected number of colors blocked in $\Psi_C\cap\Psi_v$ by sampling in other trees is $(k-1)\cdot \Theta(\Delta/\alpha k)\cdot k/\Delta = \Theta(k/\alpha)$. When $k=\Omega(\log n)$, we get concentration and show that w.h.p.\ a constant fraction of the leaves still have $\Theta(k)$ colors available, even after revealing the randomness in other trees. Therefore, as long as $k=\Omega(\log n)$, \harvest colors a constant fraction of the uncolored nodes with high probability (\cref{lem:harvest-high}).

When $k=O(\log n)$, leaves sampling $\Theta(\log n)$ colors ensure w.h.p.\ that every leaf has $\Theta(\log n)$ colors to choose from. The drawback to such intensive sampling is that we must resolve conflicts between trees. Using the high expansion property of the sparsified graph, it is possible to deliver to every node in the clique the list of $\Theta(\log n)$ available colors to each of the $k$ uncolored nodes. Conflicts are then resolved locally (by every node).

\subsection{Growing the Trees}
\label{sec:aug-grow}

\begin{Algorithm}
\label{alg:growbranch}
\label{alg:growtree}
\growtree (for the $\ell$-th iteration of \augpath).

\textbf{Input:} An almost-clique $C$ with $|\hC|=k < \Delta/(\alpha\beta)$ uncolored nodes, a colorful matching $M$ of size at most $2\alpha\cdot\avganti_C$ and at most $\Delta/\alpha$ unpromising nodes.

\vspace{1em}

Define $d = \floor*{\log_\beta\frac{\Delta}{\alpha k}}$ and $U_0 = \hC$.

For $i=0$ to $d$,
\begin{enumerate}[label=(\textbf{G}\arabic*)]
\item\label[step]{step:sample} If $i < d$, let $x=5\beta$. 

For $i=d$, nodes compute the exact size of $|U_d|$ in $O(d)$ rounds and set $x=\floor*{\frac{6\Delta}{\alpha|U_d|}}$.

Each $u\in U_i$ picks a set $S_u$ of $x$ fresh random colors from $L_{3,\ell}(u)$.
\item\label[step]{step:filter} For each active node $u\in U_i$, let $S_u'\subseteq S_u$ be the colors $c$ that are
\begin{enumerate}[label=\roman*)]
    \item\label[step]{step:filter-ext} unused by colored external neighbors of $u$ and $c\notin L_{3,\ell}(\bigcup_{C'\neq C} C'\cap N(u))$;
    \item unused by nodes of $B_i=U_{\le i}\cup M$;
    \item uniquely sampled: $c\notin\bigcup_{v\in U_i\setminus\set{u}} S_v$; and
    \item \underline{used} by a colored node $v_{u,c}\in N(u)\cap C$
\end{enumerate}
\item\label[step]{step:grow} For each $u\in U_i$, choose $y_u$ arbitrary colors $c_1, \ldots, c_{y_u}\in S_u'$ where 
\[
y_u = \left\{\begin{array}{ll}\beta&\text{if } i < d\\ |S_u'|&\text{if }i=d\end{array}\right. \ .
\]
Define $v_{u,j}$ as the node of $(N(u)\cap C)\setminus B_i$ with color $c_j$ for all $j\in[y_u]$.

Let $U_{i+1}=\bigcup_{u\in U}\set{v_{u,1},\ldots, v_{u,y_u}}$ and $\pi(v_{u,j})=u$ for all $j\in[y_u]$ and $u\in U_i$.
\end{enumerate}

\end{Algorithm}

\paragraph{Bounding Conflicts with the Outside.}
In the next lemma, we bound the number of spoiled nodes in $C$ (\cref{def:spoiled-nodes}). Note that the probability is taken only over the randomness of nodes \emph{outside of $C$}.

\begin{lemma}
    \label{lem:conditionning-outside}
    Consider an almost-clique $C$. With high probability over $L_{3,\ell}(V\setminus C)$, almost-clique $C$ contains
    at most $3\Delta/\alpha$ spoiled nodes.
\end{lemma}

\begin{proof}
    By \cref{cor:promising}, the clique contains at least $|C|-\Delta/\alpha\ge(1-\epsilon-1/\alpha)\Delta \ge (1-2/\alpha)\Delta$ promising nodes, and each such node $v$ has at least $k$ colors in $\Psi_C \cap \Psi_v$. Let us focus on a set $S$ of exactly $(1-2/\alpha)\Delta$ promising nodes. For each selected promising node $v\in S$, we focus on exactly $k$ colors $\Psi'_v\subseteq \Psi_C \cap \Psi_v$. We consider the $k \card{S}$ pairs $(c,v)$ where $v\in S$ is selected promising node and $c$ is a color $c\in \Psi_v'$ from its selected colors.

    Let us denote by $\Next(v)\eqdef N(v)\cap \bigcup_{C'\neq C}C'$, the external neighbors of $v$ that might get (re)colored. Let $B = \bigcup_{v \in S} \Next(v)$ the set of vertices that are external neighbors of at least one node in $S$. Recall that, by \cref{thm:strong-ACD}, for each $v\in C$, $|\Next(v)|\le \emax$.
    For each $u \in B$ and color $c \in [\Delta+1]$, let $X_{u,c}$ be defined as\footnote{where $\mathbb{I}[\evt]$ is the indicator r.v.\ of some event $\evt$}:
    \[
    X_{u,c} = \frac{\card{\set{v\in S\cap N(u) \text{ such that } c \in \Psi'_v}}}{\emax}\cdot \mathbb{I}[c \in L_{3,\ell}^*(u)]\ ,
    \]
    counting how many times a $u\in B$ is in conflicting with the selected nodes $S$ over a selected color $c$, re-scaled by $1/\emax$ so $X_{u,c}$ is distributed in $[0,1]$. 
    
    Let $X=\sum_{u\in B}\sum_{c\in[\Delta+1]} X_c$. Each edge between $v\in S$ and $u\in B$ contributes $1/\emax$ to $X$ for each color $c\in \Psi_v'$ such that $c\in L_{3,\ell}(u)$. There are at most $|S|\emax$ such edges, and each color $c\in \Psi_v'$ is sampled in $L_{3,\ell}(u)$ with probability at most $20\beta/\Delta$. By linearity of expectation, $\Exp[X] \le |S|\cdot k\cdot \frac{40\beta}{\Delta} = \Theta(k\beta)$ because $|S|=\Theta(\Delta)$. Note $\emax$ cancelling itself.

    Random variables $X_{u,c}$ are independent, since each color is sampled independently.
    By Chernoff Bound (\cref{lem:chernoff}), $\Pr[X \geq 2\Exp[X]] \le \exp(-\Exp[X]/3) \le 1/\poly(n)$.
    Therefore, there are at most $\emax \cdot \frac{40\beta k \card{S}}{\Delta} \leq k \card{S}/(4\alpha)$ conflicts between the selected colors of nodes in $S$ and the colors sampled by their external neighbors (by \cref{eq:emax}). Therefore, at most $\frac{k\card{S}/(4\alpha)}{k/2} \le \Delta/\alpha$
    node are left with fewer than $k/2$ colors in $\Psi_v\cap \Psi_C$.
\end{proof}

\paragraph{Growing the Forest.}
Henceforth, we fix the random lists $L_{3,\ell}(V\setminus C)$ such that $C$ contains at most $3\Delta/\alpha$ spoiled nodes, which holds w.h.p.\ by \cref{lem:conditionning-outside}. 
Let $U=\bigcup_i U_i$ and $\pi$ be the values produced by \cref{alg:growtree}, the graph $F=(U,E(F))$ with $E(F)=\set{u\pi(u): u\in U_{>0}}$ is a forest. Suppose that at each \cref{step:grow}, each $u$ satisfies $|S_u| \ge y_u$. Then $F$ is a forest of $\beta$-ary trees of depth $d+1$ and at most $6\Delta/(\alpha k)$ leaves. We reveal the randomness inside $C$ as we grow the tree, conditioning at each growing step on arbitrary randomness from nodes that are already in the forest. The following lemma shows that it is unlikely that nodes sample bad colors.

\begin{lemma}
\label{lem:sample-color}
Let $0 \le i < d$. Then, for any lists in $\Lg_{3,\ell}(U_{\le i})$, if a node $u$ samples a fresh color $c$, we have
\[\Pr_{c\in[\Delta+1]}[c\text{ violates a condition in \cref{step:filter}}|~U_{\le i}] \le 12/\alpha \ . \]
\end{lemma}

\begin{proof}
For a fixed $i < d$, we have $|U_i|\le k\beta^{i}$. 
We bound the number of colors $c$ might be conflicting with for each item in \cref{step:filter}:
\begin{enumerate}[label=\roman*)]
\item Node $u$ has at most $\epsilon\Delta$ colored external neighbors (\cref{thm:strong-ACD}, \cref{part:SACD-inside}). The number of colors in $L_{3,\ell}(\bigcup_{C'\neq C} C'\cap N(u))$ is at most $40\beta \emax < \Delta/\alpha$ by \cref{eq:emax}.
\item For $i \le d$, the number of nodes in $U_{\le i}$ is $\sum_{j=0}^{d} |U_j| \le 2k\beta^{d} \le 2\Delta/\alpha$. Adding the colorful matching, at most $\Delta/\alpha + 2\alpha\avganti_C \le (1/\alpha + 2\epsilon\alpha)\Delta \le 3\Delta/\alpha$ nodes are colored in $B_i$ (by \cref{eq:eps-alpha}).
\item The number of colors sampled (and thus blocked) by active nodes is 
$x|U_i| = xk\beta^i \le 6k\beta^{d} \le 6\Delta/\alpha$
where the first inequality comes from $i<d$.
\item The number of uncolored nodes in $C$ is at most $\Delta/\alpha\beta$; hence, the number of colors used in $N(u)\cap C$ is be at least\footnote{Note that, apart for the colorful matching, each node in $C$ uses a different color} $(1-\epsilon)\Delta-\Delta/\alpha\beta \ge (1-2/\alpha)\Delta$ by \cref{eq:eps-alpha}. 
\end{enumerate}

Summing all failure probabilities with an union bound, we get the claimed bound.
\end{proof}

Since nodes sample $\Omega(\log n)$ colors when $i<d$, \cref{lem:sample-color} implies the following corollary.

\begin{corollary}
With high probability over $\Lg_{3,\ell}(U_{< d})$, for each $i < d$ and $u\in U_i$, $|S_u'| \ge \beta\eqdef y_u$.
\end{corollary}

Because of rounding in $d$, trees might not contain enough leaves after $d$ growing steps. Furthermore, we also need to show that most leaves are unspoiled.

\begin{lemma}
\label{lem:unspoiled-leaves}
Let $U_d$ be the set given by \cref{alg:growtree}. With high probability over $\Lg_{3,\ell}(U_d)$, the number of unspoiled nodes in $U_{d+1}$ is at least $\Delta/\alpha$. 
\end{lemma}

\begin{proof}
For each node $u\in U_{d-1}$, define a random variable $X_{u,i}$ for each of its sampled color $i\in[x]$. Let $X_{u,i}$ be one if and only if the $i$-th color it samples 1) has no conflict in \cref{step:filter} and 2) the corresponding $v_{u,c}$ is unspoiled. The analysis is similar to \cref{lem:sample-color} but has to be a bit more careful. Namely, failures caused by i), ii) and iv) remain unchanged but the number of colors sampled by active nodes is different. Furthermore, we now have to filter out spoiled nodes.
\begin{itemize}
\item Spoiled nodes are easily dealt with by \cref{lem:conditionning-outside}. Indeed, there are at most $3\Delta/\alpha$ spoiled nodes in $C$. Since each such node blocks one color, they only amount to a small fraction.
\item Since nodes try $x=\floor*{6\Delta/(\alpha|U_d|)}$ colors, the total number of colors sampled by active nodes is  $x|U_d| \le 6\Delta/\alpha$.
\end{itemize}
If we union bound over all possible failures for a random color $c\in[\Delta+1]$, we get $\Pr[X_{u,c}=1|U_{d}\setminus\set{u}] \le 15/\alpha \le 1/25$.
By Markov inequality, a node $u\in U_d$ samples more than $x/5$ bad colors w.p.\ at most $\frac{x}{25}\cdot\frac{5}{x}=1/5$. Giving priority to nodes of lowest ID, the martingale inequality (\cref{lem:chernoff}) shows that, w.p.\ $1-e^{-\Omega(|U_d|)}$, at most $|U_d|/4$ nodes sample more than $x/5$ bad colors. Note that, by definition of $d$,
\[ |U_d| = k\beta^d \ge k\beta^{\log_\beta(\Delta/\alpha k)-1} \ge \frac{\Delta}{\alpha\beta} \ . \]
Therefore, $1-e^{-\Omega(|U_d|)} \ge 1-e^{-\Omega(\Delta/\beta)} \ge 1-\poly(n)$ (because $\Delta \in \Omega(\log^4 n)$) and the previous claim holds with high probability. That means that w.h.p.\ the number of unspoiled nodes in $U_{d+1}$ is at least 
\begin{align*}
    \frac{|U_d|}{4}\cdot\frac{4x}{5} &\ge \frac{|U_d|}{5}\parens*{\frac{6\Delta}{\alpha|U_d|}-1} \tag{because $x\eqdef\floor*{\frac{6\Delta}{\alpha|U_d|}}$}\\
    &=\frac{6\Delta/\alpha-|U_d|}{5} \ge \frac{\Delta}{\alpha} \ ,\tag{because $|U_d|\le\Delta/\alpha$}
\end{align*}
which concludes the proof of the Lemma.
\end{proof}

\subsection{Harvesting the Trees}
\label{sec:harvest-tree}

In the previous section, we argued that there were $\Delta/\alpha$ unspoiled leaves. In this section, we argue that enough of these leaves find good colors to color a constant fraction of uncolored nodes. While in \cref{lem:unspoiled-leaves}, we bound from below the total number of unspoiled nodes, because all trees have roughly the same size (at most $6\Delta/\alpha k$ each), a simple counting argument gives the following claim.

\begin{claim}
There are at least $0.9k$ trees with $\Delta/(2\alpha k)$ unspoiled leaves.
\end{claim}

Henceforth, we will be focusing on those trees with many unspoiled leaves. Note that at \cref{step:filter} of \cref{alg:growtree}, we ensure that if a leaf finds a color, each node on its path to the root can change its color without creating conflicts. Hence, this section focuses on counting successful leaves in each tree.

\paragraph{When $k$ is large.} Assume first that $k\ge \Omega(\log n)$.

\begin{Algorithm}
\label{alg:harvest-high}
\harvest (for $k$ greater than $\Omega(\log n)$).

\textbf{Input:} the forest $F$ of augmenting trees computed by \cref{alg:growtree}.

\begin{enumerate}[label=(\textbf{H}\arabic*)]
\item Leaves $v\in F$ try \emph{one} fresh color $c_v\in \Lh_3(v)$. We call $v$ \emph{successful} if it can adopt $c_v$, and thereby recolor the path in $F$ connecting $v$ to its uncolored root. Leaves can learn if they are successful in $O(d)$ rounds. 
\item 
Each leaf can learn in $O(1)$ rounds if a leaf from another tree sampled the same color (by \cref{claim:color-group}). We call a tree \emph{successful} if it has at least one successful leaf that is not conflicting with leaves from other trees.

\item\label[step]{step:augpath-recolor} 
If a tree is successful, it can recolor the path from its root to its successful leaf in $O(d)$ rounds.
\end{enumerate}

\end{Algorithm}

\begin{lemma}
\label{lem:harvest-high}
Let $C$ be a clique with $\beta \le k \le \Delta/(\alpha\beta)$ uncolored nodes.
In \cref{step:augpath-recolor}, with high probability over $\Lh_{3,\ell}(C)$, we color at least $\Omega(k/\alpha)$ uncolored nodes.
\end{lemma}
\begin{proof}
    Let $k'=0.9k$ and 
    $T_1, \ldots, T_{k'}$ be $k'$ trees with at least $\Delta/(2\alpha k)$ successful paths to unspoiled leaves. For each such tree, let us only consider exactly $\Delta/(2\alpha k)$ selected unspoiled leaves.
    For each selected leaf $v$, let us focus on a subset of size $k/2$ of its palette $\Psi_C\cap\Psi_v$, which we call its selected colors.

    For each $i \in [k']$, let $X_i$ be the indicator random variable for the event that (1) tree $i$ contains exactly one selected leaf that samples a selected color, and (2) no other node in the almost-clique tried the same color as this selected leaf. Note that $X_i$ may be expressed as the difference between two random variables $Y_i$ and $Z_i$, where $Y_i$ corresponds to the event that at least one selected leaf of $T_i$ tries one of its selected colors, and $Z_i$ corresponds to at least two selected leaves trying a selected color or a selected leaf trying a selected color but failing to keep it. We have:

    \[
    \Pr[X_i=1] \geq \frac{\Delta}{2\alpha k} \cdot \frac{k}{2 \Delta} \cdot \parens*{1-\frac{k/2 + 1}{\Delta}}^{\Delta/(2\alpha k)-1} \cdot \parens*{1-\frac{1}{\Delta}}^{6\Delta/\alpha}
    \]

    The first term ($\Delta/(2\alpha k)$) corresponds to doing a sum over all selected nodes in $T_i$ of their probability of being the selected node that succeeds. The second term ($k/(2\Delta) $) corresponds to the probability that each selected node samples one of its selected colors. Call this color $c$. The third term corresponds to all other selected leaves of $T_i$ sampling neither one of their selected colors nor $c$. The last term corresponds to all leaves not trying $c$.
    Using $1-x/2 \geq e^{-x}$ for $x\in [0,1]$ (\cref{lem:useful-inequalities}), $\Pr[X_i=1] \geq e^{-12/\alpha}/(4\alpha)$. Let $X = \sum_{i=1}^{k'100} X_i$, by linearity $\Exp[X] \geq k' \cdot e^{-12/\alpha}/(4\alpha)=\Omega(k/\alpha)$.

    Consider similarly the random variables $Y_i$ and their sum $Y$. We have:
    \[\Pr[Y_i=1] = 1-\parens*{1-\frac{k}{2\Delta}}^{\Delta/(4\alpha k)} \geq \frac{1}{8 \alpha +1}\]

    Therefore, $\Exp[X]$ and $\Exp[Y]$ are both of order $\Theta(k / \alpha)$. Additionally, $Y$ is $1$-Lipschitz and $1$-certifiable, and $Z=Y-X$ is $2$-Lipschitz and $2$-certifiable. Hence, by \cref{lem:talagrand-difference}, $\Pr[X < k' \cdot e^{-12/\alpha}/(8\alpha)] \leq \exp(-\Omega(k/\alpha))$. Since $X$ is a lower bound on the number of trees that successfully recolor their root, at least $\Omega(k/\alpha)$ uncolored nodes get colored, w.h.p.
\end{proof}

\paragraph{Coloring the last nodes.} Assume now that only $k = O(\log n)$ uncolored nodes remain.

\begin{Algorithm}
\label{alg:harvest-low}
\harvest (for $k$ smaller than $O(\log n)$).

\textbf{Input:} the forest $F$ of augmenting trees computed by \cref{alg:growtree}.
\begin{enumerate}[label=(\textbf{L}\arabic*)]
\item\label[step]{step:harvest-low-sample} Leaves try $\beta$ fresh colors. A leaf keep its color if it not used by any colored neighbor nor tried by external neighbors.
\item Via simple aggregation on the augmenting tree, each uncolored node $v$ learns a list $S_v$ of up to $\Theta(k)$ \emph{candidate} colors that its leaves could use to recolor themselves.
\item\label[step]{step:match-uncolored} 
After some routing, the whole almost-clique knows about all of $\set{S_v, v\in \hC}$. All nodes then locally 
compute the same matching of size $\Omega(k)$ in the graph with vertices $\hC\cup [\Delta+1]$ and edges $(v,c)$ iff $c\in S_v$.
\end{enumerate}
\end{Algorithm}

Sending the sets $\set{S_v, v\in \hC}$ in \cref{step:match-uncolored} is done by \randompush.

\begin{Algorithm}
\label{alg:randompush}
\randompush.

\textbf{Input:} An almost-clique $C$ with $x$ messages.
\smallskip

For each node $v$ that knows at least one message and each incident edge, 
$v$ picks a random message that it knows and sends it along the edge.
\end{Algorithm}

\begin{restatable}{lemma}{randompushlemma}
    \label{lem:randompush}
    Let $x \leq \beta^3$ messages of $O(\log n)$ bits each known by exactly one node in the sparsified almost-clique $\sparse{C}$. After $O(\log \Delta)$ iterations of \randompush, each node in the sparsified almost-clique learns all $x$ messages, with high probability.
\end{restatable}

The proof of \cref{lem:randompush} follows easily from the expansion of the sparsified almost-clique (see \cref{lem:expansion}), and we defer it to \cref{app:randompush}.

\begin{lemma}
\label{lem:last-nodes-progress}
Let $C$ be a clique with $k \le \beta$ uncolored nodes. At the end of \cref{step:match-uncolored}, with high probability over $\Lh_{3,\ell}(C)$, we color at least $\Omega(k/\alpha)$ uncolored nodes.
\end{lemma}

\begin{proof}
There exist $k'=0.9k$ uncolored nodes with at least $\Delta/(2\alpha k)$ successful paths to unspoiled leaves.
  
Let us now argue that an uncolored node $v$ with this many successful paths in its tree has them find at least $k/4$ distinct colors, w.h.p. Let us associate to the $i$th such leaf a random variable $X_i$ such that:
\begin{itemize}
    \item If previous leaves discovered $k/4$ distinct colors already, then $X_i = 1$ w.p.\ 1,
    \item If previous leaves discovered fewer than $k/4$ distinct colors, then $X_i = 1$ iff leaf number $i$ discovers a new color.
\end{itemize}

When previous leaves have discovered fewer than $k/4$ distinct colors, since the $i$-th leaf is unspoiled, it still has at least $k/4$ colors it can discover. The probability that it finds one of them is at least:
\[ \Exp[X_i]= 1-\parens*{1-\frac{k}{4\Delta}}^\beta \geq 1 - \frac{1}{1+k\beta/(4\Delta)} = \frac{k\beta}{4\Delta + k\beta} \geq \frac{k\beta}{5\Delta}\]
where we used \cref{lem:useful-inequalities} for the first step and the last comes from $\Delta \ge \beta^2 \ge k\beta$.

Therefore, $\Exp[\sum_i X_i] \geq \frac{\Delta}{2\alpha k} \cdot \frac{k\beta}{5\Delta} = \beta / (10\alpha)$. The series of $X_i$ satisfy the conditions of \cref{lem:chernoff}, hence it has value at least $\beta / (20\alpha)$, w.h.p.
By definition of $\set{X_i}$, if $\beta / (20\alpha) \geq k/4$, this gives that at least $k/4$ colors are found in the tree, while if $\beta / (20\alpha) \leq k/4$, we only get that at least $\beta / (20\alpha) \geq k/(20\alpha)$ colors are found in the tree (as $\beta \ge k$).

We now argue that our algorithm has the claimed runtime. Nodes can share their sampled colors with nodes in other cliques in $O(\log\Delta)$ rounds as the total amount of bits to communicate is $O(\log n\cdot\log \Delta)$.

In each tree, the root learns a subset of the colors its leaves can pick in $O(\log \Delta)$ rounds as follows: in each round, each node that knows about an available color that it has not yet sent towards the root sends as many such colors that it can towards the root (with a maximum of $\log \Delta / \log n$ colors per round). In $O(\log \Delta + k\cdot \log \Delta / \log n)$ rounds, the root learns about a set $S_v$ of $\Omega(k)$ colors, if that many are available in the tree.

Each root $v$ crafts a message of the form $(ID_v,c)$ for each of the colors $c\in S_v$, and selects a subset of $k$ of them if it has more than $k$. The almost-clique then runs \randompush with $O(k^2)=O(\beta^2)$ messages for $O(\log \Delta)$ rounds with the selected messages. 
Note that the bipartite graph with vertices $\hC\cup [\Delta+1]$ and edges $(v,c)$ such that $c\in S_v$ is now known to all nodes in $C$. Moreover, this graph has $\Omega(k^2)$ edges and maximum degree $k$. Therefore, it has a matching of size $\Omega(k)$ (which can be computed locally by a simple deterministic greedy algorithm). It follows that all nodes compute the same matching without extra communication and recolor the path corresponding to each edge in the matching in $O(d)$ rounds. Therefore, at least $\Omega(k)$ nodes get colored.
\end{proof}

\section{Colorful Matching}
\label{sec:colorful-matching}

In this section, we show the following theorem:

\colorfulmatching*

Throughout this section, we fix a clique $C$ and fix the colors \emph{used and sampled} outside of $C$ adversarially.
At the beginning of this step, nodes of $C$ are uncolored. 
For a set $D$ of colors, define $\avail_D(e)$ to be the number of colors that an anti-edge $e$ can adopt in $D$ without conflicting with external neighbors. 
This includes all possible colors in $L_2$ that external neighbors might use to color themselves. 
If $D$ contains a single color $c$, we abuse notation and denote $\avail_D(e)$ by $\avail_c(e)\in\set{0,1}$.
By extension, for a set $F$ of anti-edges $\avail_D(F)=\sum_{e\in F} \avail_D(e)$. 
A similar quantity was introduced by \cite{ACK19}. A major difference with \cite{ACK19} is that colors become unavailable if an uncolored external neighbor merely samples it in $L_2$. In particular, one uncolored external neighbor blocks $\Theta(\log n)$ colors. We can afford to lose so many colors because of preconditioning and (in contrast to \cref{sec:augmenting-path}), at this stage $\Omega(\Delta)$ colors are still available in the almost-clique.
The following lemma states that, at the beginning of this step, many edges have many available colors regardless of conditioning of random variables outside of $C$.

\begin{lemma}
\label{lem:avail}
Let $D=[\Delta+1]$ and $F$ be the set of all anti-edges in $C$.
For any (possibly adversarial) conditioning outside of $C$,  we have $\avail_D(F) \ge \avganti_C\Delta^2/3$.
\end{lemma}

\begin{proof}
For a fixed edge $e\in F$, we bound from below its number of available colors. Each \emph{colored} neighbor blocks at most one color. Observe that to compute a colorful matching, uncolored nodes use only colors sampled in $L_{2,i}$ for $i\in[O(1/\epsilon)]$ and $L_2^*$. By \cref{lem:lists-size}, an \emph{uncolored} neighbor in another clique blocks at most $L_{\max}=O(\log n)$ colors. Since a node in $C$ has at most $\epsilon\Delta$ colored neighbors (necessarily outside $C$) and at most $\emax=\Delta/\eta$ neighbors in other cliques (by \cref{part:SACD-ext} of \cref{thm:strong-ACD} and \cref{eq:emax}), we have $\avail_D(e) \ge \Delta - \epsilon\Delta - L_{\max} \cdot \Delta/\eta \ge (1-2\epsilon)\Delta$.
Summing over all edges, we get
\[\avail_D(F) = \sum_{e}\avail_D(e) \ge \frac{\avganti_C|C|}{2}\parens*{1-2\epsilon}\Delta \ge \avganti_C\Delta^2/3\ .\qedhere\]
\end{proof}

To compute a colorful matching, we greedily add same-colored anti-edges to $M$. Each time we do so, we remove the color of that edge from $D$, which then becomes unavailable to other edges, and we remove the matched edge as well as all its adjacent edges from $F$. We argue that as long as the total number of available colors is large, there must be colors available to many edges. We call a color $c$ \emph{heavy} if 
$\avail_c(F)\ge \avganti_C\Delta/20$. 
The following claim is immediate from the limited contributions to $\avail_D(F)$ from both all non-heavy colors and from each heavy color individually.

\begin{claim}
\label{claim:heavy-colors}
As long as $\avail_D(F)\ge \avganti_C\Delta^2/6$, there are at least $\Delta/10$ heavy colors in $D$. 
\end{claim}

\paragraph{When $\avganti_C$ is large.} We first run the following algorithm. It produces a large enough matching in cliques with $\avganti_C \ge \beta$.

\begin{Algorithm}
\matching.

\textbf{Input:} a constant $0 < K < 1/(18\epsilon)$.

\textbf{Output:} a colorful matching $M_C$ of size $K\cdot\avganti_C$ in each almost-clique such that $\avganti_C \ge \beta$.

\vspace{1em}
Initially, $M_C=\emptyset$ for each almost-clique $C$.

For $i=1$ to $5\cdot10^3\cdot K$, in each almost-clique $C$ \emph{in parallel}:
\begin{enumerate}
\item\label[step]{step:clique-slack-RC} Each uncolored $v\in C$ is active. It samples each $c\in[\Delta+1]$ independently into $L_{2,i}(v)$ with probability $p=1/(4\Delta)$. 
\item\label[step]{step:clique-slack-inactive} Each $v\in C$ sampling less or fewer than 1 color ($|L_{2,i}(v)| \ne 1$) becomes inactive.
\item\label[step]{step:clique-slack-inactive2} Each $v\in C$ 
sampling a single color $c$ becomes inactive if $c$ is used by a colored external neighbor, by an anti-edge in $M_C$, or is sampled by an external neighbor of $v$.
\item\label[step]{step:clique-slack-adopt} Each \emph{active} node retains its color if it has an \emph{active} anti-neighbor with this color. If the same color is used by two or more anti-edges, we keep the anti-edge with the smallest ID.
\end{enumerate}
\end{Algorithm}

\begin{lemma}
\label{lem:colorful-matching-porous}
Let $K > 0$ be a constant such that $K < 1/(18\epsilon)$.
W.p. $1-\exp(-\Omega(\avganti_C))$ over the randomness in $L_2(C)$, the set $M_C$ produced by $\matching$ is a colorful matching in $C$ of size at least $K\cdot\avganti_C$.
\end{lemma}

\begin{proof}
Suppose $\avail_D(F)\ge \avganti_C\Delta^2/6$. For each color $c$, define $A_c$ as the indicator random variable of the event that at least one anti-edge in $C$ samples $c$. For a heavy color $c\in D$, we bound $\Pr[A_c=1]$ from below by the probability that exactly one anti-edge samples $c$: 
\begin{align*}
    \Pr[A_c=1]&\ge\sum_{e\in F}\Pr[\text{endpoints of $e$ are the only nodes to sample $c$}]\\
    &\ge \avail_c(F) \cdot p^2 \cdot \parens*{1-p}^{|C|-2} \\
    &\ge \frac{\avganti_C\Delta}{20} \cdot \frac{1}{16\Delta^2} \cdot \exp(-2p|C|)  \tag{because $c$ is heavy}\\
    &\ge \frac{\avganti_C}{320e\Delta} \ . \tag{because $|C| \le (1+\epsilon)\Delta$}
\end{align*}
By \cref{claim:heavy-colors}, there exist at least $\Delta/10$ such heavy colors in $D$. Therefore, $\Exp[A]=\Exp[\sum_{c\in D} A_c] \ge \avganti_C/320e$. Since each color is sampled independently, random variables $A_c$ are independent and we can apply the classic Chernoff bound. With probability $1-\exp(-\Omega(\avganti_C))$, we have $A \ge \avganti_C/640e$.

Random variable $A$ is not the number of anti-edges we can insert in $M$ because we do not account for nodes sampling more than one color (\cref{step:clique-slack-inactive}). We emphasize that nodes can adopt any color available to them: by definition of $\avail$ it cannot be conflicting on the inside nor with colored nodes in $M$.
Let $B_c$ be the random variable equal to one if and only if at exactly one anti-edge sampled $c$ (i.e., $A_c=1$) and at least one endpoint becomes inactive in \cref{step:clique-slack-inactive}.
Condition on $A_c=1$ and let $e$ be the anti-edge that sampled $c$. 
Since an endpoint samples each color independently with probability $p$, it only samples $c$ with probability $(1-p)^{\Delta} \ge \exp(-1/2) \ge 3/5$ (by \cref{eq:exp-ineq}). Therefore, the probability that one endpoint of $e$ sampled more colors is at most $2\times (1-3/5)=4/5$. 
We overestimate the number of anti-edges inserted in $M$ by $B=\sum_{c\in D} B_c$ which, in expectation, is $\Exp[B]=\sum_{c\in D}\Exp[B_c|A_c=1]\Exp[A_c] \le (4/5)\cdot \Exp[A]$.
Random variable $B$ is $2$-Lipschitz.
It is also $2$-certifiable, because to certify that $B_c=1$ we can point at two colors sampled by one of the endpoints. By Talagrand inequality (\cref{lem:talagrand-difference}),
with probability $1-\exp(-\Omega(\avganti_C))$, we add $A-B \ge \Exp[A-B]/2 \ge \frac{\avganti_C}{5\cdot10^3}$ anti-edges to the colorful matching.

This shows that as long as $\avail_D(F) \ge \avganti_C\Delta^2/6$, we add $\frac{\avganti_C}{5\cdot10^3}$ anti-edges to the matching at each iteration. Hence, after $5\cdot10^3 \cdot K$ iterations, the matching has $K\cdot\avganti_C$ anti-edges. Observe that when we insert an anti-edge in $M_C$, we remove one color $c$ from $D$ and at most $2\epsilon\Delta$ anti-edges from $F$. Color $c$ contributed at most $\avail_c(F)\le\avganti_C\Delta \le \epsilon\Delta^2$ to $\avail_D(F)$ and each anti-edge at most $\Delta$.
Inserting an anti-edge in $M_C$ decreases $\avail_D(F)$ by at most $3\epsilon\Delta^2$. If, after some iterations, the number of available colors drops below $\avganti_C\Delta^2/6$, we must have inserted at least $\frac{\avganti_C\Delta^2/6}{3\epsilon\Delta^2} = \avganti_C/(18\epsilon) > K\cdot\avganti_C$ anti-edges into $M$.
\end{proof}

\paragraph{When $\avganti_C$ is small.}
If $\avganti_C \le O(\log n)$, \cref{lem:colorful-matching-porous} fails to compute a colorful matching with high probability. In this section, we explain how to compute a large enough matching in cliques with small anti-degree. Note that nodes can count the number of anti-edges in the matching in $O(\log\Delta)$ rounds. If they find fewer than $K\beta$ anti-edges, it must be that $\avganti_C < \beta$. We then uncolor all nodes in $C$ and run \cref{alg:matching-compact}.

Intuitively, since $\avganti_C \ge \Omega(1)$, if we repeat the previous procedure $\Theta(\log n)$ times, we would get a failure probability of $\parens*{e^{-\Theta(\avganti_C)}}^{\Theta(\log n)} = 1/\poly(n)$. This implies that even when $\avganti_C$ is a small constant, a colorful matching of size $\Theta(\avganti_C)$ exists. This was, in fact, already shown in the first palette sparsification theorem.

\begin{lemma}[{\cite[Lemma 3.2]{ACK19}}]
\label{lem:ack-colorful-matching}
Let $C$ be a $\epsilon$-almost clique, $D\subseteq[\Delta+1]$ and $F$ a subset of anti-edges in $C$. Fix any partial coloring given by \cref{thm:strong-ACD} where nodes of $C$ are uncolored and $\avail_D(F) \ge \avganti_C\Delta^2/3$. Suppose each node sample colors in $[\Delta+1]$ independently with probability $q\eqdef\gamma\frac{\beta}{\Delta}$ for some constant $\gamma(\epsilon)=\gamma > 0$ (depending on $\epsilon$ but not $n$ nor $\Delta$). Then there exists a colorful matching of size at least $\avganti_C/(414\epsilon)$ with high probability.
\end{lemma}

Note that our definition of $\avail$ is stronger than the one of \cite{ACK19} because it removes colors sampled by active external neighbor. Regardless of that, \cref{lem:avail} shows that we have a large number of available colors, which is the only requirement for the proof of \cite{ACK19}.

\begin{Algorithm}
\label{alg:matching-compact}
\matching (for almost-cliques $C$ with $\avganti_C < \beta$ in parallel).

\begin{enumerate}
\item Each $v\in C$ samples each $c\in[\Delta+1]$ into $L_2^*(v)$ independently with probability $q\eqdef \gamma\frac{\beta}{\Delta}$. 

\item Each node computes $S_v \eqdef L_2^*(v)\setminus L_2^*(V\setminus C)$, the set of colors that do not collide with those of external neighbors.

Let $\hS_v\eqdef\set{c\in S_v:~\exists u\in C\setminus N(v)\text{ such that } c\in S_u}$ be the colors of $v$ sampled by at least one anti-neighbor.

\item We count the number of candidate anti-edges in $C$: the number of pairs $(uv,c)$ where $uv$ is an anti-edge and $c\in \hS_v\cap \hS_u$. If there are more than $U\eqdef 4\gamma^2 K\avganti_C\beta^2$ candidate anti-edges, we select a set of colors $D$ such that the number of candidate edges with that color is at least $D$ and at most $O(\beta^3)$. If there are less than $U$ candidate edges, we let $D=[\Delta+1]$.

\item Each node $v$ forms messages $(ID_v, c)$ for each $c\in \hS_v\cap D$. 

Use \randompush to disseminate all messages $(ID_{v}, c)$ within $C$. 

Each node $v$ with $c\in \hS_v\cap D$ forms a message $(ID_u,ID_v,c)$ for each anti-neighbor $u$ with $c\in \hS_u\cap D$.

Use \randompush to disseminate all messages $(ID_u, ID_v, c)$ within $C$. 

\item Compute locally the colorful matching using anti-edges disseminated in the previous step.
\end{enumerate}
\end{Algorithm}

\begin{lemma}
\label{lem:colorful-matching-compact}
Let $K\in(0,1/(18\epsilon))$ be a constant. Consider all cliques with $1/(2\alpha) \le \avganti_C\le \beta$.
If nodes 
sample each color independently with probability 
$q=\gamma \frac{\beta}{\Delta}$ where $\gamma$ is the constant from \cref{lem:ack-colorful-matching},  
then in each clique $C$, with high probability, there exists a colorful matching that does not conflict with nodes on the outside. 
Moreover, \cref{alg:matching-compact} finds this matching in $O(\log\Delta)$ rounds.
\end{lemma}

\begin{proof}
We first explain how nodes compute $\hS_v$.
Each node $v$ starts by broadcasting $L_2^*(v)$ in $O(\log\Delta)$ rounds (since $|L_2^*(v)|=O(\log n)$).
We run a BFS for each color; two hops suffice by \cref{claim:color-group}.
To learn $\hS_v$, we count the number of nodes that sample each color using the BFS trees. A node needs to communicate over an edge only if both endpoints sampled the same color. Hence, we send at most $O(\log n\cdot\log\Delta)$ bits on an edge for each round of the BFS. In $O(\log\Delta)$ rounds, all nodes $v\in C$ know for each $c\in S_v$ how many other nodes $u\in C$ have $c\in S_u$. If this is more nodes than they know from their neighborhood, they must have an anti-neighbor with that color.

A candidate edge is a pair $(uv,c)$ where $u$ and $v$ are anti-neighbor and $c$ is a color such that $c\in \hS_v\cap \hS_u$. Namely, we could add edge $uv$ to the colorful matching using color $c$. For each color, we elect a leader amongst nodes that sampled that color. Using aggregation on 2-hops BFS trees, each leader learns the number of candidate edges for its color in $O(\log\Delta)$ rounds. We then run a BFS in the whole clique $C$ and aggregate the total number of candidate edges.

Suppose that the number of candidate edges is at most $U\eqdef4\gamma^2\cdot K\avganti_C\beta^2=O(\beta^3)$. For each candidate edge $(uv,c)$, we craft two messages $(ID_u,c)$ and $(ID_v,c)$. Note that a node can be in $O(\beta^2)$ anti-edges.
The total number of messages is $O(\beta^3)$; hence, can be disseminated to all nodes in $O(\log\Delta)$ rounds by \randompush (\cref{lem:randompush}). After this step, a node $v$ knows all candidate edges it belongs to. We run one extra \randompush for all nodes to know all candidate edges. By \cref{lem:ack-colorful-matching}, a colorful matching of size $K\avganti_C$ must exist, and nodes can find it with local computations.

Suppose now that the number of candidate edges is more than $U$. By the same argument as in \cref{lem:sparse-graph}, each node is contained in at most $2\gamma^2\beta^2$ candidate edges with high probability. Therefore, the colorful matching can be computed by a simple greedy algorithm from any set $F$ of at least $U$ candidate edges: start with an empty matching; as long as the matching has not size $K\avganti_C$, insert an arbitrary edges from $F$ into the matching and remove adjacent candidates edges from $F$. When we add an edge to the matching, we remove at most $4\gamma^2\beta^2$ edges from $F$. Since we assumed $F$ contained $U=4\gamma^2\cdot K\avganti_C\beta^2$ anti-edges, the algorithm always finds a colorful matching of $K\avganti_C$ edges. To select and disseminate a set $F$ of anti-edges, we select a subset $D$ of the colors, enough to have $U$ candidate edges but small enough to be able to disseminate the candidate edges with \randompush.
Using the same process as when the number of candidate edges is small, but using only colors of $D$, we can disseminate all selected candidate edges in $O(\log\Delta)$ rounds and compute the colorful matching locally. 

A simple recursive algorithm on the BFS tree spanning $C$ selects a subset $D$ of the colors such that the number of candidate edges with that color is at least $U$ and at most $O(\beta^3)$. Recall that each color has a unique leader which knows the number of candidate edges for its color. We say a subtree holds candidate edges with color $c$ if the leader for color $c$ belongs to this subtree. Note that when we compute the total number of candidate edges, each node learns the number of candidate edges held their subtree. Let $v$ be the root of the BFS tree spanning $C$ and $x_1, \ldots, x_t$ the number of candidate edges held by each subtree. Let $i$ be the smallest index in $[t]$ such that $\sum_{j\le i} x_j \ge U$. We select all colors whose leaders are in subtrees $0$ to $i-1$. We selected $\sum_{j<i} x_j$ candidate edges. We recursively run the algorithm on the $i$-th subtree to find $U-\parens*{\sum_{j < i} x_j}$ candidate edges. It is clear that we select at least $U$ candidate edges. We do not select more than $O(\beta^3)$ edges because each color group contains at most $O(\beta^2)$ edges. It is easy to see that the algorithm explore the tree top to bottom once as information can propagate independently in each subtree. In $O(\log\Delta)$ rounds each leader knows if its color was selected. Each leader relays the information to nodes of its group in $O(\log\Delta)$ rounds. At this point each $v\in C$ knows which color belongs to a selected candidate edge, i.e., colors such that $c\in \hS_v\cap D$.
\end{proof}

\section{Lower Bound}
\label{sec:lower}

As discussed in \Cref{sec:loweroverview}, at its core, our lower bound result is based on proving a lower bound for distributed computing a perfect matching in a random bipartite graph. More concretely, let $B$ be a bipartite graph on nodes $V\times C$, where $V$ models $n$ nodes and $C$ models $n$ colors. $B$ contain each of the $n^2$ possible edges between $V$ and $C$  with  probability $\poly\log n / n$. In the following, we show that computing a perfect matching of $B$ by a distributed message passing algorithm on $B$ requires $\Omega(\log n / \log\log n)$ rounds, even in the \LOCAL model (i.e., even if the nodes in $B$ can exchange arbitrarily large messages with their neighbors in $B$). We start with a simple observation regarding the structure of perfect matchings in bipartite graphs.

\begin{lemma}\label{lemma:bipmatching}
 Let $H=(V_E,E_H)$ be a bipartite graph, let $v_0\in V_H$ be a node of $H$, and for every integer $d\geq 0$, define $V_d\subset V_H$ be the set of nodes at distance exactly $d$ from $v_0$ in $H$. Then, if $H$ has a perfect matching, for every perfect matching $M$ of $H$ and for every $d\geq 0$, the number of edges of $M$ between nodes in $V_d$ and nodes in $V_{d+1}$ is equal to 
 \[
 S_d \eqdef \sum_{i=0}^i (-1)^i\cdot |V_d|.
 \]
\end{lemma}
\begin{proof}
  We prove the statement by induction on $d$. For $d=0$, we have $V_0=\set{v_0}$ and thus clearly the number of matching edges between $V_0$ and $V_1$ must be $S_0=|V_0|=1$. Let us, therefore, consider $d>0$ and assume that the statement holds for all $d'<d$. First note that for all $d>0$, we have $S_d = |V_d| - S_{d-1}$. Note that all neighbors of nodes in $V_d$ are either in $V_{d-1}$ or in $V_{d+1}$. Because every node in $V_d$ must be matched, the number of matching edges between $V_d$ and $V_{d+1}$ must be equal to $|V_d|$ minus the number of matching edges between $V_{d-1}$ and $V_d$. By the induction hypothesis, the number of matching edges between $V_{d-1}$ and $V_d$ is equal to $S_{d-1}$. The number of matching edges between $V_d$ and $V_{d+1}$ is therefore equal to $S_d=|V_d|-S_{d-1}$ as claimed.
\end{proof}

Note that \Cref{lemma:bipmatching} essentially states that the bipartite perfect matching problem is always a global problem in the following sense. In order to know the number of matching edges in a perfect matching between two sets $V_d$ and $V_{d+1}$, one must know the sizes of all the sets $V_0,\dots,V_d$.
As sketched in \Cref{sec:loweroverview}, we can use this observation to prove an $\Omega(\log n /\log\log n)$-round lower bound for computing a perfect matching in the random bipartite graph $B$. The formal details are given by the following theorem.

\begin{theorem}\label{thm:LOCALperfectmatching}
  Let $B=(V\cup C, E_B)$ be a random bipartite $2n$-node graph with $|V|=|C|=n$ that is defined as follows. For every $(v,c)\in V\times C$, edge $\set{v,c}$ is in $E_B$ independently with probability $p$, where $p\leq \poly\log(n)/n$ and $p\geq \alpha \ln(n)/n$ for a sufficiently large constant $\alpha>0$. Any distributed (randomized) \LOCAL algorithm that succeeds in computing a perfect matching of $B$ with probability at least $2/3$ requires at least $\Omega(\log n /\log\log n)$ rounds.
\end{theorem}
\begin{proof}
  First note that if $p\geq \alpha\ln(n)/n$ and the constant $\alpha$ is chosen sufficiently large, then $B$ has a perfect matching w.h.p. This is well-known~\cite[Section VII.3]{bollobas1998random} and can be seen by verifying Hall's condition.
  
  Let $T\eqdef\eta\cdot \ln n / \ln\ln n$ for a sufficiently small constant $\eta>0$ that will be determined later and assume that there exists a $T$-round randomized distributed perfect matching algorithm for the random graph $B$. We assume that after $T$ rounds, every node outputs its matching edge such that with probability $\geq 2/3$, the outputs of all nodes are consistent, i.e., the algorithm computes a perfect matching of $B$. Consider some node $v_0\in V\cup C$ and for every integer $d\geq 0$, let $V_d$ be the set of nodes at distance exactly $d$ from $v_0$. We next fix two parameters $\ell\eqdef T$ and $h\eqdef\ell+T+2$. To prove the lower bound, we concentrate on the nodes in $V_h$ and the computation of their matching edges. In a $T$-round algorithm, a node $v$ can only receive information from nodes within $T$ hops, and therefore the output of a node $v$ must be a function of the combination of the initial states of the nodes of the $T$-hop neighborhood of $v$ (when assuming that all the private randomness used by a node $v$ is contained in its initial state). Assume that the initial state of a node contains its ID, as well as the IDs of its neighbors. Then, $v$'s output of a $T$-round algorithm is a function of the subgraph induced by the $(T+1)$-hop neighborhood of $v$. The outputs of the nodes in $V_h$ therefore only depend on nodes in $V_d$ for $d\in \set{\ell+1,\dots,h + T+1}$ and on edges between those nodes. And it in particular means that nodes in $V_h$ do have to decide about their matching edges without knowing anything about nodes in $V_\ell$.
  
  In the following, we assume that nodes in $V_h$ can collectively decide about their matching edges. We further assume that to do this, the nodes in $V_h$ have the complete knowledge of the subgraph of $B$ induced by $(V_0\cup\dots\cup V_{\ell-1})\cup(V_{\ell+1}\cup\dots\cup V_{h+T+1})$. That is the nodes in $V_h$ have the complete knowledge of the graph induced by the nodes that are within distance $h+T+1=\ell+2T+3$ of $v_0$, with the exception of the nodes in $V_\ell$ and all their edges. Because we want to prove a lower bound, assuming coordination between the nodes  in $V_h$ and assuming knowledge of parts of the graph that are not seen by nodes in $V_h$ can only make our result stronger. Note that by \Cref{lemma:bipmatching}, the number of matching edges between level $V_h$ and $V_{h+1}$ is equal to $S_h=\sum_{i=0}^h  (-1)^i |V_{h-i}|$, which is an alternating sum that contains the term $|V_\ell|$ (either positively or negatively, depending on the parity of $h-\ell$). Hence, given the knowledge of the subgraph induced by $(V_0\cup\dots\cup V_{\ell-1})\cup(V_{\ell+1}\cup\dots\cup V_{h+T+1})$, the number of matching edges between nodes in $V_h$ and nodes in $V_{h+1}$ is in a one-to-one relation with the number of nodes in $V_\ell$. Without knowing $|V_\ell|$ exactly, the nodes in $V_h$ can therefore not compute their matching edges. Therefore, in order to prove the lemma, we need to prove that from knowing the subgraph induced by $(V_0\cup\dots\cup V_{\ell-1})\cup(V_{\ell+1}\cup\dots\cup V_{h+T+1})$, the size of $V_\ell$ can at best be estimated exactly with probability $2/3$.
  
  For this, we define several random variables. Let $X_\ell=|V_\ell|$ be the number of nodes in $V_\ell$, i.e., $X_\ell$ is the random variable that the nodes in $V_h$ need to estimate exactly in order to compute their matching edges. If we define $\cK$ to be a random variable that describes the knowledge that is provided to the nodes in $V_h$ to determine $X_\ell$, then we intend to estimate
  \[
  p_{\mathrm{est},K} \eqdef \max_{x_\ell>0}\Pr\big(X_{\ell}=x_\ell \,|\, \cK = K\big).
  \]
  Clearly, if $K$ is the actual state of the subgraph induced by $(V_0\cup\dots\cup V_{\ell-1})\cup(V_{\ell+1}\cup\dots\cup V_{h+T+1})$, the nodes in $V_h$ can determine the exact value of $X_\ell$ with probability at most $p_{\mathrm{est},K}$. To prove the theorem, we will show that with high probability, $\cK$ takes on a ``good'' state $K$ for which $p_{\mathrm{est},K}\leq 1/2+o(1)$. For estimating $X_{\ell}$ correctly, we then either need to have a ``bad'' $K$, which happens with probability $o(1)$ or we need to have a ``good'' $K$ and estimate $X_{\ell}$ correctly, which happens with probability $1/2+o(1)$. Overall, the probability for estimating $X_{\ell}$ correctly is then at best $1/2 + o(1) \leq 2/3$ for sufficiently large $n$. In order to estimate the probability of $X_{\ell}=x$, given the knowledge of the nodes in $V_h$, we first look at the conditioning on $\cK=K$ more closely. First note that by symmetry, the probability $\Pr(X_{\ell}=x | \cK = K)$ only depends on the topology of the subgraph induced by $(V_0\cup\dots\cup V_{\ell-1})\cup(V_{\ell+1}\cup\dots\cup V_{h+T+1})$ and not on the set of node IDs that appear in the part of the graph known by $V_h$. Further, the probability also does not depend on the edges of the induced subgraph known by $V_h$. The size of $X_{\ell}$ only depends on the additional edges of the nodes in $X_{\ell-1}$ and $X_{\ell+1}$. The probability $\Pr(X_{\ell}=x | \cK = K)$ therefore only depends on the sizes of the sets $V_0,\dots, V_{\ell-1}$ and $V_{\ell+1},\dots,V_{h+T+1}$.
  
  We first introduce the necessary random variables and some notation to simplify our calculation.
  For each $d\in\set{0,\dots,h+T+1}$, we define a random variable $X_d=|V_d|$. For convenience, for every $d$, we define $V_{\geq d}=V_d\cup V_{d+1}\cup \dots$ to be the set of nodes at distance at least $d$ from $v_0$. Throughout the calculations, we will concentrate on some fixed knowledge of the nodes in $V_h$. We therefore consider some  values $x_0, x_1, \dots, x_{h+T+1}$ and for each $d$, we define $\cX_d$ as a shortcut for the event $\set{X_d=x_d}$ that the random variable $X_d$ takes the value $x_d$. For convenience, we also define $\cX_{<d}\eqdef\cX_0 \cap \dots \cap \cX_{d-1}$, $\cX_{\leq d}=\cX_{<d}\cap \cX_d$, $\cX_{> d}\eqdef \cX_{d+1}\cap\dots\cap \cX_{h+T+1}$, as well as $x_{< d}\eqdef x_0+\dots+x_{d-1}$ and $x_{\leq d} = x_{<d} + x_d$. Note that for every $d>0$, if $V_0,\dots,V_{d-1}$ are fixed and no randomness of the edges connecting the remaining nodes $V_{\geq d}$ to $V_{d-1}\cup V_{\geq d}$ is revealed, then the size $X_d$ of $V_d$ is binomially distributed with parameters $|V_{\geq d}|$ and $q_d \eqdef 1- (1-p)^{x_{d-1}}$. For all $d\geq1$, we therefore have 
  \begin{equation}\label{eq:singleleveldistribution}
      \Pr(X_d=x_d \,|\, \cX_{<d}) = \binom{n-x_{<d}}{x_d}\cdot q_d^{x_d} 
      \cdot (1-q_d)^{n-x_{\leq d}}\text{, where }
      q_d \eqdef 1 - (1-p)^{x_{d-1}}.
  \end{equation}
  Let us first look at the probability of seeing a concrete assignment of values $x_0,\dots,x_{h+T+1}$ to the random variable $X_d$, including the value of $x_\ell$ for the random variable $X_\ell$ the nodes in $V_h$ need to estimate. By applying \eqref{eq:singleleveldistribution} iteratively, we obtain
  \begin{eqnarray}
  \Pr(\cX_{\leq h+T+1}) & = & \Pr(\cX_{<\ell} \land X_\ell=x_\ell \land \cX_{>\ell})\nonumber\\
  & = & \Pr(\cX_{<\ell})\cdot \prod_{i=\ell}^{h+T+1}\Pr(X_i=x_i\,|\, \cX_{<i})\nonumber\\
  & = & \Pr(\cX_{<\ell})\cdot \prod_{i=\ell}^{h+T+1}\binom{n-x_{<i}}{x_i} \cdot q_i^{x_i} \cdot (1-q_i)^{n-x_{\leq i}}.\label{eq:beforechange}
  \end{eqnarray}
  We next analyze what happens to the above probability if the number of nodes in $V_\ell$ is only $x_{\ell}-1$ instead of $x_\ell$. Our goal is to show that this only changes the probability by a $1-o(1)$ factor. If this is true for the most likely value $x_\ell$, this will imply that the nodes in $V_h$ can exactly estimate the value of $X_{\ell}$ at best with probability $1/2+o(1)$. To analyze the above probability if $X_\ell=x_\ell-1$, for all $d\geq 1$, we define the even $\cX_{<d}'$ as follows. If $d\leq \ell$, we have $\cX_{<d}'\eqdef\cX_{<d}$ and if $d> \ell$, we have $\cX_{<d}'\eqdef\cX_{<\ell}\cap\set{X_\ell=x_{\ell}-1}\cap\bigcap_{i=\ell+1}^{d-1}\set{X_i=x_i}$. We also define $\cX_{\leq d}'$ analogously. We then have
  \[
  \Pr(\cX_{\leq h+T+1}') = 
  \Pr(\cX_{<\ell})\cdot \Pr(X_{\ell}=x_{\ell}-1\,|\,\cX_{<\ell})\cdot
  \prod_{i=\ell+1}^{h+T+1}\Pr(X_i=x_i\,|\, \cX_{<i}').
  \]
  In order to compare $\Pr(\cX_{\leq h+T+1})$ and $\Pr(\cX_{\leq h+T+1}')$, we therefore need to compare $\Pr(X_{\ell}=x_{\ell}\,|\,\cX_{<\ell})$ and $\Pr(X_{\ell}=x_{\ell}-1\,|\,\cX_{<\ell})$, as well as $\Pr(X_i=x_i\,|\, \cX_{<i})$ and $\Pr(X_i=x_i\,|\, \cX_{<i}')$ for all $i\geq \ell+1$. We have
  \begin{eqnarray}
    \Pr(X_\ell=x_\ell-1\,|\, \cX_{<\ell}) & = &
    \binom{n-x_{<\ell}}{x_\ell-1}\cdot q_\ell^{x_\ell-1}\cdot (1-q_\ell)^{n-x_{\leq \ell}+1}\nonumber\\
    & = & 
    \frac{\binom{n-x_{<\ell}}{x_\ell-1}\cdot q_\ell^{x_\ell-1}\cdot (1-q_\ell)^{n-x_{\leq \ell}+1}}
    {\binom{n-x_{<\ell}}{x_\ell}\cdot q_\ell^{x_\ell}\cdot(1-q_\ell)^{n-x_{\leq \ell}}}\cdot \Pr(X_\ell=x_\ell \,|\, \cX_{<\ell})\nonumber\\
    & = & \frac{x_\ell \cdot (1-q_\ell)}{(n-x_{\leq \ell}+1)\cdot q_\ell} \cdot
    \Pr(X_\ell=x_\ell \,|\, \cX_{<\ell}).\label{eq:levelell}
  \end{eqnarray}
  For the following calculation, we define $q_{\ell+1}' = 1 - (1-p)^{x_\ell-1}$, i.e., $q_{\ell+1}'$ is the probability that a node outside $V_0\cup\dots\cup V_{\ell}$ is connected to $V_\ell$, if we assume that $X_\ell=x_\ell-1$ (instead of $X_\ell=x_\ell$).
  We obtain
  \begin{eqnarray}
    \lefteqn{\Pr(X_{\ell+1}=x_{\ell+1}\,|\, \cX_{<\ell+1}')\ =\ \Pr(X_{\ell+1}=x_{\ell+1}\,|\, \cX_{<\ell}\land X_\ell=x_\ell-1)}\nonumber\\
     & = &
    \binom{n - x_{\leq \ell}-1}{x_{\ell+1}}\cdot \big(q_{\ell+1}'\big)^{x_{\ell+1}}\cdot (1-q_{\ell+1}')^{n-x_{\leq \ell+1}+1}\nonumber\\
    & = & \frac{\binom{n - x_{\leq \ell}-1}{x_{\ell+1}}}{\binom{n - x_{\leq \ell}}{x_{\ell+1}}}\cdot 
    \left(\frac{q_{\ell+1}'}{q_{\ell+1}}\right)^{x_{\ell+1}}\cdot
    \left(\frac{1-q_{\ell+1}'}{1-q_{\ell+1}}\right)^{n-x_{\leq \ell+1}+1}\cdot
    \Pr(X_{\ell+1}=x_{\ell+1}\,|\, \cX_{<\ell}\land X_\ell=x_\ell)\nonumber\\
    & = & \frac{n-x_{\leq \ell} - x_{\ell+1}}{n - x_{\leq \ell}}\cdot 
    \left(\frac{q_{\ell+1}'}{q_{\ell+1}}\right)^{x_{\ell+1}}\cdot
    \left(\frac{1}{1-p}\right)^{n-x_{\leq \ell+1}+1}\cdot
    \Pr(X_{\ell+1}=x_{\ell+1}\,|\, \cX_{<\ell+1}).\label{eq:levelellplus1}
  \end{eqnarray}
  Finally, for $i>\ell+1$, we
  have
  \begin{eqnarray}
    \Pr(X_{i}=x_{i}\,|\, \cX_{<i}')
     & = &
    \binom{n - x_{<i}+1}{x_{i}}\cdot q_{i}^{x_{i}}\cdot (1-q_{i})^{n-x_{\leq i}+1}\nonumber\\
    & = & \frac{n-x_{<i}+1}{n-x_{\leq i}+1}\cdot (1-q_{i})\cdot 
    {\Pr(X_{i}=x_{i}\,|\, \cX_{<i})}.\label{eq:largelevels}
  \end{eqnarray}
  
  Recall that our goal is to show that if the random variables $X_1,\dots,X_{\ell-1}$ and $X_{\ell+1},\dots,X_{h+T+1}$ that are known to the nodes in $V_h$ are close enough to their expectation, then for all reasonable values $x_\ell$, when conditioning on the values of $X_1,\dots,X_{\ell-1}$ and $X_{\ell+1},\dots,X_{h+T+1}$, $X_\ell=x_\ell$ and $X_\ell=x_\ell-1$ have almost the same probability. We call an instance of the random graph in which the values of $X_1,\dots,X_{\ell-1}$ and $X_{\ell+1},\dots,X_{h+T+1}$ are close enough to their expectation \emph{well-behaved} and we formally denote this by an event $\cW$. We next define the event $\cW$ that specifies what it means that an instance is well-behaved.
  
  We assume that the probability $p$ that determines the presence  of the individual edges is equal to $p=f(n)/n$, where $f(n)\geq 32\ln n$ and $f(n)\leq \polylog n$. The event $\cW$ is defined as follows. For all $d\in \set{1,\dots,h+T+1}$, it must hold that
  \begin{equation}\label{eq:wellbehaved}
  |X_d - f(n)\cdot X_{d-1}| \leq \sqrt{7\cdot f(n)\cdot X_{d-1} \cdot\ln n}.
  \end{equation}
  Note that condition \eqref{eq:wellbehaved} and the assumption that $f(n)\geq 32\ln n$ directly imply that $X_d \leq 1.5 f(n) \cdot X_{d-1}$ (even if $X_{d-1}=1$). Therefore in well-behaved instances, $X_d/X_{d-1}$ is at most $\polylog n$. We choose the parameter $T=\Theta(\ln(n)/\ln\ln(n))$ small enough such that in well-behaved instances, $X_0 + \cdots + X_{h + T+ 1}\leq n^{1/3}$. Similarly, \eqref{eq:wellbehaved} implies that $X_d \geq 0.5 f(n) \cdot X_{d-1}$ and thus for any $d=\Theta(\log n /\log\log n)$, we have $X_d\geq n^{\nu}$ for some constant $\nu>0$. We next show that a given instance (i.e., the neighborhood of a fixed node $v_0$ in a given random bipartite graph $B$) is well-behaved with probability $>1-1/n$.
  
  To see this, consider a given value $d\geq 1$. We know that once $X_1=x_1,\dots,X_{d-1}=x_{d-1}$ is given, $X_d$ is binomially distributed with parameters $n-x_{<d}$ and $q_d=1-(1-p)^{x_{d-1}}$. By a standard Chernoff bound, for any $\delta\in[0,1]$, we therefore know that
  \begin{equation}\label{eq:wellbehavedChernoff}
  \Pr(|X_d - \Exp[X_d|\cX_{<d}]| > \delta\cdot\Exp[X_d|\cX_{>d}] \,|\, \cX_{<d}) \leq 2 e^{-\delta^2/3\cdot \Exp[X_d]}.
  \end{equation}
  We will see that \eqref{eq:wellbehavedChernoff} implies that for every $d$, Inequality \eqref{eq:wellbehaved} holds with probability at least $1-2/n^2$. To achieve this, we first have to understand what the value of $\Exp[X_d|\cX_{<d}]$ is. We first have a look at the probability $q_d$ of the binomial distribution underlying $X_d$. We have
  \begin{eqnarray*}
  q_d & = & 1-(1-p)^{x_{d-1}} \leq p\cdot x_{d-1}\quad\text{and}\\
  q_d & = & 1-(1-p)^{x_{d-1}} \geq 1 - e^{-p x_{d-1}} \geq
  p\cdot x_{d-1} - (p\cdot x_{d-1})^2 \geq p\cdot x_{d-1}\cdot
  \left(1 - \frac{f(n)}{n^{2/3}}\right).
  \end{eqnarray*}
  We used that for any real value $y\in[0,1]$ and any $k\geq 1$, we have $(1-y)^k\geq 1-ky$ and $1-y \leq e^{-y} \leq 1 - y + y^2$. In the last inequality, we have further used that in well-behaved executions, $x_{d-1}\leq n^{1/3}$. We have $\Exp[X_d|\cX_{<d}]=q_d\cdot (n-x_{<d})$. We therefore have 
  \begin{eqnarray*}
    \Exp[X_d|\cX_{<d}] & \leq & p x_{d-1} \cdot n = f(n) \cdot x_{d-1},\\
    \Exp[X_d|\cX_{<d}] & \geq & \left(1 - \frac{f(n)}{n^{2/3}}\right)\cdot
    p x_{d-1} \cdot (n-n^{1/3}) \geq 
    \left(1- \frac{2f(n)}{n^{2/3}}\right)\cdot f(n)\cdot x_{d-1}.
  \end{eqnarray*}
  Let $\xi$ denote the maximum absolute deviation from the expectation that is still guaranteed to be well-behaved by \eqref{eq:wellbehaved}. We can lower bound $\xi$ as follows
  \begin{eqnarray*}
  \xi & \geq & \sqrt{7 f(n) \cdot x_{d-1} \cdot \ln n} -
  \frac{2f(n)}{n^{2/3}}\cdot f(n)\cdot x_{d-1}\\
  & \geq & \sqrt{7 \Exp[X_d | \cX_{<d}] \cdot \ln n} - \frac{2f^2(n)}{n^{1/3}}\\
  & \geq &
  \sqrt{6 \Exp[X_d | \cX_{<d}] \cdot \ln n}\ =\ 
  \sqrt{\frac{6\ln n}{\Exp[X_d | \cX_{<d}]}}\cdot \Exp[X_d | \cX_{<d}].
  \end{eqnarray*}
  The last inequality holds if $n\geq n_0$ for a sufficiently large constant $n_0$. Together with \eqref{eq:wellbehavedChernoff}, this now implies that for all $d$, \eqref{eq:wellbehaved} holds with probability larger than $1-2/n^2$. Note that there are less only $O(\log n / \log\log n)$ different $d$-values. If $n\geq n_0$ for a sufficiently large constant $n_0$, the number of different $d$-values can therefore for example be upper bounded by $\ln(n)/2$. In this case, a union bound over all $d$-values implies that the probability that an instance is well-behaved is at least $1-\ln(n)/n^2$.
  
  Let us now assume that we have a well-behaved instance (i.e., that $\cW$ holds). Consider some assignment to the random variables $X_1,\dots,X_{\ell-1}$ and $X_{\ell+1},\dots,X_{h+T+1}$ that are consistent with \eqref{eq:wellbehaved}. Given the values of those random variables, the nodes in $V_h$ have to guess the value of $X_h$. Note that there are extreme cases where the values of $X_1,\dots,X_{\ell-1}$ and $X_{\ell+1},\dots,X_{h+T+1}$ only allow one single value of $X_\ell$ such that \eqref{eq:wellbehaved} is satisfied. We need to show that even when conditioning on $\cW$, this only happens with a very small probability. Let us therefore define an even $\cW'\subseteq \cW$ as an instance in which replacing $X_{\ell}$ by $X_{\ell}-1$ or $X_{\ell}+1$ still satisfies \eqref{eq:wellbehaved}. Note that the above analysis has enough slack to ensure that also $\Pr(\cW')\geq 1-\ln(n)/n^2$ if $n\geq n_0$ for a sufficiently large constant $n_0$. The same analysis for example also works if the fixed constant $7$ in \eqref{eq:wellbehaved} is replaced by any smaller fixed constant that is larger than $6$. We therefore have
  \[
  \Pr(\cW'\,|\,\cW) = \frac{\Pr(\cW'\cap \cW)}{\Pr(\cW)} =
  \frac{\Pr(\cW')}{\Pr(\cW)} \geq \Pr(\cW') \geq 1-\frac{\ln n}{n^2}.
  \]
  
  We use $x_1,\dots,x_{\ell-1}$ and $x_{\ell+1},\dots,x_{h+T+1}$ to denote the concrete values of those random variables. We further define $x_\ell$ to be an arbitrary value such that the values $x_{\ell}$ and $x_{\ell}-1$ are both valid values for $X_\ell$ to make the instance well-behaved. Note that if $\cW'$ holds, there is at least one such value $x_{\ell}$ and we have seen that even conditioning on $\cW$, the probability of $\cW'$ is still at least $1-\ln(n)/n^2$. Because we are assuming $\cW$, we know that the value $x_\ell$ satisfies the following condition. 
  \[
  x_\ell = x_{\ell-1}\cdot np + \theta,\ \ \text{where}\ \ 
  |\theta| \in O(\sqrt{x_{\ell-1}f(n)\log n}) = 
  O\left(\sqrt{\frac{f(n)\log n}{x_{\ell-1}}}\right)\cdot x_{\ell-1}.
  \]

  We first look at the ratio between $\Pr(X_\ell=x_\ell-1\,|\, \cX_{<\ell})$ and $\Pr(X_\ell=x_\ell \,|\, \cX_{<\ell})$. By Equation \eqref{eq:levelell}, this ratio is equal to $\frac{x_\ell \cdot (1-q_\ell)}{(n-x_{\leq \ell}+1)\cdot q_\ell}$. In the following, we denote this ratio by $\rho_1$. By using that $x_{\leq h+T+1}\leq n^{1/3}$, we can then bound $\rho_1$ as follows.
  \[
  \rho_1=\frac{x_\ell \cdot (1-q_\ell)}{(n-x_{\leq \ell}+1)\cdot q_\ell} \leq
  \frac{f(n)\cdot x_{\ell-1}\cdot\big(1 + O\big(\sqrt{f(n)\log(n)/x_{d-1}}\big)\big)}
  {\left(1-\frac{1}{n^{2/3}}\right)\cdot n \cdot p x_{\ell-1}\cdot (1-p x_{\ell-1})}\leq 1 + \frac{1}{n^{\Omega(1)}}.
  \]
  In the last inequality, we used that $x_\ell = x_{\ell-1}\cdot np$, that $x_{\ell-1}\leq n^{1/3}$, and that $x_{\ell-1}\geq n^{\kappa}$ for some constant $\kappa>0$. Similarly, we have
  \[
  \rho_1=\frac{x_\ell \cdot (1-q_\ell)}{(n-x_{\leq \ell}+1)\cdot q_\ell} \geq
  \frac{x_{\ell-1}\cdot\big(1 - O\big(\sqrt{f(n)\log(n)/x_{d-1}}\big)\big)\cdot (1-px_{\ell-1})}{n \cdot p x_{\ell-1}}\geq 1 - \frac{1}{n^{\Omega(1)}}.
  \]
  
  We next look at the ratio between $\Pr(X_{\ell+1}=x_{\ell+1}\,|\, \cX_{<\ell+1}')$ and $\Pr(X_{\ell+1}=x_{\ell+1}\,|\, \cX_{<\ell+1})$. We denote this ratio by $\rho_2$. By Equation \eqref{eq:levelellplus1}, $\rho_2$ can be written as
  \begin{eqnarray*}
    \rho_2 & = & \frac{n-x_{\leq \ell} - x_{\ell+1}}{n - x_{\leq \ell}}\cdot 
    \left(\frac{q_{\ell+1}'}{q_{\ell+1}}\right)^{x_{\ell+1}}\cdot
    \left(\frac{1}{1-p}\right)^{n-x_{\leq \ell+1}+1}\\
    & \leq & 
    \left(\frac{p\cdot (x_{\ell}-1)}{p\cdot x_{\ell} \cdot (1- p\cdot x_{\ell})}\right)^{x_{\ell+1}}\cdot
    \left(\frac{1}{1-p}\right)^{n-x_{\leq \ell+1}+1}\\
    & \leq &
    \left(1+\frac{O(f(n))}{n^{1/3}}\right)\cdot 
    \left(1 - \frac{1}{x_{\ell}}\right)^{x_{\ell+1}}\cdot 
    \left(1 + \frac{p}{1-p}\right)^{n}\\
    & \leq & 
    \left(1+\frac{O(f(n))}{n^{1/3}}\right)\cdot 
    e^{-\frac{x_{\ell+1}}{x_{\ell}}}\cdot
    e^{np + O(np^2)}\\
    & \leq & 
    \left(1+\frac{O(f(n))}{n^{1/3}}\right)\cdot 
    e^{np - np + O(\sqrt{f(n)\log(n)/x_{\ell}})}\ \leq\ 
    1+ \frac{1}{n^{\Omega(1)}}.
  \end{eqnarray*}
  In the last inequality, we used that because $\cW$ holds, we have $x_{\ell+1} \leq x_\ell \cdot np + O(\sqrt{f(n) x_\ell \log n})$ and that $x_{\ell}\geq n^\nu$ for some constant $\nu>0$. We can similarly lower bound $\rho_2$ as follows.
  \begin{eqnarray*}
    \rho_2 & \geq &
    \frac{n-n^{1/3}}{n}\cdot
    \left(\frac{p(x_{\ell}-1)(1-p(x_{\ell}-1))}{p\cdot x_{\ell}}\right)^{x_{\ell+1}}
    \cdot\left(\frac{1}{1-p}\right)^{n - n^{1/3}}\\
    & \geq &
    \left(1-\frac{O(f(n))}{n^{1/3}}\right)\cdot 
    \left(1 - \frac{1}{x_{\ell}}\right)^{x_{\ell+1}}\cdot 
    (1 + p)^{n}\\
    & \geq &
    \left(1-\frac{O(f(n))}{n^{1/3}}\right)\cdot 
    e^{-\frac{x_{\ell+1}}{x_{\ell}}}\cdot \left(1-\frac{1}{x_{\ell}^2}\right)^{x_{\ell+1}}\cdot
    e^{pn}\cdot(1-p^2)^n\\
    & \geq &
    \left(1-\frac{O(f(n))}{n^{1/3}}\right)\cdot 
    \left(1- \frac{x_{\ell+1}}{x_{\ell}^2}\right)\cdot
    (1-p^2n)\cdot e^{np - np - O(\sqrt{np\log(n)/x_{\ell})}}\\
    & \geq &
    \left(1-\frac{O(f(n))}{n^{1/3}}\right)\cdot 
    \left(1-\frac{O(f(n))}{x_{\ell}}\right)\cdot
    \left(1-O\left(\frac{\sqrt{np\log(n)}}{\sqrt{x_{\ell}}}\right)\right)\ \geq\
    1- \frac{1}{n^{\Omega(1)}}.
  \end{eqnarray*}
  In the third inequality, we use that for $y\in [-1,1]$, it holds that $1+y\geq e^y\cdot (1-y^2)$. In the last inequality, we again used that $x_{\ell-1}$ and $x_{\ell}$ are both of size $\geq n^{\nu}$ for some constant $\nu>0$.
  
  Finally, let us look at the ratio between the probabilities $\Pr(X_{i}=x_{i}\,|\, \cX_{<i}')$ and $\Pr(X_{i}=x_{i}\,|\, \cX_{<i})$ for $i>\ell+1$. We denote this ratio by $\rho_{3,i}$, and by using Equation \eqref{eq:largelevels}, we can bound $\rho_{3,i}$ as follows.
  \[
  1-O\left(\frac{f(n)}{n^{2/3}}\right) \leq 
  \rho_{3,i} = \frac{n-x_{<i}+1}{n-x_{\leq i}+1}\cdot (1-q_i) \leq
  1+O\left(\frac{f(n)}{n^{2/3}}\right).
  \]
  We therefore have
  \[
  \frac{\Pr(\cX_{\leq h + T + 1}'\,|\,\cW')}
  {\Pr(\cX_{\leq h + T + 1})\,|\,\cW')}\ =\ 
  \rho_1\cdot\rho_2\cdot \prod_{i=\ell+2}^{h+T+1} \rho_i\ =\ 
  1 \pm o(1).
  \]
  Recall that $\cX_{\leq h + T + 1}=\cX_{<\ell}\cap \set{X_\ell=x_\ell}\cap \cX_{>\ell}$ and $\cX_{\leq h+T+1}'=\cX_{<\ell}\cap \set{X_\ell=x_\ell-1}\cap \cX_{>\ell}$. We therefore have
  \begin{eqnarray*}
    \Pr(\cX_{\leq h + T + 1}\,|\,\cW')
    & = &
    \Pr(X_\ell=x_\ell\,|\,\cX_{<\ell}\cap\cX_{>\ell}\cap\cW')\cdot
    \Pr(\cX_{<\ell}\cap\cX_{>\ell}\,|\,\cW')\quad\text{and}\\
    \Pr(\cX_{\leq h + T + 1}\,|\,\cW')
    & = &
    \Pr(X_\ell=x_\ell-1\,|\,\cX_{<\ell}\cap\cX_{>\ell}'\cap\cW')\cdot
    \Pr(\cX_{<\ell}\cap\cX_{>\ell}'\,|\,\cW')
  \end{eqnarray*}
  and thus
  \[
  \frac{\Pr(X_\ell=x_\ell-1\,|\,\cX_{<\ell}\cap\cX_{>\ell}'\cap\cW')}
  {\Pr(X_\ell=x_\ell\,|\,\cX_{<\ell}\cap\cX_{>\ell}\cap\cW')}\ =\ 
  \frac{\Pr(\cX_{\leq h + T + 1}'\,|\,\cW')}
  {\Pr(\cX_{\leq h + T + 1})\,|\,\cW')}\ =\ 
  1\pm o(1).
  \]
  However, this means that if $\cW'$ holds, the interval of possible values for $X_\ell$ contains at least two values and for any two adjacent values, the conditional probability that this is the correct guess is equal up to a $1\pm o(1)$ factor. Hence, even for the possible value $x_\ell^*$ that maximizes $\Pr(X_\ell=x_\ell^*\,|\,\cX_{<\ell}\cap\cX_{>\ell}\cap\cW')$, we have
  $\Pr(X_\ell=x_\ell^*\,|\,\cX_{<\ell}\cap\cX_{>\ell}\cap\cW')\leq 1/2 + o(1)$. Thus, if $\cW'$ holds, the nodes in $V_h$ exactly estimate $X_\ell$ with probability better than $1/2+o(1)$. Because $\cW'$ holds with probability $1-o(1)$, even if the nodes in $V_h$ always succeed in case $\cW'$ does not hold, the probability that $V_h$ can correctly guess $X_\ell$ is still at best $1/2+o(1)$. If the number of nodes $n\geq n_0$ for a sufficiently large constant $n_0$, this is at most $2/3$, which proves the claim of the theorem.
\end{proof}

We note that the success probability of $1/3$ could be boosted significantly in several ways. First, note that it would not be hard to adapt the proof so that for some constant $\nu>0$, $\cW$ allows $n^{\nu}$ different values for $X_{\ell}$ and that the probabilities for the $n^{\nu}$ most likely values are all approximately the same. This reduces the success probability to $n^{-\Omega(1)}$. Further, instead of looking at one neighborhood in the graph, we could look at polynomially many independent and disjoint neighborhoods and thus make the success probability even exponentially small in $n^{\nu}$ for some constant $\nu>0$.

Given the lower bound on computing a perfect matching in a random bipartite graph, our main lower bound theorem now follows in a relatively straightforward fashion. The following is a more precisely phrased version of \Cref{thm:intro-lowerbound}.

\begin{theorem}\label{thm:lowerbound}
    Assume that each node $v$ of a complete graph $K_n$ on $n$ nodes uniformly and independently computes a subset $S_v$ of the colors $\set{1,\dots,n}$ as follows. Each color $x$ is included in $S_v$ independently with probability $f(n)/n$, where $f(n)\geq c\ln(n)$ for a sufficiently large constant $c$ and $f(n)\leq \poly\log n$. Let $G$ be the subgraph of $K_n$ defined by all $n$ nodes and the set of edges between nodes $u$ and $v$ with $S_u\cap S_v\neq \emptyset$. Any randomized \LOCAL algorithm on $G$ to properly color $K_n$ with colors from the sets $S_v$ requires $\Omega(\log n / \log\log n)$ rounds. The lower bound holds even if the algorithm only has a success probability of $2/3$.
\end{theorem}
\begin{proof}
  We define a bipartite graph $B$ between the set of nodes $V=\set{1,\dots,n}$ and the set of color $C=\set{1,\dots,n}$. There is an edge between $v\in V$ and $x\in C$ iff $x\in S_v$. We note that since for every color $x$ and every node $v$, $\Pr(x\in S_v)=f(n)/n$ and those probabilities are independent for different pairs $(v,x)$, the bipartite graph on $V$ and $C$ contains each possible edge between $V$ and $C$ independently with probability $p=f(n)/n$. Further, a valid $n$-coloring of $K_n$ is a one-to-one assignment between nodes and colors. Therefore, each valid $n$-coloring of $K_n$ that respects the sampled color set corresponds to a perfect matching in the bipartite graph $B$ between $V$ and $C$ and vice versa. Also note that clearly, in the \LOCAL model, any \LOCAL algorithm on $B$ can be run on $G$ with only constant overhead, and vice versa (when simulating $B$ on $G$, each color node $x$ has to be simulated by one of the nodes $v$ for which $x\in S_v$). Hence, any distributed coloring algorithm for $K_n$ that runs on the sampled graph $G$ implies a perfect matching algorithm on $B$ with the same asymptotic round complexity. The theorem therefore directly follows from \Cref{thm:LOCALperfectmatching}.
\end{proof}


\bibliography{arxiv.bbl}

\appendix

\section{Concentration Bounds}

\paragraph{Some useful inequalities.} We use the following classic inequalities:

\begin{lemma}[{\cite{Doerr2020}}]
\label{lem:useful-inequalities}
For $x\in [0,1]$ and $y > 0$, we have 
\begin{equation}\
\label{eq:exp-ineq}
1 - x \le e^{-x} \le 1-\frac{x}{2}\qquad and \qquad (1-x)^y \le \frac{1}{1+xy}.
\end{equation}
\end{lemma}

\paragraph{Chernoff bound with domination.} The classical version of the Chernoff bound shows concentration for sum of binary random variables and assumes independence between each variable. We use a more general form allowing for some dependencies and non-binary variables. 

\begin{lemma}[Martingales]\label{lem:chernoff}
Let $\{X_i\}_{i=1}^r$ be random variables distributed in $[0,1]$, and $X=\sum_i X_i$.
Suppose that for all $i\in [r]$ and $(x_1,\ldots,x_{i-1})\in \{0,1\}^{i-1}$ with $\Pr\range{X_1=x_1,\dots,X_r=x_{i-1}}>0$, $\Pr\range{X_i=1\mid X_1=x_1,\dots,X_{i-1}=x_{i-1}}\le q_i\le 1$, then for any $\delta>0$,
\begin{equation}\label{eq:chernoffless}
\Pr\range*{X\ge(1+\delta)\sum_{i=1}^r q_i}
\le \exp\parens*{-\frac{\min(\delta,\delta^2)}{3}\sum_{i=1}^r q_i}\ .
\end{equation}
Suppose instead that $\Pr\range*{X_i=1\mid X_1=x_1,\dots,X_{i-1}=x_{i-1}}\ge q_i$, $q_i\in (0,1)$ holds for $i, x_1, \ldots, x_{i-1}$ over the same ranges, then for any $\delta\in [0,1]$,
\begin{equation}\label{eq:chernoffmore}
    \Pr\range*{X\le(1-\delta)\sum_{i=1}^r q_i}\le \exp\left(-\frac{\delta^2}{2}\sum_{i=1}^r q_i\right)\ .
\end{equation}
\end{lemma}

\paragraph{Talagrand inequality.}
A function $f(x_1,\ldots,x_n)$ is  \emph{$c$-Lipschitz} iff changing any single $x_i$ affects the value of $f$ by at most $c$, and $f$ is  \emph{$r$-certifiable} iff whenever $f(x_1,\ldots,x_n) \geq s$ for some value $s$, there exist $r\cdot s$ inputs $x_{i_1},\ldots,x_{i_{r\cdot s}}$ such that knowing the values of these inputs certifies $f\geq s$ (i.e., $f\geq s$ whatever the values of $x_i$ for $i\not \in \{i_1,\ldots,i_{r\cdot s}\}$).

\begin{lemma}[Talagrand's inequality~\cite{Talagrand95,DP09}]
\label{lem:talagrand}
Let $\{X_i\}_{i=1}^n$ be $n$ independent random variables and $f(X_1,\ldots,X_n)$ be a $c$-Lipschitz $r$-certifiable function; then for $t\geq 1$,
\[\Pr\range*{\abs*{f-\Exp[f]}>t+30c\sqrt{r\cdot\Exp[f]}}\leq 4 \cdot \exp\parens*{-\frac{t^2}{8c^2r\Exp[f]}}\]
\end{lemma}

In the next lemma, $\mathbb{I}_X$ denotes the indicator random variable of an event $X$.

\begin{lemma}[\cite{HKNT22}]
\label{lem:talagrand-difference}
Let $\set*{X_i}_{i=1}^n$ be $n$ independent random variables. Let $\set*{A_j}_{j=1}^k$ and $\set*{B_j}_{j=1}^k$ be two families of
events that are functions of the $X_i$'s. Let $f=\sum_{j\in[k]} \mathbb{I}_{A_j}$, $g=\sum_{j\in[k]} \mathbb{I}_{A_j \cap \overline{B}_j}$, and $h=f-g$ be such that $f$ and $g$ are $c$-Lipschitz and $r$-certifiable w.r.t.\ the $X_i$'s, and $\Exp[h] \geq \alpha \Exp[f]$ for some constant $\alpha \in (0,1)$. Let $\delta \in (0,1)$. Then for $\Exp[h]$ large enough:
\[\Pr\range*{\abs*{h - \Exp[h]} > \delta \Exp[h]} \leq \exp(-\Omega(\Exp[h])).\]
\end{lemma}

\section{Omitted Proofs}

\subsection{Computing the Almost-Clique Decomposition}
\label{sec:ACD}

In this section, we show the following lemma:

\ACD*

\begin{Algorithm}
\label{alg:ACD}
Algorithm computing a $\epsilon$-almost-clique decomposition.

\textbf{Parameters.} Define 
\[ \delta=\frac{\epsilon}{12},\qquad
\lambda = \frac{16\Delta}{\delta}\qquad\text{and}
\qquad\sigma=\frac{384\beta}{\delta^4}. \]

\textbf{When sparsifiying the input graph.} Each node $v$ samples a value $r(v)\in [\lambda]$ uniformly at random.
Each node $v$ then computes
\begin{itemize}
\item the set $F(v)$ containing all values $r(u)$ for neighbors $u\in N(v)$ such that $r(v)\le\sigma$.
\item a set $E_s(v)$ of $O(\log n/\delta^2)$ random edges using reservoir sampling.
\end{itemize}

\textbf{Communication Phase.} Each $v$ performs the following algorithm:
\begin{enumerate}
\item Send $F(v)$ to each neighbor in $E_s(v)$. 

If $|F(v)\cap F(u)| \ge (1-\delta)\Delta\sigma/\lambda$, then $u$ and $v$ are \emph{friends}.
\item By counting its number of friends in $E_s(v)$, $v$ learns if it is \emph{popular}.
\end{enumerate}
\end{Algorithm}

\begin{definition}[Friendly edges]
\label{def:friend-edges}
For any $\delta\in(0,1)$, we say that nodes $u$ and $v$ are \emph{friends} if they are connected, i.e., $uv\in E$, and share a $(1-\delta)$-fraction of their neighborhood, i.e., $|N(u)\cap N(v)| \ge (1-\delta)\Delta$.
\end{definition}

To detect friendly edges, the approach of \cite{ACK19} was to sample nodes with probability $O(\frac{\log n}{\delta^2\Delta})$ and, for each sampled edge $uv$, compare the set of nodes sampled in $N(v)$ to that of $N(u)$. This approach requires nodes to communicate $O(\log^2 n)$ bits with their neighbors ($O(\log n)$-bits identifiers for $O(\log n)$ sampled neighbors); hence, it exceeds the bandwidth requirements of our model.

In recent work, \cite{HNT22} proposed a \congest algorithm for solving this task in $O(1)$ rounds. They devise an algorithm (\cite[Algorithm 1]{HNT22}) using families of pseudo-random hash functions to estimate up to $\delta\Delta$ precision the similarities of two $\Theta(\Delta)$-sized sets. They observed it could be used to compute ACD in a rather straightforward way by using this primitive to compare neighborhoods. The main obstacle to implement this algorithm in our model are memory constraints: their families of hash function are non-constructive and of size $\poly(n)$.

We note, however, that for this specific use, we can sample \emph{a truly random function}. To sample a random function $r$ mapping nodes to values in $[\lambda]$, it is enough if each node $v$ samples a value $r(v)\in[\lambda]$ independently. By \cite[Claim 1]{HNT22}, the induced random function has few enough collisions for nodes to estimate the size of their shared neighborhood with sufficient accuracy with high probability. Furthermore, to do so, nodes need only to know the hash value of their neighbors. The parameters of \cref{alg:ACD} are set to match the ones of \cite[Algorithm 1]{HNT22} with up to $(\delta/2)\Delta$ error.

\begin{lemma}[{Detecting Friendly Edges, \cite[Claim 1 + Lemma 2]{HNT22}}]
\label{lem:friendly}
Let $\delta \in (0, 1/10)$. For every pair of adjacent nodes $uv$, with high probability, we have that 
\begin{itemize}
    \item if $u$ and $v$ are $\delta$-friends, we have $|F(v)\cap F(u)| \ge (1-1.5\delta)\Delta\sigma/\lambda$; and
    \item if $u$ and $v$ are not $2\delta$-friends, we have $|F(v)\cap F(u)| < (1-1.5\delta)\Delta\sigma/\lambda$.
\end{itemize}
\end{lemma}

The other primitive required to compute the almost-clique decomposition is for distinguishing between $\delta$-popular nodes and those that are not $2\delta$-popular.

\begin{definition}[Popular Nodes]
\label{def:popular}
For any $\delta\in (0,1/10)$, we say $u$ is \emph{$\delta$-popular} if it has $(1-\delta)\Delta$ friendly edges.
\end{definition}

The following lemma states that by sampling edges with probability $\Theta(\log n/\delta^2\Delta)$ edges in its neighborhood, a node can distinguish between it being $\delta$-popular and it not being $2\delta$-popular. It follows directly from the Chernoff bound (\cref{lem:chernoff}) as the number of sampled edges allows us to estimate w.h.p.\ the number of friendly edges up to $(\delta/2)\Delta$ by sampling.

\begin{lemma}[Detecting Popular Nodes]
\label{lem:popular}
Let $\delta\in (0,1/10)$. If edges are sampled in $E_s$ with probability $p=\Theta(\log n/(\delta^2\Delta))$, then, with high probability, for every node $u$, we have that
\begin{itemize}
    \item if $u$ is $\delta$-popular, it samples at least $(1-1.5\delta)\Delta p$ $\delta$-friendly edges in $E_s(u)$;
    \item if $u$ is not $2\delta$-popular, it samples fewer than $(1-1.5\delta)\Delta p$ $2\delta$-friendly edges in $E_s(u)$.
\end{itemize}
\end{lemma}

\cref{lem:friendly,lem:popular} are sufficient to find a $\delta$-almost-clique decomposition.

\begin{lemma}[\cite{ACK19}]
\label{lem:find-decomp}
Let $H$ be the subgraph of $G$ with $2\delta$-popular nodes. Let $C_1, \ldots, C_t$ be the connected components of $H$ with at least one $\delta/2$-popular node and $\Vsparse=V\setminus\bigcup_{i\in[t]} C_i$. This decomposition is a $12\delta$-almost-clique decomposition.
\end{lemma}

We are now ready to prove \cref{lem:acd}.

\begin{proof}[Proof of \cref{lem:acd}]
Let $\Vsparse, C_1, \ldots, C_t$ be the decomposition described in \cref{lem:find-decomp}. Consider a cluster $C_i$ for some $i\in [t]$. By \cref{lem:popular}, nodes can tell if they are $\delta/2$-popular. Moreover, the subgraph  $\tilde{H}(C)$ of $C_i$, consisting of the sampled edges in $E_s(\cdot)$, is a random graph where edges are sampled with probability $O(\log n/\Delta)$. This means that $\tilde{H}(C)$ has a constant rate vertex expansion (see \cref{lem:expansion} for a proof of a similar fact). Therefore, in $O(\log\Delta)$ rounds, every node $v$ in the connected component $C_i$ knows it belongs to an almost-clique, as well as the identifier of the $\delta/2$-popular node in $C_i$ with minimal ID (used as identifier for the clique) and which edges in $E_s(v)$ are connecting it to $C_i$.
\end{proof}

\subsection{Preconditioning Almost-Clique}
\label{sec:strong-ACD}

\strongACD*

Assume we computed an $\epsilon'$-almost-clique decomposition using \cref{alg:ACD} for $\epsilon'=\epsilon/3$. In this section, we use this decomposition to compute the partial coloring described in \cref{thm:strong-ACD}.

Sparse nodes receive slack after a single randomized color trial. Intuitively, this happens because each non-edge in the neighborhood of a $\zeta$-sparse node $v$ has both its endpoints colored the same with probability $\Omega(1/\Delta)$. Since it has $\Delta\zeta$ such non-edges in $E(N(v))$ (see \cref{def:sparsity}), it receives $\Omega(\zeta)$ slack in expectation. Formally, this gives the following lemma.

\begin{lemma}[{\cite[Lemma 6.3]{HKMT21}}]
\label{lem:slackgen}
Let $v$ be a $\zeta$-sparse node. After a random fraction of its neighbors try colors, it has slack $\Omega(\zeta)$ with probability $1-e^{-\Omega(\zeta)}$. Furthermore, if $v$ is a dense node, it receives slack $\Omega(e_v)$ with probability $1-e^{-\Omega(e_v)}$.
\end{lemma}

In \cref{thm:strong-ACD}, we want to get rid of high external degree nodes. By \cref{lem:slackgen}, if a node has a high external degree, it should also have a lot of slack. \cref{alg:SACD} carefully generates slack to color all nodes of high external degree.

First, we claim that high-external-degree nodes can be easily detected by randomly sampling edges.

\begin{claim}
\label{claim:detect-extrovert}
There is an algorithm partitioning the dense nodes into two classes: \emph{extroverted} nodes of external degree at most $\Delta/\eta$, and \emph{introverted} nodes of external degree at least $\Delta/(2\eta)$.
The algorithm samples $O(\eta\log n)$ edges per node.
\end{claim}

\begin{proof}
Let $\emax=\Delta/\eta$.
During the streaming phase, a node samples edges with probability $p\eqdef\frac{\eta\beta}{\Delta}$. Once the nodes have computed the almost-clique decomposition, they know which edges connect them to external neighbor. If a node sampled fewer than $0.75\beta$ edges to external neighbors, it classify itself as introvert; otherwise, as extrovert.
\begin{itemize}
\item Consider a node with external degree at most $\emax/2$. In expectation, it samples $p\emax/2 \le \beta/2$ edges to external neighbors. By Chernoff, it samples fewer than $0.75\beta$ edges with high probability. 
Nodes with external degree less than $\emax/2$ are classified as introverts.
\item Consider an extroverted nodes, i.e., with external degree at least $\emax$. In expectation, it samples at least $p\emax \ge \beta$ edges to external neighbors. By Chernoff, it samples at least $0.75\beta$ edges with high probability. All nodes with external degree more than $\emax$ are classified as extrovert, w.h.p.
\end{itemize}
Nodes with external degree between $\emax/2$ and $\emax$ can be arbitrarily classified as introvert or extrovert.
\end{proof}

\begin{definition}[Extrovert/Introvert]
\label{def:extrovert-cliques}
An almost-clique is \emph{extrovert} if it has more than $2\epsilon'\Delta$ extroverted nodes, and
\emph{introvert} otherwise.
\end{definition}

\begin{Algorithm}
\label{alg:SACD}
The algorithm preconditioning almost-cliques.

\textbf{Input:} an $\epsilon'$-almost-clique decomposition $\Vsparse, C_1, \ldots, C_t$ for some $t$.

\begin{enumerate}
\item In each clique $C_i$, let $W_i \subseteq C_i$ be its set of extroverted nodes. Each clique learns if it is introvert or extrovert in $O(\log\Delta)$ rounds by aggregating the size of $W_i$ on a BFS tree. Denote by $J$ the set of indices $i\in[t]$ such that $C_i$ is extrovert.
\item\label[step]{step:SACD-genslack} (Generate Slack) With probability $1/20$, sparse nodes and dense nodes from extroverted cliques $\Vsparse\cup\bigcup_{i\in J}C_i$ independently try a random color.
\item\label[step]{step:SACD-Vprime} Let $V'=\Vsparse\cup\bigcup_{i\notin J} W_i \cup \bigcup_{i\in J} (C_i\setminus W_i)$ be the set containing sparse nodes, extroverted nodes from introverted cliques and introvertednodes from extroverted cliques. All nodes in $V'$ have slack $\Omega(\epsilon'^2\Delta)$ and can be colored in $O(\log\Delta)$ rounds by \multitrial.
\item\label[step]{step:SACD-extroverted} Run randomized color trial for $O(\log\eta)$ rounds in extroverted cliques. The number of uncolored nodes left in each $W_i$ for $i\in J$ is at most $O(\Delta/\eta)$. Complete the coloring of extroverted cliques using \multitrial.
\end{enumerate}
\end{Algorithm}

\begin{proof}[Proof of \cref{thm:strong-ACD}]
After \cref{step:SACD-genslack}, nodes in $\Vsparse$ have $\Omega(\epsilon'^2\Delta)$ permanent slack and extroverted nodes have $\Omega(\emax)=\Omega(\Delta/\eta)$ permanent slack (by \cref{lem:slackgen}). 
Let $J\subseteq [t]$ the set of extroverted cliques. In \cref{step:SACD-Vprime}, we color nodes of $V'$ where
\[ V'=\Vsparse\cup\bigcup_{i\notin J} W_i \cup \bigcup_{i\in J} (C_i\setminus W_i) \ . \]
Dense nodes of $V'$ receive slack from their inactive neighbors in $V\setminus V'$.
\begin{itemize}
\item An extroverted nodes $v\in W_i$ in some introverted clique $C_i$ with $i\notin J$ has $|N(v)\cap (C\setminus W_i)| \ge (1-3\epsilon')\Delta$ introverted neighbors in $C_i$. Note that none of them was colored in \cref{step:SACD-genslack}.
\item A introverted node $v\in C$ in an extroverted clique $C_i$ with $i\in J$ has at least $|N(v)\cap W_i| \ge (2\epsilon' - \epsilon')\Delta = \epsilon'\Delta$ extroverted neighbors in $C_i$. Each such neighbor gets colored in \cref{step:SACD-genslack} with probability at most $1/20$; hence, w.h.p. at least $0.9\epsilon'\Delta$ are uncolored.
\end{itemize}

Adding sparse nodes, all nodes in $V'$ have slack $\Omega(\Delta)$ for a small enough universal constant. Hence, by \cref{lem:multi-trial}, we can color all nodes in $V'$ in $O(\log\Delta)$ rounds and $O(\log n)$ fresh colors with high probability.

After \cref{step:SACD-Vprime}, the only extroverted nodes to remain uncolored are in extroverted cliques. We now explain how we color these nodes. Nodes have $\Omega(\Delta/\eta)$ slack. By an argument similar to \cref{lem:reduce}, w.h.p., we reduce the degree of each node by a constant factor. After $O(\log \eta)$ rounds, each node has uncolored degree $O(\Delta/\eta)$. It samples $O(\eta\log\log n)$ colors.
Nodes now have slack proportional to their degree and can be colored by \multitrial in $O(\log \Delta)$ rounds and using $O(\eta\log n)$ colors.

We now prove that our coloring verifies the properties of \cref{thm:strong-ACD}.
\emph{The crux is that the only uncolored nodes remaining are introverted nodes in introverted almost-cliques}. For each introverted almost-clique $C$ in the $\epsilon'$-almost-clique decomposition, we get an $\epsilon$-almost-clique $C'$ with the claimed properties by simply removing colored nodes. This is because $C'$ is an $\epsilon'$-almost-cliques from which we removed at most $2\epsilon'\Delta$ extroverted nodes. Hence, the upper bound $|C'|\le (1+\epsilon')\Delta\le(1+\epsilon)\Delta$ trivially holds (recall $\epsilon'=\epsilon/3$) and for all $v\in C'$, we have $|N(v)\cap C'| \ge (1-3\epsilon')\Delta = (1-\epsilon)\Delta$. Furthermore, all nodes of $C'$ are introverted, therefore they are connected to at most $\Delta/\eta$ nodes in other cliques. Note however that they can be connected to $\epsilon\Delta$ colored nodes (as they include sparse nodes and extroverted nodes from $C$).
\end{proof}

\subsection{Analysis of RandomPush}
\label{app:randompush}

\randompushlemma*

\begin{proof}
    Consider a particular message, and for each $i \in [O(\log \Delta)]$, let $S_i$ be the set of nodes in the almost-clique that know this message before iteration $i$.
    Let $\ov{S_i} = C \setminus S_i$. 
    
    Initially, $\card{S_1} \geq 1$. 
    Each node has degree $\Theta(\beta^4)$, w.h.p., by \cref{lem:sparse-graph}. 
    Thus, nodes forward the message to $\Omega(\beta)$ of its neighbors, so $\card{S_2} \geq \beta$, w.h.p.
    We now show that $S_i$ grows geometrically while $\card{S_i} \leq 3\Delta/4$, then afterwards $\ov{S_i}$ decreases geometrically.
    
    By \cref{lem:expansion}, while $\card{S_i} \leq 3\Delta/4$, there are at least $\card{S_i}\beta^4/40$ edges between $S_i$ and $\ov{S_i}$.
    For an uninformed node $v$ in $\ov{S_i}$, let $d_v^{S_i} = \card{N_{\sparse{C}}(v) \cap S_i}$ be its number of informed neighbors. Letting $X_v$ be the event that $v$ learns the message in this iteration, we have that
    \[
        \Pr[X_v] = 1 - \parens*{1-\frac{1}{x}}^{d_v^{S_i}}
        \geq 1 - \frac{1}{1+d_v^{S_i}/x} = \frac{d_v^S}{x+d_v^{S_i}}
    \]
    where we used \cref{lem:useful-inequalities}. Hence, the expected number of nodes that learn the message is at least
    \[
    \sum_{v\in \ov{S_i}} \frac{d_v^{S_i}}{x+d_v^{S_i}} \geq \frac{\sparse{\Delta}}{x+\sparse{\Delta}}\cdot \frac{\card{S_i}\beta^4}{8\sparse{\Delta}} \geq \Omega(\card{S_i}) \ ,
    \] 
    since by concavity of the function $f(y) = y/(x+y)$, this sum is minimized when the degrees $d_v^S$ are as unevenly distributed as possible, with $\frac{\card{\ov{S_i}}\beta^4/40}{\sparse{\Delta}}$ nodes satisfying $d_v^S = \sparse{\Delta}$ and the rest satisfying $d_v^S = 0$.

    Since the $X_v$ are independent, by \cref{lem:chernoff} (Chernoff bound) it holds w.h.p.\ that $\card{S_{i+1}} \geq (1+\Omega(1))\card{S_{i}}$ while $\card{S_i} \in [\beta, 3\Delta/4]$. Therefore, after $i \in \Theta(\log \Delta)$ iterations, $\card{S_i} \leq 3\Delta/4$.

    The rest of the argument is similar. The nodes in $\ov{S_i}$ have at least $\card{S_i}\beta^4/40$ edges with $S_i$. In expectation, $\Theta(\card{\ov{S_i}})$ of them get colored in each iteration in expectation, and this holds w.h.p.\ while $\card{\ov{S_i}} \geq \beta$. When $\card{\ov{S_i}}$ drops below $O(\beta^3)$, each node in $\ov{S_i}$ is adjacent to $\Omega(\beta^4)$ nodes in $S_i$. Therefore, it receives the message $\Omega(\beta)$ times in expectation, and thus receives it w.h.p.
\end{proof}

\section{Corollaries for Other Models}

\subsection{Coloring in Distributed Streaming}
\label{sec:local-stream}

\begin{definition}[\textbf{Local Streaming Model}]
In the {\distream} model, there are $n$ nodes with unique $O(\log n)$-bit identifiers and $p(n) = \poly(\log n)$ bits of local space. The nodes have no initial information but have a limited source of randomness. There are two phases: a streaming phase and a communication phase.

\begin{itemize}
\item \textbf{(Streaming Phase)}
Nodes receive their incident edges in the graph $G$ as a stream. Attached to each edge are (some of the) random variables of the incident vertices. I.e., each node $v$ receives a sequence $(v,u_i,s_i)_i$, where $s_i^j$ is the random bits of neighbor $u_i$ in iteration $j$
\footnote{An alternative would be to supply the nodes with shared randomness. Then the ID of the other node would suffice to learn its random bits.}.

\item \textbf{(Communication Phase)}
The nodes communicate in synchronous rounds with their neighbors with $O(\log n)$ bit messages (as in the \CONGEST model). They can only send a message to a neighbor whose ID they have stored, and \emph{we additionally limit them to send/receive $\poly\log n$ messages per rounds}.
\end{itemize}
At the end of the computation, each node outputs its color, which together should form a valid $\Delta+1$-coloring. The objective is to minimize the total number of communication rounds.
\end{definition}

\begin{corollary}
There exists a \distream algorithm using $O(\log^4 n)$ memory per node and $O(\log^2\Delta)$ rounds of communication.
\end{corollary}

\subsection{Coloring in the Cluster Graph Model}
\label{sec:clusterGraphs}
We first define the cluster graph model (a variant appears in \cite{ghaffari2015flow} and similar concepts appear in other places in the literature, see e.g., \cite{rozhovn2022undirected, ghaffari2016distributed,ghaffari2013cut,ghaffari2022universalCut}). Then, we state our result. 

\begin{definition}[\textbf{Cluster graph model}] Consider a cluster graph defined as follows: Given a graph $G=(V, E)$, suppose that the nodes have been partitioned into vertex-disjoint clusters. Definite the cluster graph as an abstract graph with one node for each cluster, where two clusters are adjacent if they include two nodes that are neighboring each other in $G$. Furthermore, for each cluster, we are given a cluster center and cluster tree that spans from the cluster center to all nodes of the cluster. One round of communication on the cluster graph involves the following three operations: 
\begin{itemize}
    \item (\textbf{Intra-cluster broadcast}) Each cluster center starts with a $\poly(\log n)$-bit message and this message is delivered to the nodes in its cluster. 
    \item (\textbf{Inter-cluster communication}) For each edge $e=\{v, u\}$ for which $v$ and $u$ are in two different clusters, node $v$ can send a $\poly(\log n)$-bit message and this message is delivered to $u$, simultaneously for all such inter-cluster edges.  
    \item (\textbf{Intra-cluster convergecast}) Each node can start with a $\poly(\log n)$-bit message and, in each cluster, we deliver a $\poly(\log n)$-bit aggregate of the messages of the cluster's nodes to the cluster center. The aggregate function can be computing the minimum, maximum, summation, or even gathering all messages if there are at most $\poly(\log n)$ many. These suffice for our application. More generally, this intra-cluster convergecast operation can be any problem that can be computed in $O(h)$ rounds of the \congest model communication on a given tree of depth $h$ and using $\poly(\log n)$-bit messages.
\end{itemize}
\end{definition}
\begin{theorem} There is a distributed randomized algorithm that computes a $\Delta+1$-coloring in $\poly(\log n)$ rounds of the cluster graphs model.
\end{theorem}

\begin{proof}[Proof Sketch] The proof follows essentially directly from our distributed palette sparsification theorem, stated in \Cref{thm:DPS}. We just need to discuss how the cluster graph computes and simulates the corresponding sparsified graph.

Each cluster center samples the $\poly(\log n)$ colors of its node in the palette sparsification theorem. Then, via intra-cluster broadcast, the cluster center delivers these colors to all nodes of its cluster. Afterward, via inter-cluster communication, each node sends the colors of its cluster to all neighboring nodes in other clusters. Each node $v$ in a cluster $\mathcal{C}$ that notices a neighboring cluster $\mathcal{C}'$ that sampled a common color remembers the cluster identifier of $\mathcal{C}'$, as a neighboring cluster in the sparsified variant of the cluster graph. We then perform one intra-cluster convergecast, where each node starts with the neighboring clusters that it remembered as neighboring clusters in the sparsified graph, and we gather all of these neighboring cluster identifiers to the cluster center. Since each cluster has $\poly(\log n)$ neighboring clusters after the sparsificaiton, this can be done as a $\poly(\log n)$-bit aggregation. 

In the course of this process, we could also elect for each pair of neighboring clusters $\mathcal{C}$ and $\mathcal{C}'$ in this sparsified graph one physical edge from node $v \in \mathcal{C}$ to a node $u\in \mathcal{C}'$. For instance, that can be the edge $(v, u)$ with the highest ID tuple. Again, this fits easily as a $\poly(\log n)$-bit aggregation.  

At this point, each cluster center knows all its $\poly(\log n)$ neighboring clusters and has identified a physical edge 
connected to each neighboring cluster. Hence, the cluster graph model can simulate one round of the \congest model communication on the sparsified graph. Therefore, to compute a $\Delta+1$ coloring of the cluster graph, it suffices to invoke \Cref{thm:DPS}.
\end{proof}

\subsection{Coloring in the Node Capacitated Clique}
\label{sec:NCC}

We show, in fact, that any 1-pass \distream algorithm with $\poly\log n$ memory and bandwidth can be turned into a $\poly\log n$ rounds \NCC algorithm.

\begin{theorem}
\label{thm:NCC-simulation}
Let $\cA$ be a randomized \distream algorithm using \emph{one streaming pass} and $T$-communication rounds. If it has bandwidth $B$ and communication with at most $D$ different neighbors within a round, then there is an algorithm emulating $\cA$ with high probability in the \NCC model in $O(\log n + \frac{T\cdot B D}{\log n})$ communication rounds.
\end{theorem}

Consider an arbitrary communication round of \distream. In the worst case, a node must send $B$ bits to $D$ nodes. Since in the \NCC model, a node can only communicate $O(\log n)$ bits to $O(\log n)$ nodes in $G$ within a round, it can emulate one communication round of \distream in $O(BD/\log n)$ rounds. Note that this upper bound can be improved in some specific cases, e.g., if the algorithm only broadcast messages, but we ignore such optimizations here. This gives the following claim:

\begin{claim}
\label{claim:NCC-simulation}
If nodes know the $\poly\log n$ random bits of their neighbors, then emulating $\cA$ requires $O\parens*{\frac{T\cdot (BD)}{\log n}}$.
\end{claim}

The only information missing to nodes in order to run the \distream algorithm is the initial state of their neighbors. As nodes have no memory restriction in \NCC, we are free to use pseudo-random initial states. For a function $S$ mapping nodes to $\poly\log n$ bits binary strings, we write $\cA[S]$ for the algorithm $\cA$ where each node $u$ has the string $S(u)$ as random bits.

\begin{lemma}
\label{lem:pseudorandom-init}
For any fixed $n$, there is a family $\cS$ of $\poly(n)$ functions mapping nodes to initial states such that for any $n$-node graph input, if we run $\cA$ on a random function in $\cS$, the streaming is correct with high probability.
\end{lemma}

\begin{proof}[Proof of \cref{lem:pseudorandom-init}]
Fix a $n$-node graph $G$. 
Sample $t$ functions $S_1, \ldots, S_t$ assigning $\poly\log n$-bits binary string to nodes.

Since $\cA$ is correct with high probability, it means that on a random $S_i$, algorithm $\cA[S_i]$ fails with probability at most $1/n$. For a fixed $G$, call $X_i$ the random variable equal to one iff algorithm $\cA[S_i]$ fails on $G$ with probability more than $1/n$. In expectation, the number of bad assignments is $\Exp\range*{\sum_{i\in [t]} X_u} \le t/n$. Samples are independent; hence, by Chernoff, we get
\begin{equation}\label{eq:NCC}
\Pr\range*{\sum_{i\in [t]} X_i > \frac{2t}{n}} \le \exp\parens*{-\frac{2t}{3n}}.
\end{equation}

We conclude the proof by using the union bound on all $n$ nodes graphs. There are at most $2^{n^2}$ input graphs $G$ on $n$ nodes. Therefore, for some large enough $t=\Omega(n^3)$, the bound in \cref{eq:NCC} is strictly less than 1; hence, there is a family $\cS=\set{S_1, \ldots, S_t}$ such that the probability that $\cA[S_i]$ fails for a random $i\in[t]$ is at most $2/n$ for all $n$-nodes graphs.
\end{proof}

\begin{proof}[Proof of \cref{thm:NCC-simulation}]
For a fixed $n$-sized network. Nodes can locally compute the family $\cS$ described in \cref{lem:pseudorandom-init} (recall there is not memory or local time constraints on nodes in the \NCC model). The node of minimum ID then sample a random index $i\in \range*{|\cS|}$ and broadcast it. Since $|\cS| \le \poly(n)$, index $i$ can be described in $O(\log n)$ bits. Broadcasting a message to every one takes $O(\log n)$ rounds.

Nodes then know the randomness of every node in $G$ as well as their adjacency list. They can therefore run the streaming phase without any communication. Once this is done, they can emulate $\cA$ in $O\parens*{TBD/\log n}$ rounds. By \cref{lem:pseudorandom-init}, it fails with probability $1/\poly(n)$.
\end{proof}

\end{document}